\documentclass[11pt,letterpaper]{article}
\usepackage[margin=1in]{geometry}
\usepackage{datetime}
\usepackage{libertine}
\usepackage{enumitem}
\usepackage{wrapfig}
\usepackage{subcaption}
\usepackage[T1]{fontenc}
\usepackage{pgfplots}
\pgfplotsset{compat=1.15}
\usepackage{amsmath, amssymb}
\usepackage{amsthm}
\usepackage{thmtools,thm-restate} %For restating the lemmas from main part/appendix in the introduciton.\usepackage{thmtools,thm-restate} %For restating the lemmas from main part/appendix in the introduciton.
\usepackage{color}
\usepackage{cases}
\usepackage{xparse}
\usepackage{xargs}
\usepackage{appendix}
\usepackage{verbatim, xspace}
\usepackage{tikz}
\usepackage{mathtools}
\usepackage[font=footnotesize]{caption}
\usepackage{csquotes}

\usepackage{framed}
\usepackage[framemethod=TikZ]{mdframed}

\usepackage[linesnumbered,lined,boxed,commentsnumbered,noend]{algorithm2e}

\usepackage[hidelinks]{hyperref}
\usepackage{pgfplots}

\usepackage{xcolor}

\newcommand{\citet}[1]{\cite{#1}}

\usepackage[numbers]{natbib}
\ifluatex
  \usepackage{fontspec}
  \defaultfontfeatures{Ligatures=TeX}
  \PassOptionsToPackage{additionalMargin=1.5cm}{luatodonotes}
  \PassOptionsToPackage{colorinlistoftodos}{luatodonotes}
 \usepackage[disable,po,splitting=weightedMedian]{luatodonotes} % (FOR LUALATEX) TODO COMMENTS DISABLED
\else
  \usepackage[utf8]{inputenc}
 \usepackage[disable,colorinlistoftodos,prependcaption,textsize=tiny]{todonotes} % TODO COMMENTS DISABLED
\fi

\newcommandx{\unsure}[2][1=]{\todo[linecolor=green,backgroundcolor=green!25,bordercolor=green,#1]{\normalsize #2}}
\newcommandx{\improvement}[2][1=]{\todo[inline,linecolor=blue,backgroundcolor=blue!05,bordercolor=blue,#1]{\normalsize #2}}
\newcommandx{\infodone}[2][1=]{}%\todo[linecolor=yellow,backgroundcolor=yellow!30,bordercolor=yellow,#1]{{\tiny (Info about what's done)\\}#2}}
\newcommandx{\note}[2][1=]{\todo[linecolor=yellow,backgroundcolor=yellow!8,bordercolor=yellow,#1]{{\tiny (Just a note for us)\\}#2}}
\newcommandx{\floatmodel}[2][1=]{\todo[inline,linecolor=red,backgroundcolor=yellow!25,bordercolor=yellow,#1]{#2}}
\newcommandx{\thiswillnotshow}[2][1=]{\todo[disable,#1]{#2}}
\newcommandx{\karol}[2][1=]{\todo[inline,linecolor=blue,backgroundcolor=blue!25,bordercolor=blue,caption={\normalsize \textbf{Karol}},#1]{\normalsize #2}}
\newcommandx{\importanttodo}[2][1=]{\todo[linecolor=red,backgroundcolor=red!25,bordercolor=red,#1]{{\tiny (Important Todo!)\\}#2}}
\newcommandx{\importantquestion}[2][1=]{\todo[linecolor=red,backgroundcolor=red!50,bordercolor=red,#1]{{\tiny \textcolor{white}{(Important Question!)}\\}#2}}
\newcommandx{\resolvedtodoNotDoneYet}[3][1=]{\todo[linecolor=black,backgroundcolor=gray,bordercolor=black,#1]{\tiny #2\\\textbf{Answer:}#3}}
\newcommandx{\conflict}[1]{\todo[linecolor=red,backgroundcolor=red!55,bordercolor=red,inline]{{\tiny \textbf{(GIT CONFLICT/MERGE)}\\}\normalsize #1}}
\newcommand{\reviewertodo}[1]{\todo[linecolor=blue,backgroundcolor=blue!50,bordercolor=blue,inline]{{\tiny \textcolor{red}{(Reviewer's Comment!)}\\}#1}}
\newcommand{\kt}[1]{\todo{KF: #1}}
\newcommand{\ktp}[1]{\importanttodo{KF: #1}}
\newcommand{\kint}[1]{\todo[inline]{KF: #1}}
\newcommand{\kintp}[1]{\importanttodo[inline]{KF: #1}}

\newcommand{\soda}[1]{\todo[inline,linecolor=lime,bordercolor=lime,backgroundcolor=lime!30]{\scriptsize \textbf{SODA Reviewer:}\\\normalsize #1}}

\newcommand{\cpy}{\mathrm{cpy}}

\newtheorem{theorem}{Theorem}[section]
\newtheorem{definition}[theorem]{Definition}
\newtheorem{lemma}[theorem]{Lemma}

\newtheorem{claim}[theorem]{Claim}

\newtheorem*{remark}{Remark}

\newcommand{\thistheoremname}{}
\newtheorem*{genericthm}{\thistheoremname}

\newcommand{\opt}{\mathrm{OPT}}

\newcommand{\eps}{\varepsilon}

\newcommand{\Oh}{\mathcal{O}}

\newcommand{\nat}{\mathbb{N}}

\newcommand{\Tt}{\mathcal{T}}
\newcommand{\Pp}{\mathcal{P}}

\newcommand{\Ss}{\mathcal{S}}
\newcommand{\Cc}{\mathcal{C}}

\newcommand{\Reals}{\mathbb{R}}

\newcommand{\poly}{\mathrm{poly}}
\newcommand{\polylog}{\,\textup{polylog}}

\newcommand{\red}{\mathtt{R}}
\newcommand{\blue}{\mathtt{B}}
\newcommand{\is}{\mathrm{I}}
\newcommand{\segm}{\ensuremath{\mathrm{S}}}
\newcommand{\segmalt}{\ensuremath{\mathrm{R}}}
\newcommand{\tour}[1]{\ensuremath{\pi_{#1}}}
\newcommand{\touralt}[1]{\ensuremath{\pi'_{#1}}}
\newcommand{\touralthat}[1]{\ensuremath{\hat{\pi}'_{#1}}}
\newcommand{\tourim}[1]{\ensuremath{\overline{\pi_{#1}}}}
\newcommand{\col}[1]{\ensuremath{c(#1)}}
\newcommand{\rest}[3]{\ensuremath{{#1}[{#2},{#3}]}}

\newcommand{\costapprox}{$(1+\eps)$-cost-approximation} % Used to refer to our number X from Preliminaries

\newcommand{\wt}{\mathrm{wt}}

\newcommand{\dist}{\mathrm{dist}}
\newcommand{\pro}{\mathrm{pro}}

\newcommand{\grid}{\mathrm{grid}}

\newcommand{\bd}{\partial}

\renewcommand{\leq}{\leqslant}

\renewcommand{\le}{\leqslant}
\renewcommand{\ge}{\geqslant}

\newcommand\BST{{\sc Bicolored Noncrossing Spanning Trees}\xspace}
\newcommand{\BTSP}{{\sc Bicolored Noncrossing Traveling Salesman Tours}\xspace}
\newcommand\MSP{{\sc Multicolored Noncrossing Paths}\xspace}
\newcommand{\RBGS}{{\sc Red-Blue-Green Separation}\xspace}
\newcommand{\kCS}{{\sc $k$-Colored Points Polygonal Separation}\xspace}

\newcounter{openquestion}

\newenvironment{insight}
{\mdfsetup{%
    nobreak=true,
	middlelinecolor=gray,
	middlelinewidth=1pt,
	backgroundcolor=gray!10,
    innertopmargin=7pt,
	roundcorner=5pt}
\begin{mdframed}}
{\end{mdframed}}

\newcommand{\defproblem}[3]{
	\vspace{2mm}
	\vspace{1mm}
	\noindent\fbox{
		\begin{minipage}{0.95\textwidth}
			#1 \\
			{\bf{Input:}} #2  \\
			{\bf{Task:}} #3
		\end{minipage}
	}
	\vspace{2mm}
}

\title{Gap-ETH-Tight Approximation Schemes for Red-Green-Blue Separation and Bicolored Noncrossing Euclidean Travelling Salesman Tours}
\date{}

\author{
    François Dross\footnote{Univ.~Bordeaux, CNRS, Bordeaux INP, LaBRI, UMR 5800, F-33400 Talence, France, \texttt{francois.dross@u-bordeaux.fr}
	This work is part of the project CUTACOMBS that
    received funding from the European Research Council (ERC) under the European Unions Horizon
    2020 research and innovation programme (grant agreement No. 714704).}
    \and
    Krzysztof Fleszar\footnote{Institute of Informatics, University of Warsaw,
    	Poland, \texttt{kfleszar@mimuw.edu.pl}.
    	This work is part of the project TUgbOAT that has
    	received funding from the European Research Council (ERC) under the European Unions Horizon
    	2020 research and innovation programme (grant agreement No. 772346).}
    \and
    Karol W\k{e}grzycki\footnote{Saarland University and Max Planck Institute for Informatics,
        Saarbr\"ucken, Germany, \texttt{wegrzycki@cs.uni-saarland.de}. 
    This work is part of the project TIPEA that has
    received funding from the European Research Council (ERC) under the European Unions Horizon
    2020 research and innovation programme (grant agreement No. 850979).
    }
    \and
    Anna Zych-Pawlewicz\footnote{Institute of Informatics, University of Warsaw,
        Poland, \texttt{anka@mimuw.edu.pl}.
    This work is part of the project CUTACOMBS that has
    received funding from the European Research Council (ERC) under the European Unions Horizon
    2020 research and innovation programme (grant agreement No. 714704).
}
}

\begin{document}

\hypersetup{pageanchor=false}
\begin{titlepage}
\maketitle
\thispagestyle{empty}

\begin{abstract}
    In this paper, we study problems of connecting classes of points via noncrossing structures. Given a set of colored terminal points, we want to find a graph for each color that connects all terminals of its color with the restriction that no two graphs cross each other. We consider these problems both on the Euclidean plane and in planar graphs.

On the algorithmic side, we give a Gap-ETH-tight EPTAS for the two-colored traveling salesman problem as well as for the red-blue-green separation problem (in which we want to separate terminals of three colors with two noncrossing polygons of minimum length), both on the Euclidean plane. This improves the work of Arora and Chang (ICALP 2003) who gave a slower PTAS for the simpler red-blue separation problem.  For the case of unweighted plane graphs, we also show a PTAS for the two-colored traveling salesman problem.  All these results are based on our new patching procedure that might be of independent interest.

On the negative side, we show that the problem of connecting terminal pairs with noncrossing paths is NP-hard on the Euclidean plane, and that the problem of finding two noncrossing spanning trees is NP-hard in plane graphs.

\end{abstract}

\begin{picture}(0,0)
\put(-70,-270)
{\hbox{\includegraphics[width=40px]{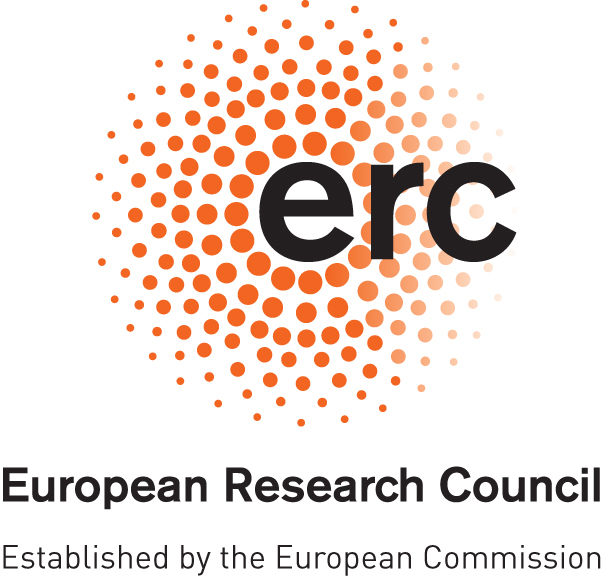}}}
\put(-80,-330)
{\hbox{\includegraphics[width=60px]{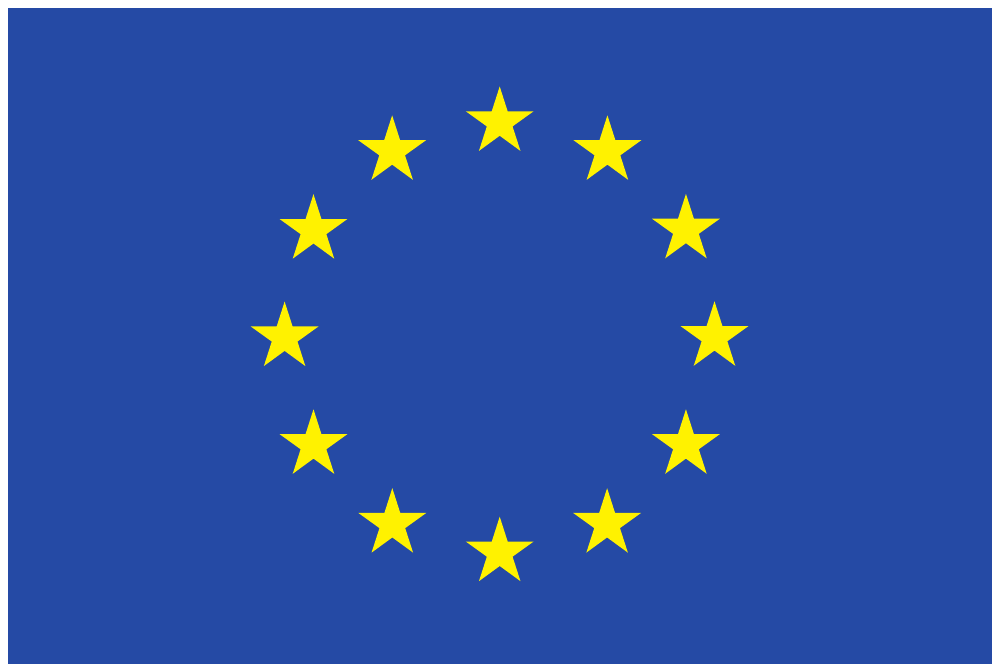}}}
\end{picture}

\end{titlepage}

\hypersetup{pageanchor=true}

\reviewertodo{variables S and R are sometimes slanted and sometimes not}
\reviewertodo{paragraph headings sometimes use title case and sometimes not}
\reviewertodo{X and Y as the axes of the plane are inconsistently capitalized}
\reviewertodo{add figures regarding allowed/disallowed solutions in plane and especially the graph setting.}
\reviewertodo{if quadtree cells can have no points, than update fig 5 accordingly; otherwise make it clear that they always have points.}
\reviewertodo{Make notation consistent! ("That description is hard to verify, due to changing notation.")}
\reviewertodo{define first(!) variables before using them. Bad/good example(?): "if a graph of n vertices has at most two terminals per color, and terminals are adjacent to h faces in total, then the problem can be solved in $2^{O(h^2)}n \log k$ time. [13]"\\ 309: "context first" is a generalization of “define before using":  “For a simple curve $\pi\in\mathbb{R}^2$ that contains two points $y$,$z$, we define…  (Otherwise the reader thinks of two points, then afterward learns they were supposed to be “on $\pi$”.)}
\reviewertodo{Our "technical tool" really only for 1 of 4 problems needed?? (... keeping the technical description of one tool ... that is needed for one of the four problems.)}
\reviewertodo{Using $\mathcal{S}$, $\ell$, and $S$ in close proximity makes this a little hard for your readers as well, by the way.}
\section{Introduction}
\kintp{I would emphasize very much that our results are Gap-ETH-tight! Even in abstract.}
\todo{KF: Make problem names consistent. Two or two-colored or bicolored?}
Imagine that you are given a set of cities on a map belonging to different communities. 
Your goal is to separate the communities from each other with fences such that each community remains connected. 
The total length of the fences should be minimized. 

 \note[inline]{KF: (For the future) One could study a variant of the problem where we do not want to minimize the fence length but where we want that the areas of the communities have the same size (where e.g. we always have a fixed third cycle that is the convex hull).}

\begin{figure}[ht!]
    \centering
    \includegraphics[width=0.6\linewidth]{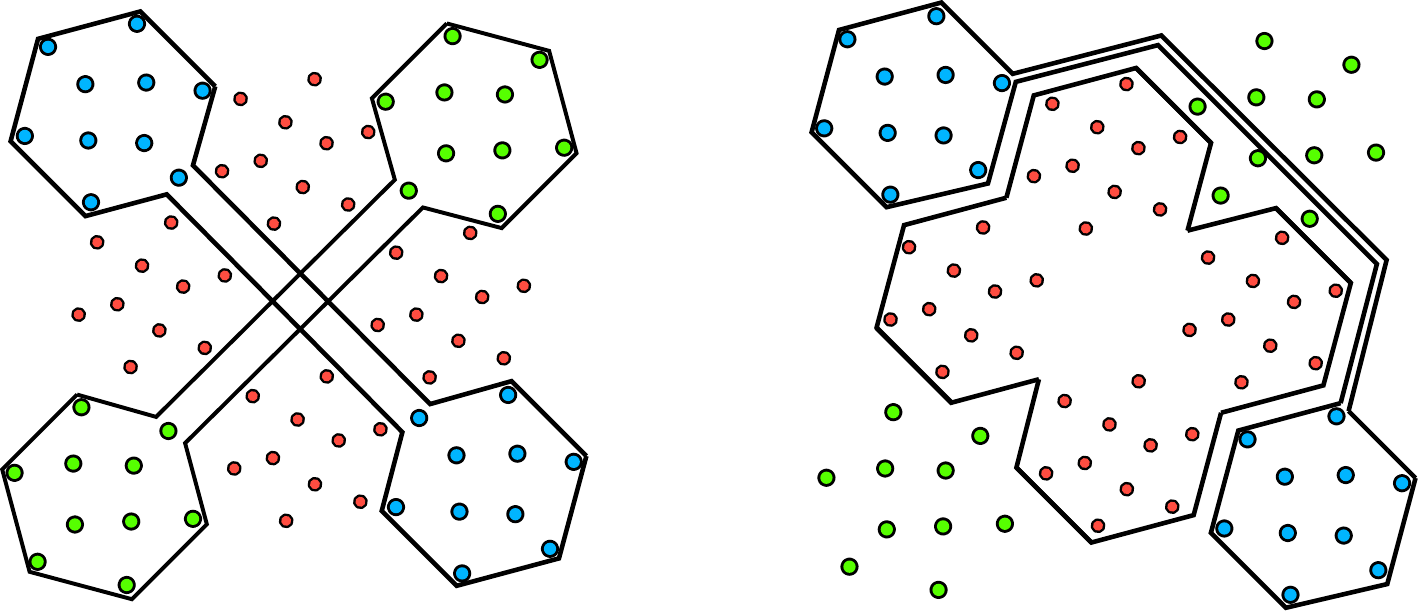}
    \caption{Points on the plane are labeled into three groups (red, blue and
        green) and the task is to find two polygons of total minimum length that separate
        the three classes. The solution on the left is not feasible because two
    polygons are crossing. The solution on the right is feasible because
polygons are noncrossing and points in different classes are separated.}
    \label{fig:rgb}
\end{figure}

In this paper, we study the computational complexity of a very general class of problems that can model such and similar 
questions and we provide efficient algorithms to solve them. 
Given~$n$ colored terminal
points (on the Euclidean plane or in a plane graph), the goal is to find a set of pairwise noncrossing geometric graphs of minimum total length that satisfy some requirement for each color.
For instance, in the problem mentioned above, we look for noncrossing cycles that separate terminals of different colors.

The study of such noncrossing problems is not only motivated by the fact that these problems are natural generalizations of well-examined fundamental problems (e.g. the spanning tree problem and the traveling salesman problem) that lead to the development of new techniques\kt{; e.g. \enquote{portal reduction}~\cite{isaac15}\\But not sure if this really counts as a (new) technique}.
Such problems are also motivated by their rich application in different fields, e.g., to VLSI design~\cite{takahashi1992planeGraphsPaths,takahashi1993rectPaths,KMN01,EN11} and set visualization of spatial data~\cite{ARRC11, HKKLS, EHKP15, HKKLS, JGAA-499, reinbacher2008delineating}. For instance, in a method for visualizing embedded and clustered graphs~\cite{EHKP15}, 
clusters are visualized by non-overlapping regions that are constructed on top of noncrossing trees connecting the vertices of each cluster.
%To optimize the visual appearance and readability, the total length of the underlying trees is minimized. 
In the following, we introduce the problems and discuss our contribution and related work. 
At the end, we give an overview of our paper.

\paragraph*{Red-Blue-Green Separation} 
In the \RBGS{} problem, 
we are given a set of
$n$ terminal points on the (Euclidean) \todo{KF: added Euclidean in parenthesis. Later we might drop the word Euclidean i think. In the preliminaries, we can say that we always mean the Euclidean plane.} plane. Each point is assigned to one of three colors (red,
blue, or green). The goal is to find two noncrossing Jordan curves that separate
terminals of different classes (i.e., every possible path connecting two terminals of different color must cross at least one of the curves) of minimum total length; see
Fig.~\ref{fig:rgb} for an illustration.

\importanttodo{Add a remark that we look at polygons and not lines}

Various forms of this problems have been studied in connection to computer
vision and collision
avoidance~\cite{np-red-blue,toussaint1986optimal,segmentation,retsinas2016imageSegmentation}, geographic
information retrieval~\cite{reinbacher2008delineating} or even finding pandemic
mitigation strategies~\cite{covid}.
To the best of our knowledge, the problem has been studied only for two colors which is already NP-hard~\cite{np-red-blue}.
Mata and Mitchel~\cite{mata-michell} obtained an $\Oh(\log n)$-approximation scheme which was subsequently strengthened by an~$\Oh(\log m)$~approximation algorithm by Gudmundsson and Levcopoulos~\cite{gudmundsson99RedBlue} where~$m$ is an output-sensitive parameter bounded by~$n$.
Finally, Arora and Chang~\cite{red-blue} proposed an $n (\log n)^{\Oh(1/\eps)}$-time algorithm that returns a $(1+\eps)$-approximation for any $\eps > 0$. 
The aforementioned approximation schemes crucially rely on
\emph{patching schemes} 
\infodone{KF: added "and other techniques" as I don't know whether there is a patching scheme in the algorithm of Gudmundsson.} 
and other techniques that seem to work only for two colors, where we look only for a single curve. %---in contrast to the case of more colors where we look for pairwise noncrossing curves.
Due to the many new technical challenges  
\todo{KF: I think we should discuss this issue somewhere (and say here \enquote{see section X for a discussion.}) in order to sell our paper better.} 
arising from the noncrossing constraint,
it is not obvious how
to generalize these results to more colors, where we want to find pairwise noncrossing curves. %to more classes of terminals 

Our contribution is twofold.
First, we succeed in generalizing the result of Arora and
Chang~\cite{red-blue} to three colors by designing, among others, a new patching procedure of independent interest for noncrossing curves. 
Second, we improve the running time of Arora and Chang~\cite{red-blue} to\ktp{near-optimal?}~$2^{\Oh(1/\eps)}n \polylog(n)$ and thus obtain an EPTAS for three colors.
%with~$2^{\Oh(1/\eps)}n \polylog(n)$ running time.

\begin{restatable}[Red-Blue-Green Separation]{theorem}{thmPTASrbg}
    \label{thm:rbg-sep}
    Euclidean Noncrossing \RBGS{} admits a randomized~$(1+\eps)$-approximation scheme \todo{KF: Shouldn't the word noncrossing be a part of the problem's name?}
    with $2^{\Oh(1/\eps)} n \polylog(n)$ running time.
%    Moreover, the problem does not admit a~$(1+\eps)$-approximation scheme 
%        with $2^{o(1/\eps)} n^{100}$ running time, unless Gap-ETH fails.
%    \kint{Mention derandomization}
\kintp{near-optimal?}
\end{restatable}

\todo{KF: Problem names consistent in this paragraph with rest of paper?}
We note that our algorithm is near-optimal. Namely, Eades and
Rappaport~\cite{np-red-blue} show a reduction of the red-blue separation problem to
Euclidean TSP with an arbitrary small gap. This fact combined with a recent lower bound
on Euclidean TSP~\cite{focs21} means that an $2^{o(1/\eps)} \poly(n)$-time
approximation for Euclidean Noncrossing \RBGS would contradict Gap-ETH.

\paragraph*{Noncrossing Tours}
In the \BTSP{} problem, we are given $n$ terminal points on the plane. Each point is
colored either red or blue.
The task is to find two noncrossing
round-trip tours of minimum total length such that all red points are visited by
one tour and all blue points are visited by the other one. The problem is
a generalization of the classical Euclidean traveling salesman
problem~\cite{Papadimitriou77} and is therefore NP-hard.  
To the best
of our knowledge, this problem has not been considered before in the literature, although it nicely fits into the recent trend of finding noncrossing geometric
structures (see, e.g., the works of Polishchuk and Mitchell~\cite{noncrossing-paths-mitchell}, Bereg et al.~\cite{isaac15} and Kostitsyna et al.~\cite{kostitsyna2017SteinerArborescences}). 
By using the same patching procedure as for Theorem~\ref{thm:rbg-sep}, we obtain an EPTAS for this problem.

%\begin{restatable}{theorem}{thmPathsSpecialZeroWeights}\label{thm:paths-special-zero-weights}
%	The problem of maximizing the shortest path via robust edge weights, when all edges have weight~$0$, can be solved in~${O(m \cdot \min\{n^{2/3},\sqrt{m}\})}$ time where~${x^*_0}$ is the optimal value for~$x_0$.
%\end{restatable}

\kt{Say \emph{Euclidean problem name} or \emph{problem name on the Euclidean plane}?}
\begin{restatable}{theorem}{thmPTAStsp}\label{thm:eptas-noncrossingtsp}
    Euclidean \BTSP{} admits a randomized~$(1+\eps)$-approximation scheme with $2^{\Oh(1/\eps)} n \polylog(n)$ running time. 
    \kintp{near-optimal?}
%    \kint{Mention derandomization}
\end{restatable}
%\begin{theorem}
%	\todo[inline]{KF: Make the cited theorem be exact copies of each other including numbering (e.g. via timetravel package).}
%    Euclidean \BTSP{} admits an approximation scheme with $2^{\Oh(1/\eps)} n \polylog(n)$ running time. 
%    \reviewertodo{"Theorems 1,2, \& 3 should be made complete by saying how the approximation depends on epsilon, not only the running time."}
%\end{theorem}

Note that the aforementioned recent lower bound on Euclidean TSP~\cite{focs21} implies that our result for \BTSP{} is Gap-ETH-tight: there is no~$2^{o(1/\eps)} \poly(n)$-time~$(1+\eps)$-approximation scheme unless Gap-ETH fails.

We also consider the \BTSP problem in plane graphs (that is, in planar graphs with a given planar embedding). There, the task is to draw the tours 
within a \enquote{thick} drawing of the given embedding without any crossings; see
Section~\ref{sec:prelim} for a formal definition. 
We show that arguments for this problem in the
geometric setting seamlessly transfer to the combinatorial setting of planar graphs. 
\todo{KF: Then why don't we obtain an EPTAS for the planar graphs setting?}
\todo{KF: We should say somewhere that we were not able to transfer the arguments of our other(!) Euclidean problems to the graph setting.}
This fact allows us to obtain a PTAS also in plane graphs.

\begin{restatable}{theorem}{thmPTAStsp}
%\label{thm:ptas:tsp}
\BTSP in plane unweighted graphs admits an~$(1+\eps)$-approximation scheme with $f(\eps) n^{\Oh(1/\eps)}$ running time for some function~$f$.
\end{restatable}

\kt{I would create an extra paragraph somewhere in the beginning where we discuss our patching lemma and include the text here below.}
At this point, we want to remark on the similarities and differences to the most related paper to our work: 
Bereg et al.~\cite{isaac15} 
consider the problem of connecting same-colored terminals with noncrossing Steiner trees
of total minimum length. 
Apart from a~$\min(k(1+\eps), \sqrt{n}\log k)$-approximation algorithm for~$k$ colors (based on ideas of previous papers~\cite{LMMPS95, CHKL13, EHKP15})
and a~$(5/3+\eps)$-approximation algorithm for three colors, they 
obtain also a PTAS for the case of two colors.
Similarly, as we do in our approximation schemes for \RBGS{} and \BTSP{}, they use a plane dissection technique of Arora~\cite{Arora98} to give their algorithmic
result\todo{KF: What do we mean by this (algorithmic result)? The PTAS?}. 
Moreover, they also design a patching procedure that allows them to
limit the number of times that the trees cross a boundary cell. Given the
nature of Steiner tree problems, their patching procedure can place additional Steiner points in
portals, which is not possible in our case of tours. As a consequence, their patching procedure is substantially simpler and
less surprising. 

\paragraph*{Noncrossing Spanning Trees}
We initiate the study of \BST{} in plane graphs. %and traveling salesman tours in plane graphs. 
%We argue that the setting of plane graphs is a natural direction and we present some surprising observations\todo{KF: Do we have more than one surprising observation?}.
The problem of finding a minimum spanning tree of a graph is well known to
be solvable in polynomial time. 
Here, we study a natural generalization into two
colors. In the \BST{} problem, \emph{every} vertex is colored either red or blue (and called terminal). 
The task is
to find two \emph{noncrossing} trees, that we call \emph{spanning trees}, of minimum total length such that the first
tree visits all red vertices and the second one visits all blue vertices. 
Two trees are noncrossing if they can be drawn in a \enquote{thick} drawing of the given embedding without crossing (see
Section~\ref{sec:prelim} for a formal definition). 

\kintp{Define $n$ somewhere!}
\begin{figure}[htpb]
    \centering
    \includegraphics[width=0.4\linewidth]{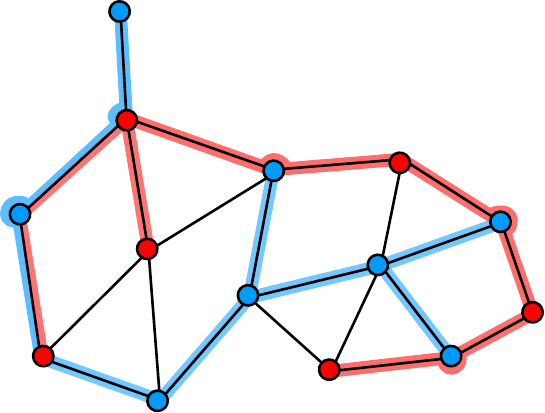}
    \caption{Illustration of \BST{} problem. Note that we can use vertices of
    \todo[inline]{KF: We have to reference this figure!}
    \todo[inline]{KF: We can think of whether drawing really a thick graph and within the thick graph our trees.}
        opposite colors as a connector. Moreover, an edge may be used by both
        trees multiple times. We require, however, that the drawings of the spanning trees
        do not intersect.}%
    \label{fig:}
\end{figure}

Contrasting the result for one color, we show
that our two-colored version of the spanning tree problem is NP-hard in plane graphs.

\begin{restatable}{theorem}{thmNPbst}
\BST in plane \kintp{unweighted?} graphs is NP-hard.\kt{Or even NP-complete?}
\end{restatable}

 We complement our hardness result with a PTAS.
 
\begin{restatable}{theorem}{thmPTASbst}
%\label{thm:ptas:bst}
\BST in plane unweighted graphs admits an~$(1+\eps)$-approximation scheme with $f(\eps) n^{\Oh(1/\eps)}$ running time for some function~$f$.
\end{restatable}

We note that the trees are allowed to visit vertices of the other color. Thus,
they can be also viewed as Steiner trees with the constraint that every Steiner
point is a colored vertex. However, because of this constraint, \BST is not a
generalization of the Steiner tree problem in planar graphs and therefore its
hardness status is not directly determined by that problem.

In the literature, only restricted variants of \BST have been considered
(without a bound on the number~$k$ of colors, though): if there are at most two
terminals per color, the problem can be solved in~${2^{\Oh(h^2)}n \log k}$ time
where~$h$ is the number of face boundaries containing all terminals, and~$n$ is
the number of vertices~\cite{EN11}.  If there are only constantly many terminals
per color and all terminals lie on~$h=2$ face boundaries, the problem is
solvable even in~${\Oh(n \log n)}$ time~\cite{KMN01}.  On the plane, only a
slightly different variant of the problem has been
studied~\cite{kindermann2018partition} where the spanning trees are not allowed
to visit terminals of other colors.  The authors obtained NP-hardness proofs and
polynomial-time algorithms for various special cases but their results have not
yet been formally published\todo{KF: Check (and ask the authors about the status/possible issues.}. 
There is also an approximation result for a
noncrossing problem called colored spanning trees~\cite{EHKP15}. However,
despite its name, this problem is a generalization of the Steiner tree problem
allowing arbitrary Steiner points. 

\paragraph*{Noncrossing Paths}
We also revisit the rather classic problem that we call \MSP. 
Given a set of \importanttodo{KF: use the word terminal? We have to be consistent throughout the whole paper!}
terminal point pairs in the plane, the task is to connect each pair by a path such that
the paths are pairwise noncrossing and their total length is minimized. 
One can
think about this problem as an extension of our results
\todo{KF: Which of our results? For the separation,tours, or spanning tree problems? And do we mean results or rather extension of the problem. In any case, I don't see it.} 
to a large number of
colors. 

The best-known result is a randomized~${\Oh(\sqrt n \log n)}$-approximation algorithm 
by Chan et al.~\cite{CHKL13}. 
It is based on the heuristic by Liebling et al.~\cite{LMMPS95}
to connect the terminal pairs along a single tour through all the points.
As a special case of the colored noncrossing Steiner forest problem, 
the problem also admits a (deterministic)~$(1+\eps)n/2$-approximation scheme~\cite{isaac15}. 

The problem has been also studied in the presence of obstacle polygons whose boundaries contain all the terminal pairs. 
For the case of a single obstacle,
Papadopoulou~\cite{Pap99} gave
a linear-time algorithm, whereas for the general case, Erickson and Nayyeri~\cite{EN11} obtained an algorithm exponential in the number of obstacles.
A practical extension to thick paths has been considered by Polishchuk and Mitchell~\cite{PM07}.

In this paper, we complement these algorithmic results and demonstrate that the
problem is NP-hard\footnote{Note that a technical report~\cite{bastert1998geometric} claims
    NP-hardness, but %after publication 
    it seems to be missing a subtle detail in
the analysis~\cite{PrivFekete}.}.

\begin{restatable}{theorem}{thmNPpaths}  \label{thm:MSPhard}
Euclidean \MSP{} is NP-hard\kt{Or even NP-complete?}.
\end{restatable}

Thus, similar as \BST{} in plane graphs, \MSP{} is a nice example of a problem whose hardness
comes from the noncrossing constraint.

\paragraph*{More Related Work}
\todo{KF: Fill his paragraph with more content?}
Very recently, Abrahamsen et al.~\cite{abrahamsen2020geometric} showed that the 
red-blue separation problem for geometric objects is polynomially time solvable (for two colors) 
if we allow any number of separating polygons
and relax the connectivity requirement by allowing objects of the same color to lie 
in different regions.
Interestingly,
they showed that the red-blue-green separation problem for arbitrary objects remains
NP-hard (for three colors) even if objects within each polygon do not need to stay
connected.

We remark that the red-blue-green separation problem is related to the painter's
problem~\cite{goethem2017painter}, where a rectangular grid and a
set of colors $\chi = \{\text{red},\text{blue}\}$ is given.  Each cell $s$ in
the grid is assigned a subset of colors $\chi_S \subseteq \chi$ and should be
partitioned such that, for each color $c \in \chi_s$, at least one piece in the
cell is identified with $c$. The question is to decide if there is a partition
of each cell in the grid such that the unions of the resulting red and blue
pieces form two connected polygons. 
Van Goethem et al.~\cite{goethem2017painter} introduce a patching procedure that is similar to
ours and used it to show that if the partition exists, then there exists one
with a bounded complexity per cell. 
In a sense the problems are incomparable: in the red-blue-green separation
problem a feasible solution always exists and the hard part is to find
a minimum length solution. Additionally, our patching needs to be significantly more robust: we have to execute the
patching procedure in portals (not in the whole cell) and our polygons can contain each
other. Because of these technical issues, the number of crossings in our patching
procedure is slightly higher (but it is still a constant).

\paragraph*{Overview of the Paper}
We start with a condensed overview of our techniques (Section~\ref{sec:techniques}) followed by a formal definition of our problems and a short preliminaries (Section~\ref{sec:prelim}).
Then, in Section~\ref{sec:patching}, we present our patching procedure which is a key ingredient to all our positive results:
an EPTAS for Euclidean \BTSP{} (Section~\ref{sec:noncrossing-tsp}) and for  Euclidean \RBGS{} (Section~\ref{sec:rgbseparation}), and
a PTAS for \BTSP{} and \BST{}, both in planar graphs (Section~\ref{sec:ptas-graphs}).
%
%Next, in Section~\ref{sec:ptas-graphs}, we obtain a PTAS for \BTSP{} in planar graphs and discuss why this result also holds for \BST{} in the planar graph setting.\kt{Mention the result for \BST{} as a corollary above?}
At the end, we show NP-hardness of \BST{} (Section~\ref{sec:bst-np}) as well of \MSP{} (Section~\ref{sec:msp-np}).

\importantquestion[inline]{KF: Maybe we should write about natural results that we didn't obtain. (Plane versus graph results.). And about more colors (see SDOA rebuttal).}

\section{Our Techniques}\label{sec:techniques}

Our algorithmic results are based on the plane dissection technique
proposed by Arora~\cite{Arora98}. In this framework, the space is recursively
dissected into squares in order to determine a \emph{quadtree} of
$\Oh(\log(n/\eps))$ levels. The idea is to look for the solution that
traverses neighboring cells via preselected \emph{portals} on the boundary of
the cells of the quadtree. Arora~\cite{Arora98}
defines portals as a set of $\Oh(\log(n)/\eps)$ equidistant points and uses dynamic
programming to efficiently find a portal-respecting solution. 
He also shows that the expected cost of the portal respecting-solution is only
$\eps$ times longer than the optimum one. See Section~\ref{sec:prelim} for a
detailed introduction to this framework~\cite{Arora98}.

This technique will not allow us to get a near-linear running time. 
One approach to
reduce the running time of Arora's algorithm~\cite{Arora98} is to use Euclidean
spanners~\cite{RaoS98}. However, it is not clear how to ensure that our solution is noncrossing 
when we restrict ourselves to solutions that respect only the edges
of the spanner. To overcome this obstacle, we use a recently developed
\emph{sparsity-sensitive patching} procedure~\cite{focs21}.  Roughly speaking, this
allows us to reduce the number of portals from $\Oh(\log(n)/\eps)$ to
$\Oh(1/\eps^2)$ many portals without a need for spanners. 

The sparsity-sensitive patching procedure was designed explicitly to solve 
the Euclidean traveling salesman problem. We
modify the procedure and show it can also be used for
the red-blue-green separation problem. In Section~\ref{sec:noncrossing-tsp}, we show
that the analysis behind sparsity-sensitive patching can be adapted to also work
for the noncrossing (traveling salesman) tours problem. This is quite subtle and we need to
resolve multiple technical issues. For example, even the initial perturbation
step used in previous work~\cite{Arora98,focs21} needs to be modified as we cannot simply
place the points of the same color in the same point. We need to snap to the
grid in such a way, that after looking at the original positions of points, the
solution is still noncrossing. We overcome all of these technical issues in
Section~\ref{sec:noncrossing-tsp}.

\subparagraph*{Patching for noncrossing polygons}

The main contribution of our work lies in the analysis. For example, in order
to bound the number of states in our dynamic programming of our framework, we need to bound the
number of times the portal respecting solution is intersecting the quadtree
cells. Our insight is a novel \emph{patching technique} that works for noncrossing
polygons. 
\begin{insight}
    \textbf{Our Insight:}
    Two noncrossing polygons can be modified (with low cost) in such a way that
    they cross the boundaries of a random quadtree only a constant number of times.
\end{insight}
Similarly to the patching scheme of Arora, our patching technique allows us to limit the number of
times that two noncrossing polygons cross a boundary of the quadtree. 
However, our situation is significantly more complicated. We
need to guarantee that, after the patching is done, the tours remain noncrossing
and contain the same set of points as the original ones. 

\begin{figure}[ht!]
    \centering
    \includegraphics[width=0.8\linewidth]{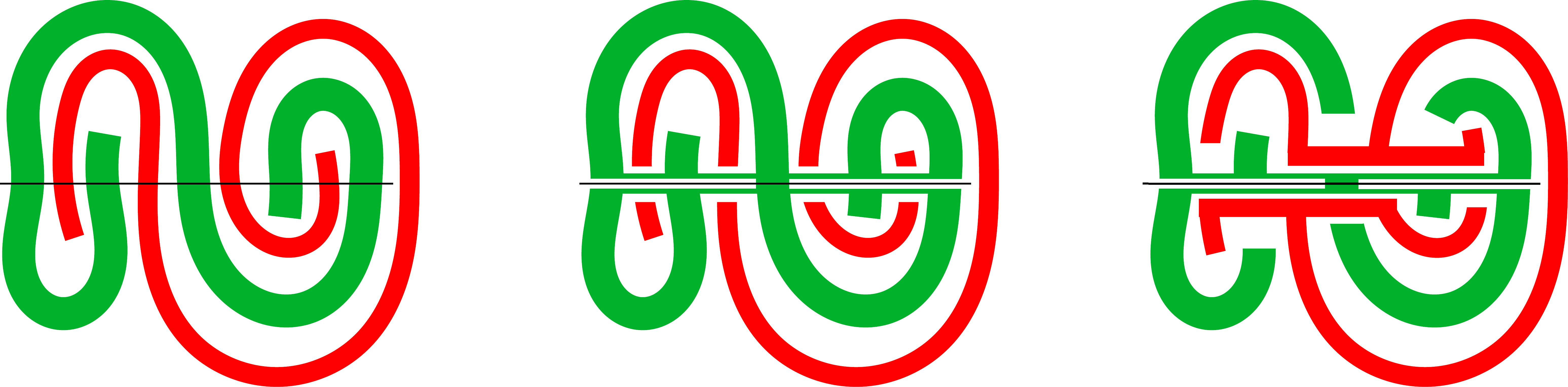}
    \caption{An example of our patching procedure for a case of disjoint objects that cross a horizontal line segment. 
    	The left figure
    presents the original objects (marked red and green). 
    In the intermediate
    step, we patch the green object by possibly disconnecting the red object. 
    Then we carefully cut the green object without disconnecting it to make \enquote{room} to reconnect the red object and guarantee connectivity. 
    Note that the boundaries of the objects cross the horizontal segment
    only four times in the final picture. The whole construction is presented in
    Section~\ref{sec:patching}}%
    \label{fig:patching}
\end{figure}

Now we briefly sketch the idea behind our patching scheme (see
Section~\ref{sec:patching} for a rigorous proof and Figure~\ref{fig:patching}
for a schematic overview). First, we apply several simplification rules to
our polygons. We want to ensure that two colors are repeating consecutively in
pairs. For example, in Figure~\ref{fig:patching}, the borders of the green and the red
polygon are pairwise intertwined in left-to-right order. Next, we introduce a
\emph{split} operation. After this operation is completed, the number of crossings is
bounded, but it is not necessarily true that the polygons of the same color are
connected (see the center of the Figure~\ref{fig:patching}). At this point, the polygons
form a \emph{laminar family}. As the next step, we \emph{merge} the
polygons from the laminar families at so-called \emph{precise interfaces} and combine
them based on how the original polygons were connected. We refer the reader to
Section~\ref{sec:patching} for a formal argument. In the end, we can bound the
number of crossings by $10$. 

\subparagraph*{Insight behind the lower bounds}

Next, we describe our techniques behind the lower bounds. For brevity, we focus
on the lower bound for \MSP. We reduce from Max-2SAT. Recall that the input
to the problem is just a set of terminal points. We need these points to have some ``rough
structure'' on top of which the paths from higher levels will bend. This
inspires the idea to partition the terminals into four levels.  These levels
are needed to devise appropriate gadgets.  We want to place the terminals in
such a way that gadgets from higher levels do not interfere with gadgets
from lower levels.  The idea is to place terminal pairs from lower levels more
densely than from the higher ones. This allows us to enforce the optimum
solution to be exact on the pairs from level one. See Section~\ref{sec:msp-np}
for a formal proof.

The NP-hardness proof of \BST is more standard. We reduce from the Steiner tree problem
and show that the vertices of the second tree can be used as Steiner
vertices in the first tree. We include the proof for completeness in
Section~\ref{sec:bst-np}.

\section{Preliminaries}
\label{sec:prelim}
\label{sec:formalDef}
    \reviewertodo{"make it clear that a tour does not have to consist of straight line segments and not necessarily between sites. But we can assume that in some sense."}
    \reviewertodo{define "open" in "open curves" (the endpoints belong to the curves but they are different)}
    \reviewertodo{TODO? add info for planar graphs "(most easily described by rotation-system-based data structures like the quad-edge or DCEL that would allow to split the cycle of vertices and faces around a vertex where the color changes.)"?}
    \reviewertodo{Shorten formal problem definition as intuition already given (but extend it for journal version)}

\importanttodo{KF: We should define \enquote{noncrossing} and \emph{to cross} (also/only in Arora's sense)}\todo{KF: We should make clear that by saying \enquote{two noncrossing curves}, we do not mean two curves that are simple but might cross each other.}
We start by formally defining the setting of our problems and discussing technical subtleties and differences to related work; see Section~\ref{sec:formalDef}. Then we review some known tools that we later use to prove\todo{KF: construct/design?} our approximation schemes; see Section~\ref{sec:arorasTools}.

% \subsection{Formal Definition of the Problem Settings}\todo{KF: Short title, e.g., The Problem Settings}
\infodone{KF: I rephrased the whole section (after our submission) as there were some inconsistencies.}
\todo{KF: Shorten the whole section (for the conference version)}

In the geometric setting, the input to our problem consists of $n$ points in
$\mathbb{R}^2$, called \emph{terminals}. Each terminal is colored with one of
$k$ available colors.
In the combinatorial setting, the input consists of an edge-weighted planar graph
together with a planar \todo{KF: combinatorial? Or topological?\\In the geometric setting paragraph, we assume the combinatorial embedding.} embedding. Some of the vertices (called terminals) are
colored with one of the~$k$ colors. 

Intuitively speaking, the goal in both settings is to draw pairwise-disjoint geometric graphs of minimum total length.
Depending on the problem, either the graphs separate any two terminals of different colors (but not of the same color), or each graph connects all the terminals of one color. 
In the first case, we thus have~$k-1$ graphs, in the second one,~$k$ many.
For the geometric setting, we
consider the Euclidean length, whereas for the combinatorial setting, the length depends
on the edge weights. Additionally, for the combinatorial setting, we require that the
solution %connecting structure 
lies within the point set of 
a \enquote{thick} planar drawing of the input graph (realizing the embedding given in the input) 
\todo{KF: We can also say: \enquote{in the given aforementioned planar drawing of the input graph}, but then we need to specify that the input gives a topological embedding, and then we have the problem how to say that we mean a \enquote{thick} variant of that problem.} 
with all Steiner points residing on the vertices.

\paragraph*{Geometric Setting.}
Before we define the goal for the geometric setting, we first precise what we mean by a \emph{feasible} and an \emph{optimal} \emph{solution}.
    \reviewertodo{	187-190: You don’t have conditions that require TSP to be tours.  (e.g., 3 open curves could share an endpoint).
    	You insist that TSP be connected; you also want that (and tour conditions) for each colored curve in RGB Separation.}
A solution consists of several \emph{drawings} \todo{KF: The word drawings might be confusing. Maybe it's better to call components?} 
where each drawing is the point set of finitely many simple open \todo{KF: I think \enquote{open} was confusing some referee, so we could define it.} 
curves in
the plane where two curves may intersect only at their endpoints and where all the curves are connected together (that is, any two curves either share an endpoint or there is a third curve such that both are connected with it). 
Thus, a drawing can also contain closed curves built up by two or more open curves.\todo{KF: then why not to allow closed curves directly? (The reason seems to be that we want that each curve has one or two special points that are used to connect to other curves.\\KF: However, I'm not so happy anymore about this definition; maybe we can make it simpler, e.g., we have a finite set of curves, closed or open, s.t. any two curves intersect at most at one point, and all curves are connected.}
The \emph{length of a drawing} is the total Euclidean length of the curves it consists
of. The \emph{length of a solution} is the total length of its drawings. 
The (Euclidean) length of an object $\pi$ is denoted by $\wt(\pi)$.
Instead of length, we also use the term \emph{cost}.

For the \kCS{} problem, %the terminals are colored in three colors. 
a feasible solution consists of~$k-1$ pairwise disjoint drawings 
where each drawing is a closed curve and where 
any two terminals of different colors are \emph{separated} by at least one of the drawings, that is, 
any path connecting these two terminals 
contains (intersects) at least one point from the solution. 
In other words, the solution consists of~$k-1$ closed curves 
that subdivide the plane into~$k$ regions 
each one containing all the terminals of one of the~$k$ colors\todo{KF: in its interior?}.
We examined the case for~$k=3$ called the \RBGS{} problem.

For all the other problems studied in this paper, 
a feasible solution consists of~$k$ pairwise disjoint \emph{feasible drawings}, one for each color.
A drawing is feasible if it \emph{visits} (intersects) all terminals of its color. 
For the \BTSP problem, we additionally require that the feasible drawings are closed curves that we also call \emph{tours} 
\todo{KF: or \emph{traveling salesman tours}? Also \emph{paths}? We should also somewhere discuss the path variant (at least why our results do not carry over to TSP path.}.

\begin{figure}[ht!]
	\centering
	\begin{subfigure}[t]{.22\textwidth}
		\centering
		\includegraphics[page=1]{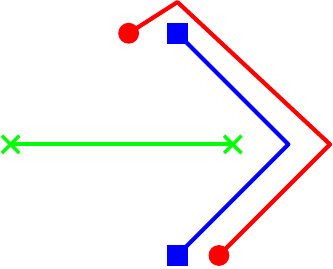} 
		\caption{A feasible solution.}
		\label{fig:no-min-sol-a}
	\end{subfigure}
    \hspace{0.7cm}
	\begin{subfigure}[t]{.22\textwidth}
		\centering
		\includegraphics[page=2]{img/no-min-sol.pdf} 
		\caption{A cheaper feasible solution.}
		\label{fig:no-min-sol-b}
	\end{subfigure}
    \hspace{0.5cm}
\begin{subfigure}[t]{.22\textwidth}
	\centering
	\includegraphics[page=3]{img/no-min-sol.pdf} 
	\caption{An infeasible but noncrossing solution in the limit.}
	\label{fig:no-min-sol-limit}
\end{subfigure}
	\caption{Three pairs of terminals (red disks, green crosses and blue squares) are connected via noncrossing paths.
	\reviewertodo{Fig 4 caption: identify the problem that this is illustrating (it is neither of the two just mentioned: Bicolored noncrossing TSP nor Red-Blue-Green Separation).  The cycles->circles typo had me confused.  
		In general, I like the biochemist convention that the caption should contain all the context needed to understand a figure. (See, e.g., https://journals.plos.org/ploscompbiol/article?id=10.1371/journal.pcbi.1003833 Rule 4.)}}
	%A feasible solution, a cheaper one, and an infeasible but noncrossing solution in the limit.}
\label{fig:no-min-sol}
\end{figure}

Informally, the goal is to find a solution of minimum length.
However, note that for
some input instances there does not exist a feasible solution of minimum length as there might
be always a cheaper one obtained by drawing the curves closer to each other; see
Fig.~\ref{fig:no-min-sol} for an example. 
In the limit, we obtain
curves that may intersect, but that do not cross in the sense introduced by
Arora~\cite{Arora98}. 
In this sense, we define \emph{optimum solutions} to be minimum-length solutions infinitesimally close to the limit (but still disjoint).
However, for simplicity, we assume that the length of an optimum solution equals the length~$\opt$ in the limit.
This assumption is also motivated by the known open problem whether it is possible to efficiently compute the Euclidean length of two sets of line segments up to a precision sufficient to distinguish their lengths.
%
%At this point, also note that we are not known open problem whether we can compute the Euclidean length up to a satisfacory precsion  it is an open problem whether one can compare the Euclidean lengths of two sets of line 
%since it is an open problem whether one can  it is an open problem whether we can compute the Euclidean length  of a given network square roots when computing the Euclidean distance, 
%
%    \reviewertodo{205: You might mention the difficulty of summing square roots to actually compute $\opt$. }
    
%However, to overcome this technical subtlety, we allow the curves \todo{KF: everything here is very confusing (critics to myself).}
%to be placed infinitesimally close to each other, and still consider them
%noncrossing. This assumption is only due to technical reasons: 
%%we consider grid lines (defined in the next section) as infinitesimally thin slabs of infinite height, \importanttodo[inline]{KF: I think we treat grid lines as lines; portals are indeed as defined.\\Shoudln't this discussion be in the paragraph about differences to previous definitions?}
%and we consider portals (defined Section~\ref{sec:arorasTools}) as infinitesimally short line segments. 
%Note that our solutions can be made feasible \todo{KF: Wait, why do we suddenly assume that our solution might be not feasible?!} 
%without practical change in the cost by slightly increasing the width of the portals and gridlines (and by slightly shifting the solution). 

Let us now formally define the goal of our $(1+\eps)$-approximation schemes (\kt{Move the definition of~$\eps$ somewhere before! (Recall reviewers' comments)} for constant~$\eps>0$) for the geometric setting in this paper. The goal is not to find a feasible almost-optimal solution but a so-called \emph{\costapprox}.
\begin{definition}\label{def:costapprox}
Let~$X$ be a problem in our geometric setting and let $\mathcal{F}$ be the (possibly infinite) set of
all feasible solutions to the problem. 
A number~$x$ is a \costapprox{} of~$X$
if
\begin{enumerate}
    \item for every solution $S \in \mathcal{F}$, it holds that $x \le (1+\eps) \wt(S)$, and
    \item there exists a solution $S' \in \mathcal{F}$ such that $(1-\eps) \wt(S') \le x$.
\end{enumerate}
\end{definition}

The two conditions guarantee that we return a number that is sufficiently
close to the lengths of the cheapest feasible solutions (and thus to $\opt$). 
\todo{KF: We should add some comment about it. E.g., say, why we do it and whether our algorithms can be extended to also output the solution?} 
The factor~$(1-\eps)$ in the second condition 
is only introduced to simplify our analysis: if there are two points in our solution 
that are infinitesimally close, we will count their distance to be $0$ (and
therefore, in principle we may return a number $x$ that is smaller than $\wt(S)$
for every $S \in \mathcal{F}$). 
However, the second condition can be easily lifted to the more natural condition $\wt(S') \le X$ 
by multiplying $x$ by $1/(1-\eps)$ and using a different constant in the first condition (on which~$\eps$ will depend). 

\paragraph*{Combinatorial Setting.}
Next, we consider our problems in plane edge-weighted graphs. 
Formally, we fix any planar straight-line drawing of the given graph that
realizes the given embedding. Without loss of generality, all vertices are drawn
as closed discs with positive radius such that no two vertices intersect, and
all edges have some positive width (in the sense of taking the Minkowski sum of
the line segment corresponding to the edge and a disc of a sufficiently small
radius). We require that an edge intersects only the vertices corresponding to its
endpoints and no other edge. 
Terminals are drawn as points in the interior of their corresponding vertices.

A feasible solution is defined in the same way as in the geometric setting \todo{KF: recall in parentheses the definition} with one
additional constraint. Namely, 
every curve is a straight-line segment drawn either within (i) the disc of a
single vertex, or (ii) 
\todo{KF: complicated definition; simpler(?) proposal: within a single edge such that the endpoints of the curve lie in opposite endpoints of the edge (by our definition, every edge slightly overlaps with the vertices corresponding to its endpoints)}
within a single edge and the discs of the edge's endpoints
such that the endpoints of the curve lie in opposite endpoints of the edge. If
the curve lies within a vertex, we define its length to be $0$, otherwise its
length is defined to be the weight of the corresponding edge. \todo{KF: Add a figure!}

In contrast to the geometric setting, there exist feasible solutions of
minimum length $\opt$ that we will call optimal solutions. Our goal is to compute the value $\opt$. \importantquestion{KF: Right? Or do we compute an actual optimal solution?}  

% \karol{Why do we need an alternative definition?}
%
% An alternative but equivalent definition of a feasible solution that we will
% implicitly use for the output is that a solution defines for each edge an
% arbitrary orientation and a (possibly empty) sequence of (possibly repeating)
% colors from~$1$ to~$k$. On each vertex, the sequences from incoming edges and
% the reversed sequences of the outgoing edges form a circular sequence for which
% the solution defines a noncrossing partition of monocolored parts. By
% interpreting each part of the partition as a vertex and each next color from the
% circular sequence as a single edge, we obtain a graph for each color. If for
% each color the corresponding graph is connected and visits all terminals of that
% given color, we call such a solution feasible. This is equivalent to the
% definition above as all curves crossing an edge define the colored sequence of
% that edge, and the way that the endpoints of these curves are connected by
% zero-length curves within a vertex defines the noncrossing partition of that
% vertex.

\paragraph*{Difference to Previous Definitions.}

Note that our definition of a feasible solution slightly differs from the one in
most previous works~\cite{EN11}. 
There the drawings of different
colors are allowed to touch each other but must not cross in an intuitive sense.
(Hence, such drawings are the limit case of our drawings when the length is
locally minimal). The advantage of the previous definition is the existence of
a minimum solution. The disadvantage is that crossings are formally
defined via an untangling argument~\cite{EN11} 
which is hard to formally define
\todo{KF: In a journal version, we could discuss this further showing a picture that the old formal definitions have some formal flaws} 
and which essentially reduces to the statement that
a solution is noncrossing if it can be made disjoint by an arbitrary small increase in the costs.
%which only works when the drawings are (non-closed) paths.
%It also not seem to be 
%\definition is not 
%extendable to cycles or trees
%that occur in our case.
%; see our discussion in Fig.~\ref{}\todo{KF: add figure with discussion.}.
Furthermore, the classical definition allows solutions with possibly undesirable properties;
for example, a feasible solution can contain arbitrarily complicated networks of zero total length residing on a
single point in the plane.

\subsection{Useful Tools from Arora's Approach}
\label{sec:arorasTools}
%\reviewertodo{Move Sec 2.2 to the appendix? "Sec 2.2: Tools from Aurora: (1.5 pages): These pages are not well-used in the abstract, since all formal application of these tools is in the Appendices."}%KF: not relevant for SODA as no appendix there

In this section, we describe the tools that Arora~\cite{Arora98,Arora2007} introduced to design his approximation scheme \infodone{KF: changed schemes to scheme} for geometric problems and that we further use in this paper. 
As defined above, the input of our geometric problems consists of 
%As our problem setting, 
%we consider geometric problems (as defined above) where the %input consists of a number of disjoint sets of terminal points in $\Reals^2$, or to put it differently, our 
%input is 
a set $P$ of terminal points in~$\Reals^2$ colored with some number of colors. It is fairly standard 
to preprocess the input points so that $P \subseteq \{0,\ldots,L\}^2$ for some integer $L = \Theta(n / \eps )$ that is a power of $2$, even with the guarantee that points of different colors do not end up with the same coordinates. 
\infodone{KF: changed it to~$\Theta$ as othewise our analysis in the end will not work where we bound the perturbation.} 
\kintp{$\eps$ undefined. Define it in the beginning of the whole section?(We us it also in subsection 1.)}
\todo{KF: $n$ also undefined.}
For noncolored problems, such a peturbation costs only a $(1+\eps)$ factor in the approximation ratio, but it is not obvious whether the bound holds also for multicolored problems due to the noncrossing constraint. 
Our contribution is to show the same cost bound for mutlicolored problems considered in this paper.\ktp{Mention this in our contribution in the introduction?}
%
%Whereas for noncolored problems it's not hard to see that such a peturbation costs only a $(1+\eps)$ factor in the approximation ratio, for multicolored problems such a cost bound is not obvious.
%
% is more challenging, but we manage to obtain the same bound.
%For mulitcolored problems, this step is more challenging and requires a careful cost analysis to gu
%We extend this perturbation step (with same approximation cost) in a non-trivial way to mutlicolored problems and obtain even the guarantee that points of different colors do not end up with the same coordinates. 
Since the perturbation differs from problem to problem, we separately describe this preprocessing in detail in the respective sections.
%However, for mutlicolored problems, this perturbation step is more challenging, but still 
%We manage to , but we manage to achieve it with the guarantee that points of different colors do not end up with the same coordinates. 
%However, we also achieve 
%
%This preprocessing costs an $\eps$ factor in the approximation ratio, and comes with the guarantee that points of different colors do not end up with the same %integer 
%coordinates\todo{KF: not clear: different positions, or different~$x$-coordinates and different~$y$-coordinates?}. % after preprocessing. 
%We separately describe this preprocessing in detail for each problem that we study in the following sections\todo{KF: for each geometric problem that we study in this paper?}. %  in the follow up sections introducing the results for particular problems. 
%%In the following we assume an instance of \textsc{Euclidean TSP} is given by a
%%point set $P \subseteq \mathbb{R}^2$.  By using a standard $\Oh(n \log n)$ time
%we assume that 

%\importanttodo[inline]{KF: Check whether~$L$ power of~$2$ or whether the box has some othre size than~$2L$ (maybe~$2L-1$ given~$2L$ grid lines??)} %It's OK, as size of big box is $2L$ implying~$2L+1$ grid lines, which satisfies~$2L+1 = 2 + (2^j-1)$.

\todo{KF: Add: \emph{Recall that\dots}} A \emph{salesman path} \importanttodo{KF: But a path is not closed!?} or a \emph{tour} for a terminal set $P$ is a closed path
that visits all points from $P$. Let us point out that in a tour the terminals are
not necessarily connected with straight line segments. \todo{KF: Why do we point this out? Do we keep in mind that Lemma 2.2 assumes line segments?\\KF: If we really want to say it, let's extend the preceding sentence by saying \dots and that does not necessarily consist of straight-line segments.} 
For a line segment
$\segm$ and a tour $\pi$, we let $\is(\pi,\segm)$ denote the (finite) set of points where
$\pi$ crosses~$\segm$ (following Arora's intuitive definition of
crossing~\cite{Arora98}\todo{KF: Check also~\cite{Arora2007}})\importanttodo[inline]{KF: We should repeat Arora's definition of noncrossing here (or even better: somewhere above) if we use it!}. The set~$\is(\pi,\segm)$ can also be interpreted as the
set of intersection points of $\pi$ with $\segm$\importantquestion{KF: But then we might have infinite many intersection points?}, noting that infinitesimally close curves do not intersect. 
%In fact, in our applications $\pi$ is always a set of finite number of
%segments (and no endpoint of this segment lies directly on $\ell$). Hence
%$|I(\pi,\ell)|$ is the (finite) number of segments that cross $\ell$ and is well
%defined (see~\cite{geomspannet}). %but
%may make some detours. 
The following folklore lemma is typically used to
reduce the number of times a tour crosses a given segment.

\infodone[inline]{KF: Since we don't seem to use this lemma anywhere, I rephrased it in order\\\quad1) to give it the same structure and the same (sufficiently low) level of detail as our patching lemma,\\\quad2) to make it more consistent with what Arora actually claimed (he did not speak of an arbitrary point $p$ (though this variant can be proven in the same way as Arora did),\\\quad3) to make it less misleading and precise (as the old variant had some small mistakes and inconsistencies because of the high level of detail).}
\begin{lemma}[Patching Lemma~\cite{Arora98,Arora2007}]\label{lem:patch}
	Let $\segm$ be a line segment and let~$\pi$ be a simple closed curve. \todo{KF: To be more precise, we could say: \emph{For any infinitesimal neighborhood of $\segm$, there exists \dots}}
	There exists a simple closed curve~$\pi'$ 
	such that~$|\is(\pi',\segm)|\leq 2$ and~$\wt(\pi') \leq
	\wt(\pi)+ 3\cdot \wt(\segm)$. Moreover, $\pi$ and
	$\pi'$ differ only within an infinitesimal neighborhood of $\segm$. 
%	in which~$\pi'$ contains only~$O(1)$ segments that are not perpendicular to~$\segm$.
%	in which~$\pi'$ intersects any line from a given finite set of lines perpendicular to~$\segm$ only~$O(1)$ times.
\end{lemma}

%% OLD buggy VERSION OF THE PATCHING LEMMA STATEMENT
%% -------------------------------------------------
%\begin{lemma}[Patching Lemma~\cite{Arora98}]\label{lem:patch}
%\todo{KF: Be consistent: infinitely close vs infinitesimal close.\\KF: To me it seems that both variants should be correct.}
%    Let $\segm$ be a line segment and $\pi$ be a closed path. %, and
%    %of intersections of $\pi$ with $\ell$.  Suppose $S\subseteq \ell$ is a
%    %segment on $\ell$ that spans $I(\pi,\ell)$.
%    For any point $p$ in $\segm$ there exist a set of line segments that are infinitely close to~$p$\importanttodo{KF: or to $S$?}
%    and parallel\importanttodo{or pependicular for the two crossing segments??} to~$\segm$ whose total length is at most $\Oh(\wt(\segm))$ and whose addition
%    to $\pi$ together with the removal of some parts of~$\pi$ changes the path into another closed path $\pi'$ that crosses~$\segm$ 
%    at most twice and only at $p$. Moreover, if $|\is(\pi,\segm)| > 2$, then
%    $\pi'$ crosses $p$ exactly twice.\importanttodo[inline]{KF: at most twice? In Arora for odd number of crossings, it will be one crossing only, for even number of crossings it will be two crossings.\\Arora's lemma does not allow n arbitrary point. Should we give some arguments (in our journal version) how we can extend Arora's proof or is this also folklore (but then we might cite something where this extended lemma is used/proven)?} 
%\end{lemma}

\paragraph*{Dissection and Quadtree.} Now we introduce a commonly used hierarchy
to decompose (subspaces of) $\mathbb{R}^2$ 
% \importantquestion[inline]{KF: Why decompose $\mathbb{R}^2$?
% Rather decompose the input instance?\\ KW: Dissection is a plane decomposition\\KF: Do you mean \emph{to decompose the whole plane~$\mathbb{R}^2$} or \emph{\dots to decompose subspaces of $\mathbb{R}^2$}?} 
that will be instrumental to guide our algorithm.
Pick $a_1,a_2 \in \{1,\ldots,L\}$ independently and uniformly at random
and define the \emph{random shift vector} as $\mathbf{a}:=(a_1,a_2)$. Consider the square
\[
C(\mathbf{a}) := [-a_1 + 1/2,\,2L-a_1+1/2] \times [-a_2 + 1/2,\,2L-a_2+1/2]
\]
Note that $C(\mathbf{a})$ has side length $2L$ and each point from $P$ is contained in $C(\mathbf{a})$ by the assumption $P \subseteq \{0,\ldots,L\}^2$.

Let the \emph{dissection} $D(\mathbf{a})$ of $C(\mathbf{a})$ be the tree $T$
that is recursively defined as follows. With each vertex of $T$, we associate an axis-aligned
square in $\mathbb{R}^2$ that we call a \emph{cell} of the dissection. 
For the root of $T$, this is $C(\mathbf{a})$. %, and we make each vertex associated with a square of unit length a leaf of~$T$.
If a vertex~$v$ of~$T$ is associated with a square of unit length, we make it a leaf of~$T$.
Otherwise,~$v$ has four children whose cells partition the cell of~$v$. Formally, if~$[l_1,u_1] \times [l_2,u_2]$ is the square associated with~$v$, then each of its four children is associated with a different square~$I_1 \times I_2$ where $I_i$ is either $[l_i,(l_i+u_i)/2]$
or $[(l_i+u_i)/2,u_i]$ for~$i \in \{1,2\}$.
%The two line segments that cut the cell of~$v$ into the cells of its  a cell into its four children,

The \emph{quadtree} $QT(P,\mathbf{a})$ is obtained from $D(\mathbf{a})$ by
stopping the subdivision whenever a cell has at most one point from the input
terminal set $P$. This way, every vertex is either a leaf whose cell (not necessarily a unit square) contains at most one terminal,  
or it is an internal vertex of the tree with four children whose cell contains at least two terminals.

\begin{figure}[ht!]
    \centering
    \includegraphics[width=0.4\textwidth]{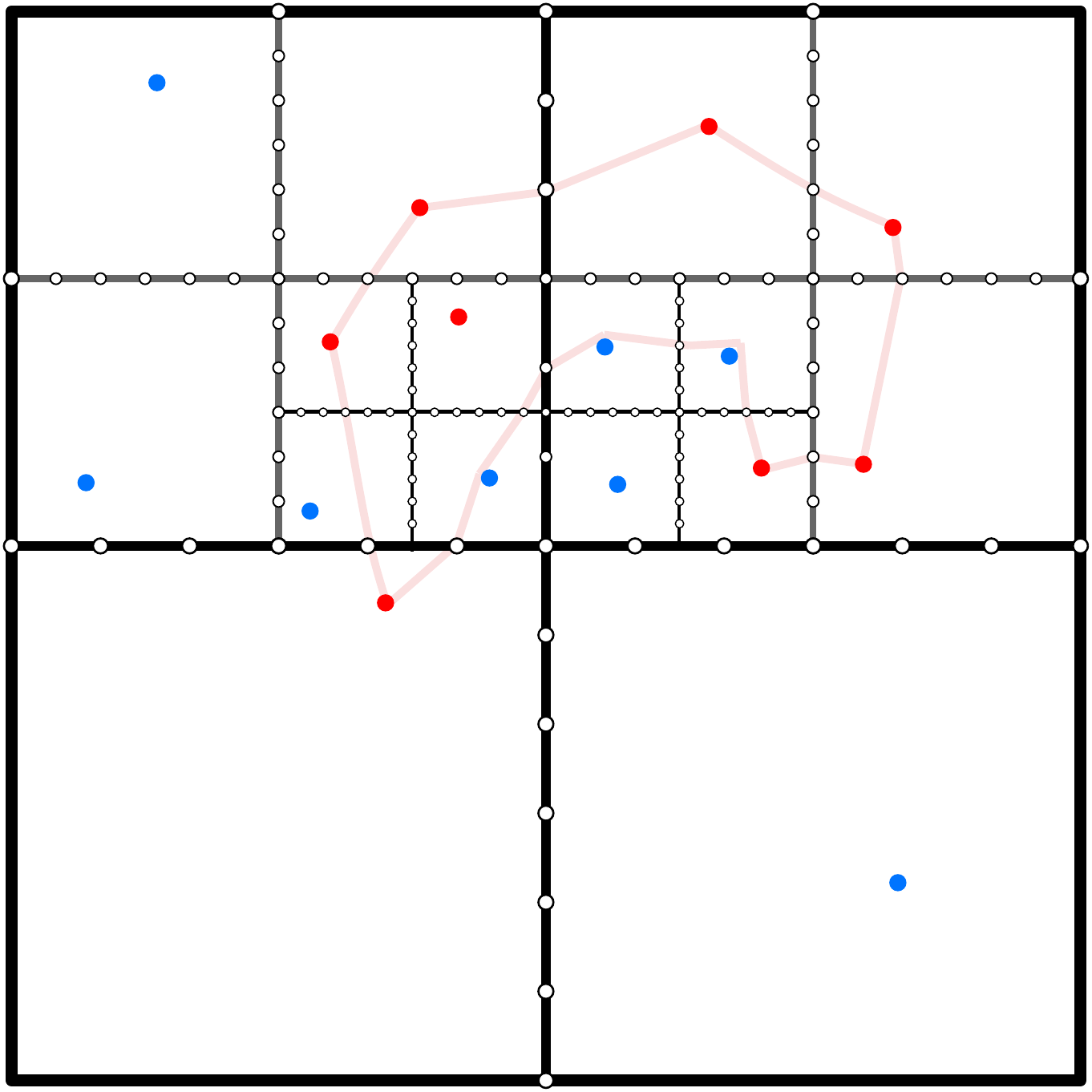}
    \caption{\todo[inline]{KF: missing reference to this figure?}\todo[inline]{KF: polish the caption.}
    \todo[inline]{KF: portals of deepest level too small.}
    \todo[inline]{KF: red tour barely visible.}
    	An illustration of the construction of the quadtree for Red-Blue
    Separation problem. Input points are represented with red/blue circles. The
    white circles are representation of portals. Note that each boundary
    boundary of quadtree has exactly the same number of portals. The pink polygon is a 
    portal-respecting\todo[inline]{KF: we haven't defined portal-respecting!} polygon based on the technique by Arora and Chang~\cite{red-blue}.}
    \label{fig:fig_ref}
\end{figure}

Let $V(x)$ \infodone{KF: dropped for our conference version the formal definition $V(x) =\{ (x,y) \}_{y \in \Reals}$, same for~$H(y)$ below.} be the vertical line crossing the point $(x,0)$ and $H(y)$ %=\{ (x,y) \}_{x \in \Reals}$ 
be the horizontal line crossing the point $(0,y)$. 
A grid line is either a horizontal line $H(y)$ for $y=\frac{1}{2}+i$ where $i$ is integer, or a vertical line $V(x)$ for $x=\frac{1}{2}+i$ where~$i$ is integer.
% A \emph{grid line} is either a horrizontal line  point set of the form $\{ (x_1,x_2 : x_i = 1/2 + j)\}$ for
% some integer $j$.  
For a line \todo{KF: It seems that we want that~$\ell$ can be a (finitely long) line segment as well as an (infinitely long) straight line. Should we extend our definition to the line segment case?} $\ell$ and a set~$\mathcal{S}$ of line segments, we define~$\is(\mathcal{S},\ell)$ 
as the set of all points through which the segments of $\mathcal{S}$ cross $\ell$.  \importanttodo{KF: We should either add that the segments are not parallel to~$\ell$ (and then also update all lemmas using this notation, eg., Lemma~\ref{lem:crossingsvslength}) or that we use the definition of crossing of Arora and emphasize that in the case of infinite many intersections points, Arora's definition allows to select a finite number of crossings points.}
Note that for every border edge~$F$
of every cell in $D(\mathbf{a})$, there is a unique grid line that contains $F$.
The following simple lemma relates the number of crossings between a set of line segments and the grid lines with the total length of the line segments; note that we assume that all endpoints are integer.
\begin{lemma}[{\cite[Lemma 19.4.1]{geomspannet}}]\label{lem:crossingsvslength} 
	If $\mathcal{S}$ is a set of line segments in the a infinitesimal neighborhood of $\mathbb{Z}^2$, then
	\[
	\sum_{\ell \text{ is a grid line}} |\is(\mathcal{S},\ell)| \le \sqrt{2} \cdot \wt(\mathcal{S})~.
	\]
	\infodone[inline]{KF: Added \enquote{infinitesimal neighborhood}; see our structure theorem!}
%	\resolvedtodoNotDoneYet[inline]{KF: Does this lemma also hold in our noncrossing setting? If we can assume that all segments of the solution have integer coordinates, then I agree (especially using the fact that grid lines do not pass through integer coordinates). However, though the terminals might be rounded to integer coordinates, can we guarantee the same about the feasible solutions that we consider? We have their infinitesimal short segments infinitesimal close to grid lines (at portals). So, I don't see it yet.}{it will certainly work if we relax the condition that all endpoints are in $\mathbb{Z}^2$ to the condition that all endpoints are in the infinitesimal neighborhood of $\mathbb{Z}^2$.}
\end{lemma} 

A \emph{portal} is an infinitesimal short subsegment; in later sections, we define restricted type of solutions that cross the grid lines only through well-defined portals.
For a segment $S$, we define~$\grid(S,m)$ as the set of $m$ equispaced portals lying on~$S$ (subsets of~$S$) with the first and last portal lying infinitesimally close the endpoints of~$S$ %without containing them 
(thus, the distance between consecutive portals is bounded by~$|S|/(m-1)$).
\resolvedtodoNotDoneYet{KF: In our proofs,~$m$ is not necessarily integer.}{We should say that we round it up. But recheck if~$m$ is really not an integer. Maybe by our changes in the oter section,~$m$ is always an integer now.}
\resolvedtodoNotDoneYet{KF: 1) Neighboring boundaries share portals which causes extra care. 2) Patching a boundary might cause new crossings at another boundary.}{Shift the first and last points such that they lie in the interior of the segment. See comment in other section why this helps.} 
%equispaced
%points on $S$ (thus, the distance between consecutive points is $|S|/(m-1)$).
A border edge of a cell from~$D(\mathbf{a})$ is called a \emph{boundary} if there is no longer border edge on the grid line containing it. 
\importanttodo{KF: Instead, make the border edges disjoint in the same flavor. We need the disjointness of the border edges for the DP.}
Due to a technical subtelty, we treat most boundaries as intervals that are open at one of their endpoints: if a point is contained in two boundaries, then we remove that point from the point set of the shorter of the two boundaries, or from an arbitrary one of the two if both have equal length; note that the removed point was always an endpoint of the affected boundary.
If~$F$ is a boundary, then its length is $2L/2^i$~for some integer~$i$, and we define the \emph{level} of the gridline containing~$F$ as~$i$.
(\todo{KF: Do we need this?} Gridlines not containing any boundaries, that is, not intersecting~$C(\mathbf{a})$, have level~$\infty$.)
%The \emph{level} of the gridline containing~$F$ is the integer~$i$ for which~$F$ has length~$2L/2^i$.
%the smallest
%integer $i$ 
%Every boundary~$F$ has the length~$2L/2^i$ for some integer~$i$, and we call~$i$ the \emph{level} of the grid line containing~$F$. 
%For a grid line $\ell$ intersecting~$C(\mathbf{a})$, we define the \emph{level of $\ell$} to be 
%the integer~$i$ such that~$2L/2^i$ is the length of a boundary contained in~$\ell$.
%the smallest
%integer $i$ 
%such that~$\ell$ contains a border edge \todo{KF: (Maybe) use consistently side, border side, or border edge!}\importanttodo{KF: Define boundary of a cell to be the edge that belongs to a segment that partitioned the parent cell! Or say it every time.}  of a cell from~$D(\mathbf{a})$ of length~$2L/2^i$.
%such that~$D(\mathbf{a})$ contains a cell with sidesof length
%$2L/2^i$, one of whose border edges is contained in $\ell$. 
Intuitively, there are more grid lines with higher level than lower level; this fact is mirrored in the following lemma.
%
%used The following lemma show that grid lines with small levels the level of a fixed grid line is small with small probability.
%
%a randomly chosen grid line has a high level, that is, it contains the border edges of only small cells\note{KF: I think this dos not immediately follow as the lemma just says \enquote{at most}. To see it, one has to argue about the sum of the exact probabilities summing up to~$1$ in order to infer that the upper bound on the low-level probability implies a lower bound on the high-level probability.}.
%
%\begin{lemma}[{\cite[Lemma 19.4.3]{geomspannet}}]\label{lem:levelprob}
%	If $\ell$ is a grid line chosen uniformly at random from the set of grid lines intersecting~$C(\mathbf{a})$ and $i$ is an integer 
%	satisfying~${0 \leq i \le 1+ \log L}$, then the probability that the level of $\ell$ is equal to $i$ is at most $2^{i-1}/L$.
%	\todo[inline]{KF: In the cited book, we argue about a randomly chosen shift vector and not grid line. But this should be equivalent, I guess.}
%\end{lemma}
%
%
\begin{lemma}[{\cite[Lemma 19.4.3]{geomspannet}}]\label{lem:levelprob}
%	If~$\mathbf{a}$
	Let~$\ell$ be a grid line and let~$i$ be an integer satisfying~${0 \leq i \le 1+ \log L}$.
	The probability that the level of $\ell$ is equal to $i$ is at most $2^i/L$.
	\infodone[inline]{KF: I use a different probability than in the reference~\cite{geomspannet}! In the reference, the probability is~$2^{i-1}/L$ which is (I think) only true for level~$i\ge 1$ even though they claim that it holds for~$i=0$. However, I got a different result. According to my calculations, even~$2^{i-2}/L$ should be a valid upper bound for~$i\ge 1$. However, for~$i=0$, the probability is twice as for~$i=1$ since we have twice as many grid lines of level~$0$ than~$1$. Thus, the smallest (simple enough) formula that is true for all~$i\ge 0$ is as I wrote it.\\Also note that in the structure theorem we actually already use the probability that I have written.}
\end{lemma}

\section{Patching Procedure}
\label{sec:patching}
\importanttodo[inline]{KF: We need to extend the lemma in the sense, that (by the same costs), we can place the ten crossings points (in the given order) arbitrarily along the segment, or at least, that we can place all the ten crossing points on any given point of the segment (using Arora's definition of noncrossingness). Alternatively, we could say that we place it infinitesimally close to any given point. Then we probably don't need to use Arora's definition. We just say, crossing at portals means, crossing in an infinitesimal neighborhood close to the portals. We could define infinitesimal by some arbitrary small function on the input size. AZP: I do not think there is time for drastic change of definitions.}

\importanttodo[inline]{KF: We should make it clear whether in this section we look at noncrossing solutions in Arora's sense or in our disjoint sense. The order relation~$<$ suggest that we mean our disjoint sense. Arora's sense seems to appear only in context of portals. AZP: the idea was that if the curves touch, we just say that they are infinitesimal close to each other, but we consider them disjoint. The same with the curve and the segment - as long as they touch but do not cross we consider them disjoint but infinitesimally close. This implies that the crossing points are really the points and not the segments. I agree that this is sloppy but we inherit this sloppyness from Arora.}
In this section, we prove the following important generalization of Arora's patching lemma~\cite{Arora98,Arora2007} (see Lemma~\ref{lem:patch} in Section~\ref{sec:prelim}) to noncrossing tours. 
The generalized patching lemma is a key
ingredient for proving  our approximation schemes for the Euclidean \BTSP{} problem (Section~\ref{sec:noncrossing-tsp}) and the
Euclidean \RBGS{} problem (Section~\ref{sec:rgbseparation}).
It allows us to reduce the number of times that a tour crosses a cell of
the quadtree. 

\begin{lemma}[Patching of noncrossing tours] \label{lem:patch2}
Let $\segm$ be a line segment, and let $\tour{\red}$ and $\tour{\blue}$ be two
simple noncrossing closed curves. There exist two simple noncrossing curves
$\touralt{\red}$ and $\touralt{\blue}$ such that $|\is(\touralt{\red},\segm)
\cup \is(\touralt{\blue},\segm)| \leq 10$, and
$\wt(\touralt{\red})+\wt(\touralt{\blue}) \leq
\wt(\tour{\red})+\wt(\tour{\blue})+ 20 \cdot \wt(\segm)$ . Moreover, for $c \in \{\red,\blue\}$, $\tour{c}$ and
$\touralt{c}$ differ only within an infinitesimal neighborhood of $\segm$. 
%Furthermore, if there is also a portal~$p$ given on~$\segm$, then we can guarantee $\is(\touralt{\red},\segm)
%\cup \is(\touralt{\blue},\segm) \subset p$ with the total cost  
%in which~$\tour{c}'$ intersects any line from a given finite set of lines perpendicular to~$\segm$ only~$O(1)$ times.
Furthermore, we can move all crossing points~$\is(\touralt{\red},\segm)
\cup \is(\touralt{\blue},\segm)$ to any given portal through~$\segm$ by increasing the cost of $\touralt{\red}$ and $\touralt{\blue}$ by only $\Oh(\wt(\segm))$.
\importantquestion[inline]{KF: Can we assume this without cost increase? AZP: probably some already resolved TODO.}
\end{lemma}
\importanttodo[inline]{KF: Say somewhere that if the input tour contains a subsegment of $\segm$, then we can locally modify that segment without notable cost increase such that it crosses~$\segm$ only in one point. In sensitive patching, we might run in such situations, therefore we can't just assume that this will not happen and thus this is important to write. AZP: we are actually assuming this does not happen by assuming infinitesimally close neighborhooods. If we want to explain it in more detail, we should do it when introducing the set I(S,tour). Right now we kind of just blame the confusion on Arora :)}

For a simple curve $\pi$ that contains two points $y,z$, 
we define $\rest{\pi}{y}{z}$ as the part of $\pi$ that goes from~$y$ to~$z$ and whose direction is counterclockwise in the case that~$\pi$ is closed. 
For the purpose of this proof, we assume  without loss of generality that $\segm$
is aligned with the $x$-axis of the coordinate system.  Let
$\is(\tour{\red},\segm) \cup \is(\tour{\blue},\segm) = \{ x_1, \ldots, x_m \}$
be the set of points in $\Reals^2$ through which $\tour{\red}$ and $\tour{\blue}$ intersect~$\segm$ \todo{KF: Is this a multi set or not? AZP: it is NOT a multiset}
\importanttodo[inline]{KF: What if they intersect with an segment, hence, have infinite many intersection points? We don't disallow such cases. Before, we said that we look at crossings points, hence, in such cases such an intersection segment corresponds to a single crossing point. AZP: that is true, the intersection segment gives only one crossing point (if the curves really cross there). Again, it this is not clear, it should be explained when defining I(segment,tour).}
where $x_1 < \dots < x_m$ \todo{KF:Since~$S$ is aligned with the~$x$-axis, we could also define~$x_1,\dots$ as the~$x$-coordinates of these crossings. See also below where we consider~$y_1,\dots$} is their order by the $x$-coordinate.  
Let $\col{x_i}=\red$ if $x_i$ is the intersection of~$\segm$ with
$\tour{\red}$ and $\col{x_i}=\blue$ otherwise. We sometimes\todo{KF: really? AZP: really, for instance right below in the claim.} refer to $\col{x_i}$
as the color of $x_i$.

\begin{claim}\label{clm:twins}
    For $1 < i < m$, it holds that $x_i$ has a neighbor of the same color, that is,
    $\col{x_{i+1}}=c(x_i)$ or $\col{x_{i-1}}=c(x_i)$.   
\end{claim}
\begin{proof}
    Let us assume for the sake of contradiction that $\col{x_i}=\red$ and
    $\col{x_{i+1}}=\col{x_{i-1}}=\blue$ (the case when $\col{x_i}=\blue$ is
    symmetrical). The segment $\rest{\segm}{x_{i-1}}{x_{i+1}}$ and  
    the curve~$\rest{\tour{\blue}}{x_{i-1}}{x_{i+1}}$ together form a closed curve that is crossed by
    $\tour{\red}$ \importanttodo{KF: by~$\tour{\red}$!??} exactly once. This is not possible
    \todo{KF: We have to emphasize in the definition that crossing points are not touching points! AZP: if we really want to do this, the place for it is much earlier} because $\tour{\red}$ \importanttodo{KF: by~$\tour{\red}$!??} is also
    a closed curve. 
    \reviewertodo{I though this was part of the proof of Lemma~\ref{lem:patch2}; since you refer to adjusting a single curve, I suspect you mean to refer to Aurora’s patching in Lemma (4?). AZP: Skipping this comment.}
    \reviewertodo{"In many places there is too much detail (In an abstract, you don’t need to formally prove Claim \ref{clm:twins}, for example, just note that every other segment is inside one of the curves.)". AZP: skipping this comment.}
\end{proof}

\paragraph*{Simplification}

We now group the intersection points $\{ x_1, \ldots x_m \}$ \todo{KF: Shorten all lists of type~$1,2, \dots, n$ to $1, dots, n$.} into maximal
groups of consecutive monochromatic points. To be more precise, if
$\col{x_{i+1}}=\col{x_i}$, then $x_{i+1}$ and $x_i$ are greedily selected to the
same group. Let $H_1,\ldots , H_\ell$ be the resulting (nonempty) groups and 
let~$\segm_j$, for~$1 \le j  \le \ell$, be the segment spanning~$H_j$.
Without loss of generality, assume that the points of $H_j$ are colored $\red$ \todo[inline]{KF: In this section we use $\red$ to refer to "red". Maybe we should (if not already) do the same consistently in other sections?} 
for odd $j$ and colored $\blue$ for even $j$. Observe that the segments
in $\{\segm_j \mid 1 \le j \le \ell\}$ are pairwise disjoint. We now use
the patching procedure of Arora (Lemma~\ref{lem:patch}) 
\soda{It would be helpful to state Lemma 5.1 of Arora.}
 to modify the tours
$\tour{\red}$ and $\tour{\blue}$ into $\tourim{\red}$ and $\tourim{\blue}$. To
be more precise\todo{KF: repetition: to be more precise}, we first modify the tour $\tour{\red}$ by applying
Lemma~\ref{lem:patch} independently to each segment $\segm_j$ where $j$ is
odd and $|H_j|>2$. Analogously, we modify the tour~$\tour{\blue}$ on segments with even index. %\todo{KF: also for odd~$j$? From the context later it seems that we do it for even~$j$. Hence, it is very important to say how we do!}. 
By Lemma~\ref{lem:patch}, the tour $\tourim{\red}$ intersects each segment $\segm_j$ of odd $j$ at most twice (and at least once because the
groups $H_1, \dots, H_\ell$ are nonempty). Similarly, the tour $\tourim{\blue}$
intersects each even segment~$\segm_j$ at most twice. Moreover, 
$\wt(\tourim{c}) \leq \wt(\tour{c})+3\cdot \wt(\segm)$ for~$c \in
\{\red,\blue\}$.

To summarize, we simplified our problem as follows.  
Let
$\is(\tourim{\red},\segm) \cup \is(\tourim{\blue},\segm) = \{ y_1, \ldots,
y_{m'} \}$ be the set of points in $\Reals^2$ through which $\tourim{\red}$ and
$\tourim{\blue}$ intersect $\segm$
where $y_1 < \ldots < y_{m'}$  is their order by the $x$-coordinate. 
By the discussion above and by Claim~\ref{clm:twins}, 
\soda{Also, using this lemma [Lemma 5.1 of Arora.] the crossing is at most twice. In the next paragraph, how did it change to exactly two? It is correct but needs an explanation.}
the color sequence~$\col{y_1}, \ldots, \col{y_{m'}}$ does not contain any monochromatic triplets, 
and, with a possible exception of the first and last element, it consists of alternating pairs~$\red\red$ and~$\blue\blue$.
In other words, $\col{y_1},\ldots, \col{y_{m'}}$ is an infix of a sufficiently long \infodone{KF: I said sufficiently long instead of $m'+2$ since then it looks less complicated and we don't need our assumption anymore that the sequence starts at~$\red$.} 
sequence~$\red,\!\red,\, \blue,\!\blue, \ldots$ of alternating pairs. % $\red\red$ and $\blue\blue$. 
We continue the proof on this simplified instance.
% In the remainder of the proof we often refer to closed curves as tours.

\paragraph*{Laminar and Parallel Tours}
\importantquestion[inline]{KF: We never seem to directly use the properties of being laminar/parallel in the following proof.s (We only prove that something is laminar/parallel. But why are we proving it?)}
Observe that, essentially, there are only two different topologies that a pair of
noncrossing tours can admit. We say that a tour $\sigma$ is \emph{parallel} to
\importantquestion[inline]{KF: Definition of laminar/parallel tours seems buggy: a pair of laminar tours is always also a pair of parallel tours: one tour lies always outside of the other one!! I think, we mean that two tours are parallel if each lies outside of the other one, and laminar, if one lies inside the other one!\\
If we define tours (not pairs) to be laminar/parallel with respect to some other tour, then the definition might work, but then we have also to define what we mean by a laminar/parallel pair of tours; see Fig.~\ref{fig:split}}
a tour $\mu$ if $\mu$ lies outside of $\sigma$, otherwise $\sigma$ is
\emph{laminar} to $\mu$. %Note that (excluding the endpoints) the segment 
\importanttodo{KF: The segment~$\segmalt$ hasn't been defined yet. It's name is also confusing with the color~$\red$! Do we mean $S$? AZP: I am removing the mention of R here.} 
\reviewertodo{341: The segment R used here is not defined.  You’ve used colors R,B for points above, and fig 6 uses S.  (see below: 354: only now do we get the definition of R.  )}
%$\segmalt$ lies entirely inside of $\sigma$ if $\sigma$ is laminar to $\mu$, and entirely outside otherwise. 
For an illustration, see Fig.~\ref{fig:split}.
%to the left for a laminar pair, and to the right for a parallel pair. %KF: description iwht left and right is confusing as we seemignly for the reader refer to a single figure.
\importanttodo[inline]{KF: Instead of defining the property of laminar/parallel tours, we could make the following observations: AZP: if i understand correctly, we are leaving parallel laminar.
\\1) Each tour separates the plane into two parts, a finite inner part and an unbounded outer part.
\\2) When walking around the tour in counterclockwise direction, the inner part is thus always on the left-hand side.
\\3) We orient both tours in counterclockwise direction. 
\\4) Any two neighboring passes of the same color (independently of whether there is another color in between(!)) must be oriented in opposite direction given 2) and 3) (either the inner or the outer part lies in between them).
\\5) Thus, the~$i$-th and~$j$-th pass (of the same color) are oriented in opposite direction if and only if~$i$ and~$j$ have different parity.
\\6) Thus, by connecting passes of different parity, we again obtain cycles.
\\7) Consequently, the splitting step is fully safe!
\\8) The merging step also as there we greedily connect passes of different parity, ensure that we have degree 2 everywhere, and that every color is connected. We must only argue that if and only if we can connect one color, we can also connect the other one. 
I think we can show this iff as follows: each split increases the number of red and blue components by one (above and below). Hence, after the split step, we have the same number of red and blue cycles. If we merge two blue cycles, then we also merge the two red cycles in between as they were disjoint before given the disjointness of the blue cycles. Thus, if we merge, then we decrease the number of blue and red components by one (above or below, respectively). Consequently, when we merged all blue cycles, we necessarily also merged all red cycles. Hence, we end up with exactly one blue and one red tour.
%However, after more thinking, we actually might need laminar/parallel properties to show this iff.
}

\paragraph*{Split Operation}

Next, we define the operation of \emph{splitting} a pair of non-intersecting \todo{KF: noncrossing? disjoint? check also other places.}
tours by a segment (for an intuitive illustration see Fig.~\ref{fig:split}).
This operation is defined only for a segment that starts and ends in two
distinct points of the first tour, intersects the second tour twice and has no
other intersections with the two tours. To be more precise, let $\sigma$ and
$\mu$ be two non-intersecting tours. Let $\segmalt$ \todo{KF: Why don't we use~$S$. Is~$\segmalt$ as subsegment of~$S$? AZP: yes.} be a segment that connects
two different points on $\sigma$ and intersects both $\mu$ and $\sigma$ exactly
twice (so the intersections with $\sigma$ are precisely the endpoints of
$\segmalt$). For the purpose of this proof, the split operation will only be used for $\segmalt$ being a subsegment of $\segm$, but we define the split operation independently of $\segm$. 

% \begin{figure}
% \includegraphics[scale=0.8]{splitting.eps}
% \includegraphics[scale=0.8]{splittingOut.eps}
% \caption{Splitting laminar and parallel tours}\label{fig:split}
% \end{figure}

\begin{figure*}[ht!]
    \centering
    \begin{subfigure}[t]{0.4\textwidth}
        \centering
        \includegraphics[width=\textwidth]{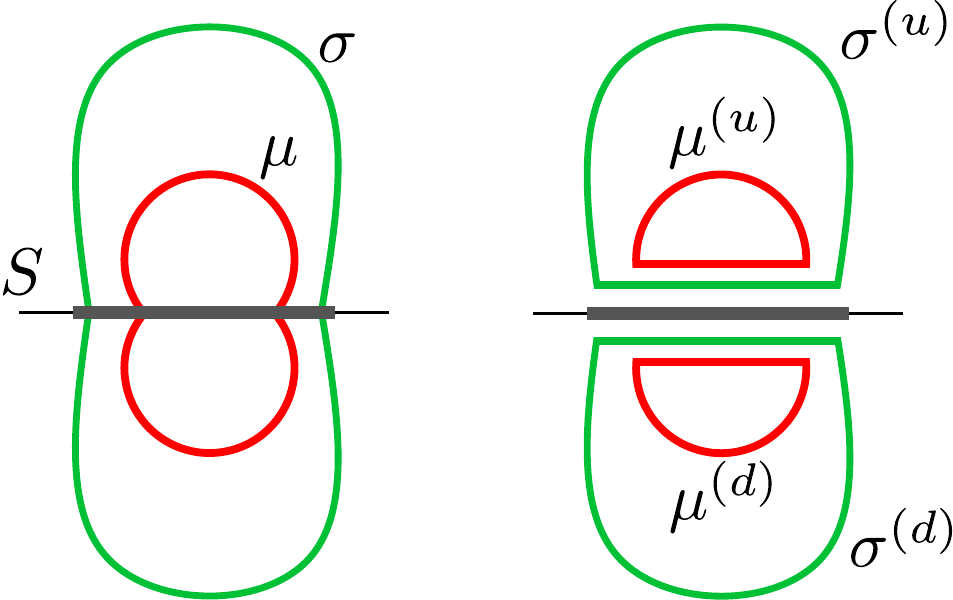}
    \end{subfigure}%
    ~ 
	\hspace{0.5cm}
    \begin{subfigure}[t]{0.4\textwidth}
        \centering
        \includegraphics[width=\textwidth]{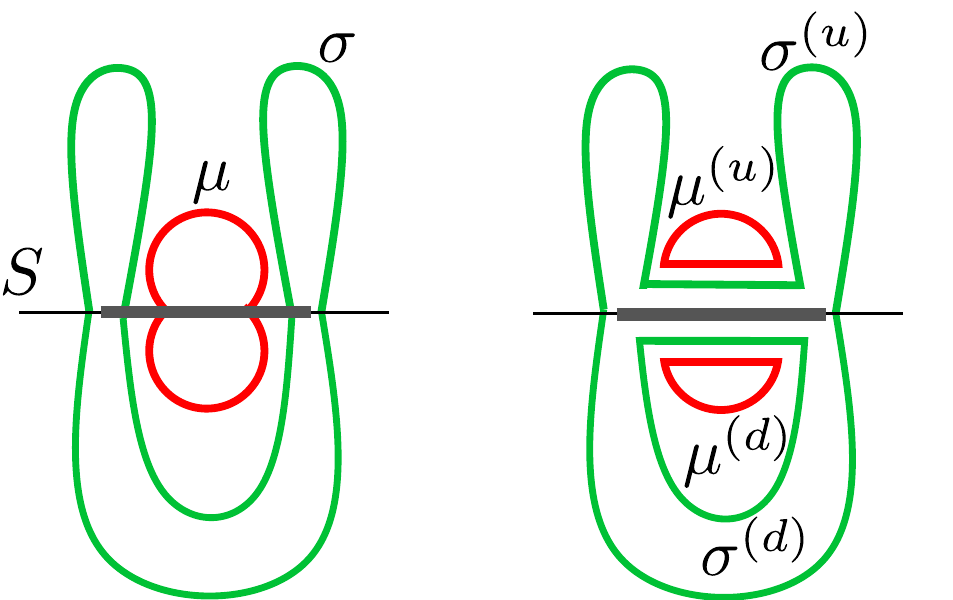}
    \end{subfigure}%
    \caption{A schematic view of the splitting procedure for laminar (left figure) and parallel (right figure) tours.
	\todo[inline]{KF: Inconsistency: in all figures we use green color whereas in the main text we write $B$ for blue. Since red-green blindness is the most common one (followed by blue-yellow), I suggest to use red and blue.}
	\todo[inline]{KF: Inconsistency2: In the figure we use~$S$, in the text we use~$R$.}
	\importanttodo[inline]{KF: Missing case (see main text): same crossing as in the first case, but red contains green.}
	\reviewertodo{Ditto for fig 6 caption: define S and explain the lighter line.  Define what distinguishes parallel and laminar. }}
	\reviewertodo{ Please use the notation of 356ff in Fig 6.}
\label{fig:split}
\end{figure*}

Let $p_1,p_2$ be the intersection points of~$\sigma$ with~$\segmalt$ (hence, the endpoints of~$\segmalt$) and 
let $q_1,q_2$ be the intersection points of~$\mu$ with~$\segmalt$. 
Without loss of generality, we assume that $\segmalt$ is aligned
with the $x$-axis and~$p_1 < q_1 < q_2 < p_2$ is the order of the points by the $x$-coordinate. \todo{KF: how to write: $\dots < \dots$ is their order by the~$x$-coordinate? check also other places!} Note that (excluding the endpoints) the segment $\segmalt$ lies entirely inside of $\sigma$ if $\sigma$ is laminar to $\mu$, and entirely outside otherwise (see Fig.~\ref{fig:split}). 
We split $\sigma$ and $\mu$ as follows. We create
two copies of~$\segmalt$, let us call them $\segmalt^{(u)}$ and
$\segmalt^{(d)}$, 
\soda{The notation $\sigma^{(u)}$ and $\sigma^{(d)}$ are a bit confusing, as they do not entirely lie in the up or down of S.}
and place them slightly above and slightly below $\segmalt$
respectively\footnote{The distance between $\segmalt^{(d)}$ and $\segmalt^{(u)}$
    is infinitesimally small (as in the original patching lemma of
Arora~\cite{Arora98,Arora2007}).}. Next, we split $\sigma$ into two tours, $\sigma^{(u)}$ and
$\sigma^{(d)}$, by first adding to~$\sigma$ the two segments $\segmalt^{(u)}$ and $\segmalt^{(d)}$
and then removing the parts of $\sigma$ that connect the endpoints of 
\importantquestion[inline]{KF: But then we remove the tour completely as a cycle always has two parts connecting the same pair of points! (In our case the situation is even more complicated as we have several cycles after adding the two segments.) Formally, we could write that we first remove from~$\sigma$ its four (or two if~$\segmalt$ is a point?) intersection points with $\segmalt^{(u)}$ and $\segmalt^{(d)}$, thus~$\sigma$ consists of four open curves. Remove the two (?) curves that cross/intersect~$\segmalt$. Then close one of the remaining paths with $\segmalt^{(u)}$, the other one with $\segmalt^{(d)}$ to obtain $\sigma^{(u)}$ and~$\sigma^{(d)}$.}
$\segmalt^{(d)}$ with the endpoints of $\segmalt^{(u)}$. 
Then, we create two
copies of the segment $\segmalt'=\rest{\segmalt}{q_1}{q_2}$, let us call them
$\segmalt'^{(u)}$ and $\segmalt'^{(d)}$, and we place $\segmalt'^{(u)}$ slightly
above $\segmalt^{(u)}$, and we place 
$\segmalt'^{(d)}$ slightly below
$\segmalt^{(d)}$. 
We split $\mu$ with the
segments $\segmalt'^{(u)}$ and $\segmalt'^{(d)}$ in an analogous way as we split~$\sigma$. 
The split operation is
presented in Fig.~\ref{fig:split} for a laminar pair (left side) and for a
parallel pair (right side). 

% We sometimes refer to splitting a parralel pair as an external split, while we
% refer to splitting a laminar pair as an internal split. 

As a result of a split operation we obtain two pairs of closed curves: above \importanttodo{KF: But on the figure, $\sigma^{(u)}$ is everywhere (above \emph{and} below) at the same time! We should state more precisely what we mean by above/below: that their endpoints are above/below.} $\segmalt$ we
have a pair $\sigma^{(u)},\mu^{(u)}$, and below $\segmalt$ we have a pair
$\sigma^{(d)},\mu^{(d)}$. Observe that if the original tours were laminar, then both of these pairs are
laminar and $\sigma^{(u)}$ is parallel to $\sigma^{(d)}$\importanttodo{KF: Not true! we can redraw Fig.~\ref{fig:split}a) such that the red tour goes around the green one (but intersect~$S$ as in the figure. After the split, $\sigma^{(u)}$ is NOT parallel to $\sigma^{(d)}$.}. On the
other hand, if the original tours were parallel, then one pair is parallel, the
other is laminar and $\sigma^{(u)}$ is laminar to $\sigma^{(d)}$ (as in
Fig.~\ref{fig:split}\todo{KF: In the figure it seems to be the other way around ($\sigma^{(d)}$ is laminar to $\sigma^{(u)}$)!} ) or vice versa\todo{KF: What do we mean by vice versa? I think, I see several possibilities.}.\todo{KF: Check if the text is consistent of our definition of laminar/parallel.}

For both resulting pairs $(\sigma^{(a)},\mu^{(a)})$, $a \in \{ u,d \}$, we refer
to the pair of segments $(\segmalt^{(a)},\segmalt'^{(a)})$ as a \emph{precise
interface} of the pair $(\sigma^{(a)},\mu^{(a)})$, while the segment $\segmalt$ is a
\emph{rough interface} of $(\sigma^{(a)},\mu^{(a)})$. Note that the precise
interface is in a negligible distance to the corresponding rough interface.
Also observe that $\wt(\sigma^{(u)})+\wt(\sigma^{(d)})+\wt(\mu^{(u)})+\wt(\mu^{(d)}) \leq
\wt(\sigma)+\wt(\mu)+4 \wt(\segmalt)$.

\paragraph*{Splitting $\tourim{\red}$ and $\tourim{\blue}$}

\todo{KF: consistency: laminar/parallel}

Our goal is to transform $\tourim{\red}$ and $\tourim{\blue}$ as to make them
intersect $\segm$ at most ten \todo{KF: check if ten} times in total. 
Without loss of generality we
assume that $\tourim{\red}$ is either laminar to $\tourim{\blue}$ or it is
parallel to $\tourim{\blue}$ (if neither is the case\importantquestion{KF: Can it be neither case??}, we swap the names of the
two tours). 

Let $\segmalt_1, \ldots, \segmalt_k$ be segments that span all quadruples $y_i,
y_{i+1}, y_{i+2}, y_{i+3} \in \segm$ such that $\col{y_i}=\col{y_{i+3}}=\red$
and $\col{y_{i+1}}=\col{y_{i+2}}=\blue$ (see Fig.~\ref{fig:before}). We assume
that $\segmalt_1 < \ldots < \segmalt_k$ are ordered \todo{KF: right usage of ordered by?} 
by the $x$-coordinates of
their left endpoints.  Observe that \todo{KF: Also refer to the fact that the sequence consists of an infix of alternating pairs\dots} 
there are at most six \todo{KF: Explain why. Add something like this: at most three before and three after the segments.}\todo{KF: Can't we easily decrease six to four (or even two??) by recoloring (if we can't, shouldn't we say so?).}
intersections \todo{KF: crossing points?} among
$y_1, \ldots, y_{m'}$ which do not belong to the segments $\segmalt_1,\dots,\segmalt_k$. 
As the first step of our transformation, we split the pair of
tours $(\tourim{\red},\tourim{\blue})$ with the segments $\segmalt_1, \ldots,
\segmalt_k$ one by one \todo{KF: in any order, e.g. from left to right?} \todo{KF: each time splitting the tour pair (resulting from the previous splits) that crosses the segment} \importanttodo[inline]{KF: It is not obvious that such a splitting will work. After each single split, we need to guarantee that we don't come up with the situation that, for some unprocessed segment, two crossing points of the same color suddenly belong to different cycles.}
(Fig.~\ref{fig:before} shows an example of
$\tourim{\red}$ and $\tourim{\blue}$ before and after the split). Let $\{
(\sigma_i,\mu_i) \}_{i \in \{ 1, \ldots, k+1 \}}$ be the set of pairs
obtained after all splits are completed. 
\importanttodo[inline]{KF: Given our mistakes above, I don't have confidence that the 
following claims (until end of this paragraph) are correct. Maybe they are, but in any case we should prove them and not just write \enquote{Note that it holds.}}
Note that if $\tourim{\red}$ is
laminar to $\tourim{\blue}$, then, for every pair~$(\sigma_i,\mu_i)$, %$i$, % for~$1\le i \le k+1$, 
the tour~$\sigma_i$ is
laminar to $\mu_i$ and, for $j \neq i$, parallel to $\sigma_j$. If, on the
other hand, $\tourim{\red}$ is parallel to $\tourim{\blue}$, then among the
obtained pairs, there is precisely one pair $(\sigma_{j},\mu_j)$ where
$\sigma_j$ is parallel to $\mu_j$; hence, for $i \neq j$, the tour~$\sigma_i$ is
laminar to $\mu_i$. Moreover, for~$i \neq j$, the tour~$\sigma_j$ is 
laminar to~$\sigma_i$ and, hence, also laminar to~$\mu_i$.

% \begin{figure}[ht!]
%     \centering
%     \includegraphics[width=0.9\linewidth]{img/beforesplit.pdf}
%     \caption{Before and after splitting}%
%     \label{fig:before}
% \end{figure}

\begin{figure}[ht!]
\   \includegraphics[width=0.47\textwidth]{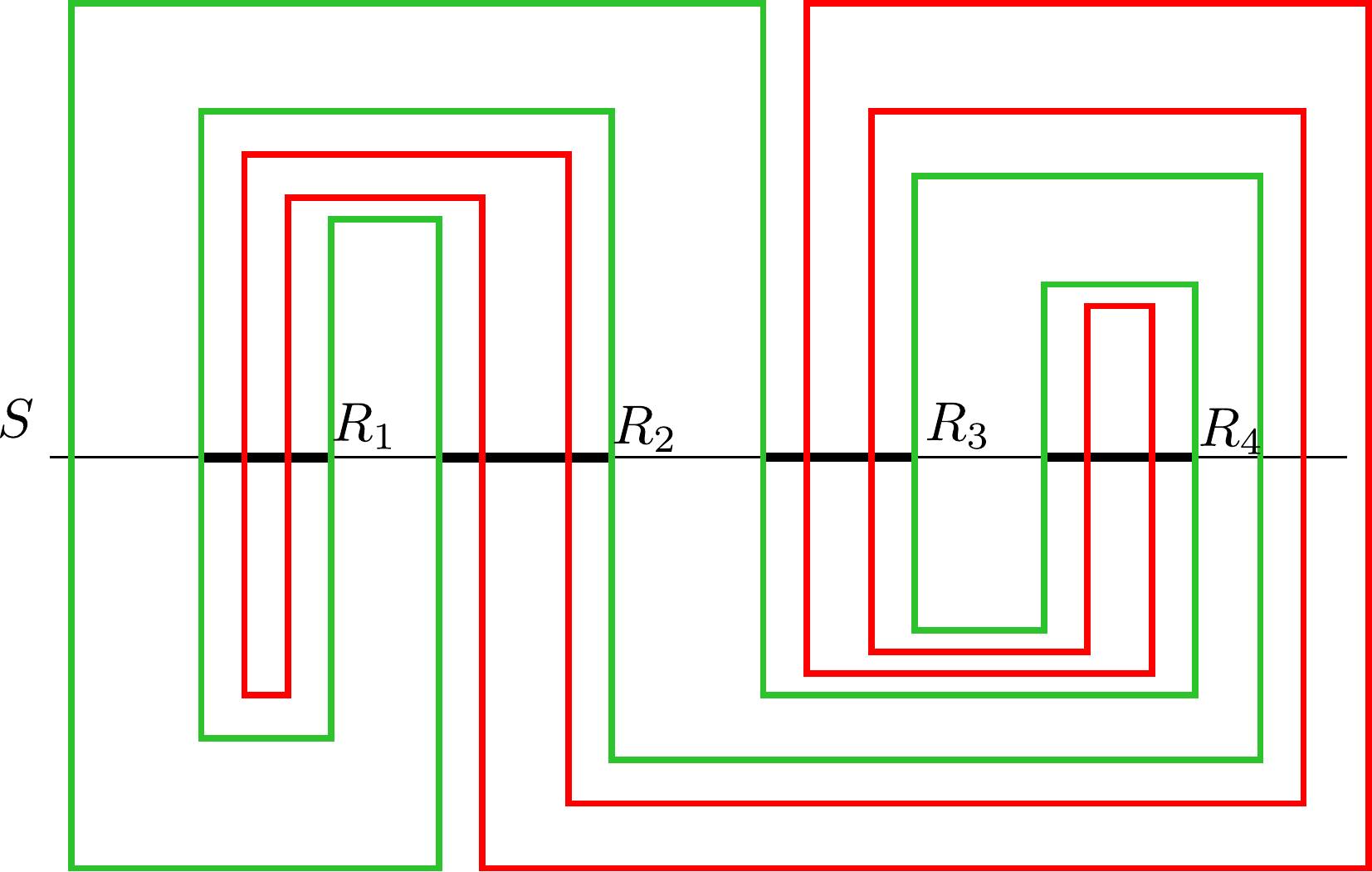} 
    ~ 
	\hspace{0.5cm}
\includegraphics[width=0.47\textwidth]{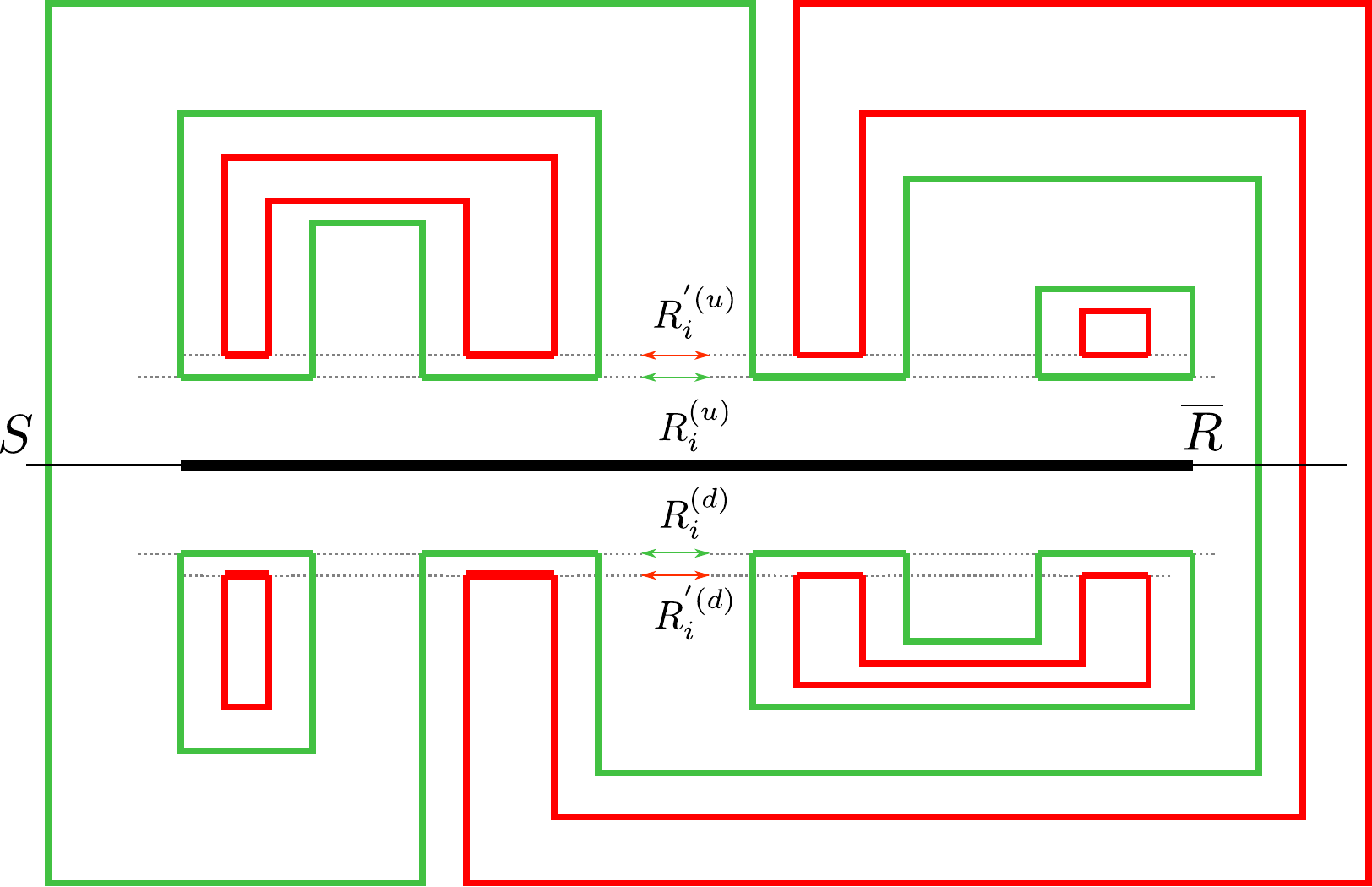} 
\caption{Before and after splitting\todo[inline]{KF: Colors in figure inconsistent with definition: a quadruple should start/end in red and not in green (blue).}\todo[inline]{KF: Some labels too small, whereas some labels in some other figures are too big. In the ideal case, all labels should have over all figures have the same size as the main text.}}\label{fig:before}
\end{figure}

Let $\overline{\segmalt}$ be the minimum segment that contains all segments
$\segmalt_i$ for $i \in \{1, \ldots, k\}$.  It is important to observe that,
after the splitting,~$\overline{\segmalt}$ is not crossed by any of the tours. \todo{KF: Write that now we have a \enquote{corridor} around it.}
\todo[inline]{KF: We should also argue about the total splitting cost somewhere as we refer to it at the final cost analysis.}

\paragraph*{Merging Precise Interfaces}
From now on, the goal of the transformation is to merge the pairs resulting from the splitting in order to obtain again two tours in total. \todo{KF: in total, right?}
% The idea of the transformation from now on is to merge these pairs made by splitting back into two tours. %KF: "splitting back" was confusing in the old sentence
%To achieve that we use the corridor around $\overline{\segmalt}$ that is not crossed by
%the pairs of tours that were created by splitting.

%\todo{KF: Move this paragraph to the next titled paragraph.} 

To formally describe the process of merging, let us first define the operation
of merging consecutive precise interfaces that both lie on the same side of
the segment~$\segm$ (both above or both below). This operation is the reverse of
splitting and it is illustrated in Fig.~\ref{fig:merge}. 

\begin{figure}[ht!]
    \includegraphics[width=\textwidth]{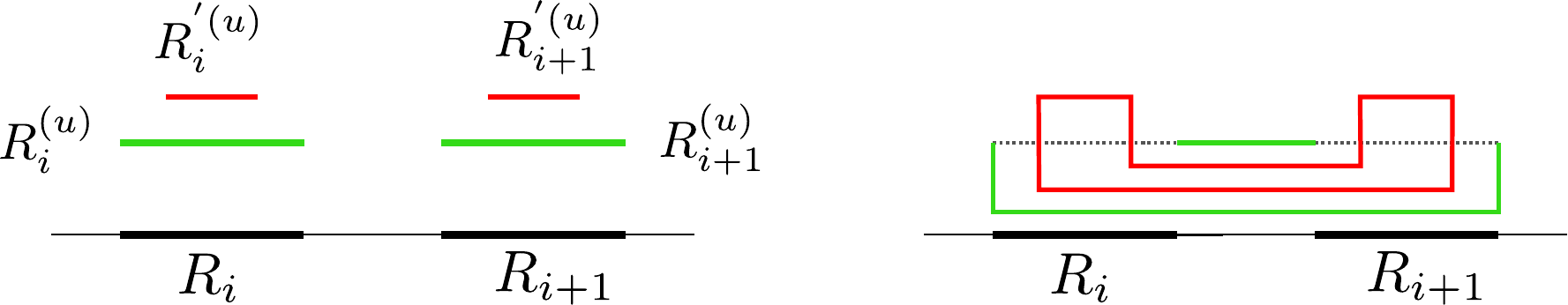}
\caption{%
\todo[inline]{KF: Labels way too big!}
\todo[inline]{KF: We should make visually more clear where the precise interfaces are on the right side (e.g., by increasing the thickness of the dashed part and also replace the $R'$s by the dashed part. The current variant of the figure wrongly suggests that the~$R'$s are still present and have not been remove; which is wrong. The thin red horizontal segments symbolize the red tour above the~$R'$s, but this is not clear).}
Merging two consecutive interfaces when the interfaces are above
$\overline{\segmalt}$. The left figure presents the interfaces before merging. The right
procedure depicts the result after merging. The segments are removed and
connected through the corridor \todo[inline]{KF: depict the corridor as a gray shaded area?} between them and \todo[inline]{KF: not so clear what we mean. In the main text, there are two different types of corridors.} $\overline{\segmalt}$ in the
noncrossing manner.}\label{fig:merge}
\end{figure}

Let $(\segmalt_i^{(a)},{\segmalt'}_{i}^{(a)})$, where $a \in \{ u,d \}$, and
$(\segmalt_{i+1}^{(a)},\segmalt_{i+1}'^{(a)})$ be two consecutive precise
interfaces lying on the same side of $\segm$. Let $(\sigma_i^{(a)},\mu_i^{(a)})$
and $(\sigma_{i+1}^{(a)},\mu_{i+1}^{(a)})$ be two distinct 
\importantquestion{KF: Can it happen that they are the same tours? Yes, it can! The word distinct wants to emphasize it, but this is not strong enough. It even wrongly suggests that the two pairs are always distinct. We should rephrase/add here something. And also, what do we do if the pairs are not distinct. We just keep and do not merge those precise interfaces?} 
pairs of tours whose
interfaces are $(\segmalt_i^{(a)},\segmalt_i'^{(a)})$ and
$(\segmalt_{i+1}^{(a)},\segmalt_{i+1}'^{(a)})$, respectively. 
To merge the interfaces $(\segmalt_i^{(a)},\segmalt_i'^{(a)})$ and
$(\segmalt_{i+1}^{(a)},\segmalt_{i+1}'^{(a)})$, we remove 
the segments~$\segmalt_i^{(a)}$,~$\segmalt_i'^{(a)}$,~$\segmalt_{i+1}^{(a)}$ and~$\segmalt_{i+1}'^{(a)}$ 
from the tours containing them. As a result, all four
tours~$\sigma_i^{(a)}$,~$\mu_i^{(a)}$,~$\sigma_{i+1}^{(a)}$ and~$\mu_{i+1}^{(a)}$
become paths. We connect the endpoints of~$\sigma_i^{(a)}$ with the endpoints of~$\sigma_{i+1}^{(a)}$ 
\importantquestion{KF: Very unclear: which endpoints are connected together? $\sigma_i^{(a)}$ has two endpoints and $\sigma_{i+1}^{(a)}$ has two endpoints.}
in a noncrossing manner, and this creates a corridor very
close to $\overline{\segmalt}$. We then connect the endpoints of $\mu_i^{(a)}$
with the endpoints of $\mu_{i+1}^{(a)}$\importantquestion{KF: Again, which endpoints are precisely connected together?}, in a noncrossing manner, via the
corridor that was just created. This operation is depicted in
Fig.~\ref{fig:merge}. Thus, by merging precise interfaces, the associated
pairs of tours are also merged into one pair $(\sigma^{(a)},\mu^{(a)})$. 
If both pairs that are merged are laminar, then observe that $\sigma^{(a)}$ \todo{KF: other way around?} is also laminar
to $\mu^{(a)}$. If, on the other hand, one of the pairs is parallel, then the
resulting pair is also parallel. 
\importanttodo{KF: Check correctness of the claims concerning laminar/parallel!}
\todo{KF: check consistency of usage of laminar/parallel, also w.r.t. tours/w.r.t. pairs}

\paragraph*{Merging All Pairs}

\begin{figure}[ht!]
    \centering \includegraphics[scale=0.5]{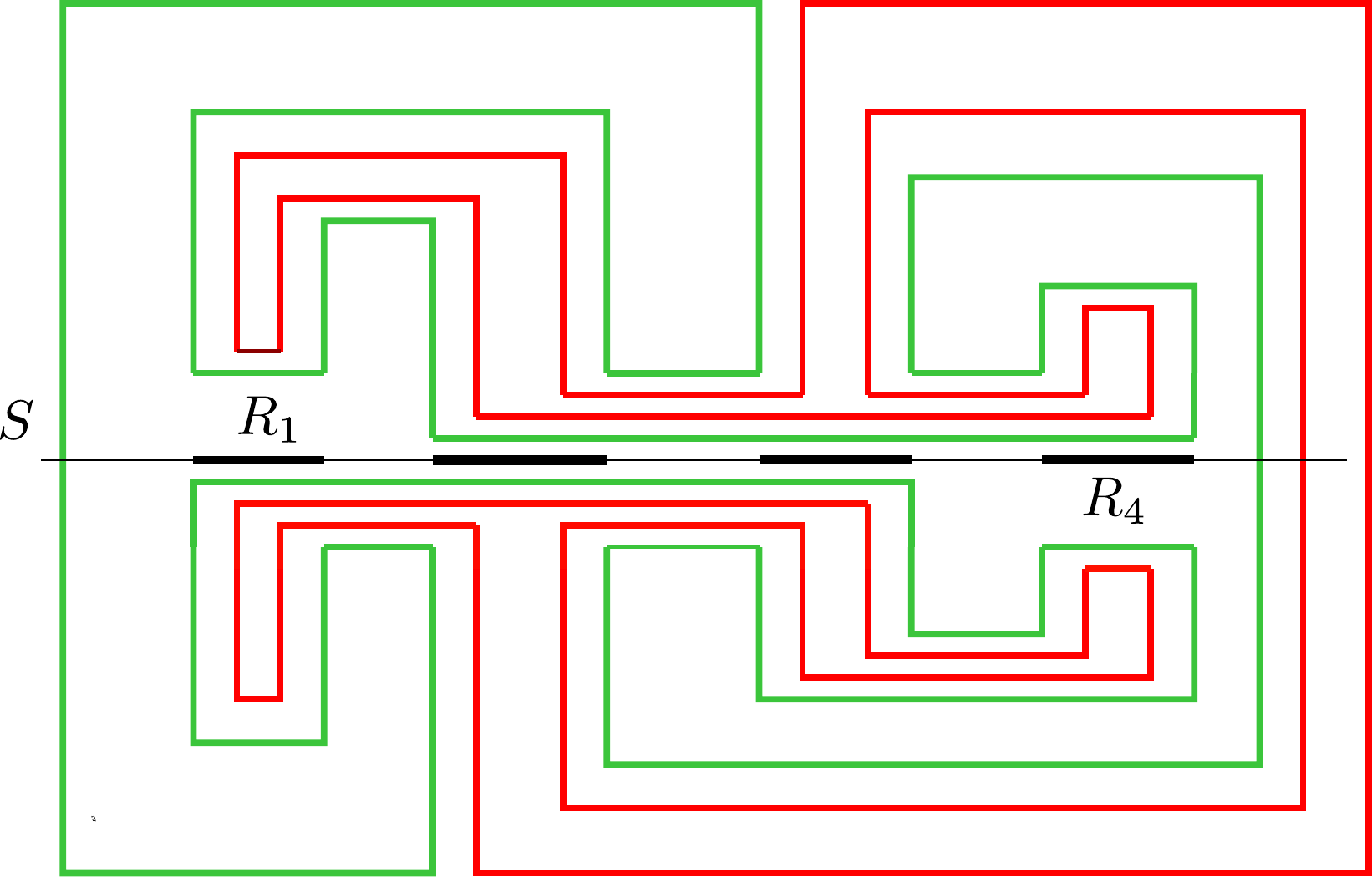}
\caption{Our example after the transformation}\label{fig:aftermerge}
\end{figure}

We now describe the process of merging all pairs of the tours that were created during
the splitting procedure. 
The idea is it to start, by merging pairs of tours whose precise interfaces lie above $\segm$.
Subsequently, we want to merge pairs of tours whose precise interfaces lie below $\segm$. 
After this procedure, we may still obtain
a constant number of distinct pairs of tours, one \todo{KF: one or some? If one, then we shouldn't write constant number but just \emph{at most two}.} above and one below $\segm$.
In that case we merge them by crossing $\segm$. 

In more detail, we first iterate through the rough interfaces~$\segmalt_2, \dots, \segmalt_k$ (skipping $\segmalt_1$). 
For each $\segmalt_i$, we look at the pair~$(\sigma^{(u)}_i,\mu^{(u)}_i)$ whose
precise interface is $(\segmalt_i^{(u)},\segmalt_i'^{(u)})$. If $\sigma^{(u)}_i$ is
not the same tour as $\sigma^{(u)}_{i-1}$\importantquestion{KF: But what if $\mu^{(u)}_i$ equals $\mu^{(u)}_{i+1}$, and what about the case that $\mu^{(u)}_i$ and $\mu^{(u)}_{i+1}$ differ but $\sigma^{(u)}_i$ and $\sigma^{(u)}_{i-1}$ not?}, we merge the two precise interfaces
$(\segmalt_{i-1}^{(u)},\segmalt_{i-1}'^{(u)})$ and
$(\segmalt_i^{(u)},\segmalt_i'^{(u)})$\importantquestion{KF: What happens to their precise interfaces after the merge? They get destroyed, right? But then what happens to the sequence of the precise interfaces? Do we update it? Can it now happen that two precise interfaces become neighbors (consecutive) that weren't consecutive before?}.
Thus, at the end, all the precise interfaces above~$\segm$ belong to the same pair of tours. %, and all interfaces below~$.
In an analogous way, we proceed with the interfaces below $\segm$ and obtain a single pair of tours below~$\segm$.
If both pairs of tours, above and below~$\segm$, are the same, we are done.
Otherwise, if the pairs are distinct\importantquestion{KF: do we include the case that two tours from the pairs are the same while two others are distinct?}, 
we merge\todo{KF: How do we merge?} the interface~$(\segmalt_1^{(u)},\segmalt_1'^{(u)})$
with~$(\segmalt_1^{(d)},\segmalt_1'^{(d)})$ which results in a single pair of tours crossing $\overline{\segmalt}$ four times. 
For an example, see Fig.~\ref{fig:merge} \todo{KF: Fig.~\ref{fig:aftermerge}?} that depicts the situation after merging the interfaces of Fig.~\ref{fig:before}. 
Note that merging all interfaces costs at most $8 \wt(\segm)$\importantquestion{KF: Why factor~$8$? Why a constant? Can't it happen during the merging that repeatedly pairs of large distance that before were not consecutive, become consecutive (after pairs in between have been merged together and their precise interfaces got removed)?}\todo{KF: Does the factor include the cost for splitting and everything before, too?}.

% \begin{figure} \centering \includegraphics[scale=0.4]{img/aftermerging.pdf}
% \caption{Our example after the transformation}\label{fig:aftermerge}
% \end{figure}

Summarized, the resulting tours $\touralt{\red}$ and
$\touralt{\blue}$ cross~$\segm$ at most ten times in total: at most six times \importantquestion{KF: Can we decrease it to four? (See comment in the beginning.)} along~$\segm\setminus\overline{\segmalt}$ as noted above, and at most four times along~$\overline{\segmalt}$ after the merging.
Splitting the tours
incurs an additional cost of at most $8 \wt(\segm)$, while merging the tours
also incurs an additional cost of at most $8 \wt{\segm}$. The initial
transformation of Arora costs at most $3 \wt(S)$. Thus, $\wt(\tour{\red}) +
\wt(\tour{\blue}) \leq \wt(\touralt{\red}) + \wt(\touralt{\blue}) +
20\wt(\segm)$ as promised. 

\section{Euclidean Bicolored Noncrossing Traveling Salesman Tours}
\label{sec:noncrossing-tsp}

In the following, we show that the problem of finding two noncrossing Euclidean traveling salesman tours, one for each color, admits an EPTAS.

\thmPTAStsp* 
%    Euclidean \BTSP{} admits a randomized~$(1+\eps)$-approximation scheme with $2^{\Oh(1/\eps)} n \polylog(n)$ running time. 

In Section~\ref{sec:structureTheorem}, we prove a helpful theorem that allows us to focus on restricted tours, which is a key ingredient for our PTAS in Section~\ref{sec:alg-2noncross}.

\subsection{Our Structure Theorem}
\label{sec:structureTheorem}
The structure theorem shows that tours obeying certain restrictions are not much more expensive than unrestricted ones. 
Kisfaludi{-}Bak et al.~\cite{focs21} called such restricted tours \emph{$r$-simple}.  
Below, we extend this notion to our setting and introduce \emph{$r$-simple pairs of tours}. 
%This extension will ease our presentation when applying the sparsity-sensitive
%patching technique of Kisfaludi{-}Bak et al.~\cite{focs21}. 
We say that a pair of tours $(\tour{1},\tour{2})$ crosses a line segment $\segm$ if $\tour{1}$ crosses $\segm$ or $\tour{2}$ crosses $\segm$. 
In this sense, we define $\is((\tour{1},\tour{2}),\segm)=\is(\tour{1},\segm) \cup \is(\tour{2},\segm)$ as the set of points at which the pair $(\tour{1},\tour{2})$ crosses $\segm$.

\begin{definition}[$r$-simple pair of tours]\label{def:simple}
%\todo[inline]{KF: I think, the definition would be simpler if we merge all the cases together! The DP should also work then. Thus: Pair of tours is~$r$-simple if it uses only portals from a portal set satisfying the condition below and each portal crossed at most ten times. The condition for the set: it contains the two extremal portals of~$F$, one arbitrary portal along the interior of~$F$, and, for~$g\le r^2/m$ where~$m$ is the number of times~$F$ is crossed, the portal set~$\grid(F, g)$.}
%\todo[inline]{KF: I think, the definition would be simpler if we merge both cases together! The DP should also work then. Thus: Pair of tours is~$r$-simple if it uses only portals from a portal set satisfying the condition below and each portal crossed at most ten times. The condition for the set: it contains one arbitrary portal along the interior of~$F$, and, for~$g\le r^2/m$ where~$m$ is the number of times~$F$ is crossed, the portal set~$\grid(F, g)$.}
Let $\mathbf{a}$ be a random shift vector. A pair of tours $(\tour{1},\tour{2})$ is \emph{$r$-simple} if it is noncrossing\kint{defined to be noncrossing. remove noncrossing where we write r-simple.} and, for any boundary~$F$ of the dissection $D(\mathbf{a})$ crossed by the pair, 
%either
	\begin{enumerate}[label=(\alph*)]
		\item\label{case:Arora} it crosses $F$ entirely through one or two portals belonging to~$\grid(F, \lfloor r\log L \rfloor)$, or %(and only the interior of~$F$), or
        \item\label{case:Sensitive} it crosses $F$ entirely through portals belonging to~$\grid(F, g)$, for~$g\le r^2/m$ where~$m$ is the number of times~$F$ is crossed.
	\end{enumerate}
	Moreover, for any portal~$p$ on a grid line $\ell$, the pair~$(\tour{1},\tour{2})$ crosses~$\ell$ at most ten times through~$p$. 
\end{definition}

%Now, we show that the structure theorem \todo{KF: we haven't talked about that theorem yet, I think. Thus, we should say some words about it, for instance, why we want to extend it, why we need it.} of Kisfaludi{-}Bak et al.~\cite{focs21} can be extended to
%noncrossing pairs of tours.
Our structure theorem is an extension of the structure theorem of Kisfaludi{-}Bak et al.~\cite{focs21} to noncrossing pairs of tours.

\begin{theorem}[Structure Theorem]\label{thm:struct}
    Let $\mathbf{a}$ be a random shift vector, 
    %$P_1 \subseteq \{0,\dots,L\}^2$, $P_2
%    \subseteq \{0,\dots,L\}^2$, 
	 and let $(\pi_1,\pi_2)$ be a pair of
    noncrossing tours 
    %such that $\pi_1$ visits $P_1$, $\pi_2$ visits $P_2$, 
    consisting only of line segments whose endpoints lie in an infinitesimal neighborhood of~$\mathbb{Z}^2$. 
    For any large enough integer $r$, there is an $r$-simple pair
    $(\touralt{1},\touralt{2})$ of noncrossing tours %visiting $P_1$ and $P_2$ respectively 
    that differs from~$(\pi_1,\pi_2)$ only in an infinitesimal neighborhood around the grid lines
    and that satisfies 
    \[\mathbb{E}_{\mathbf{a}}[\wt(\touralt{1} \cup \touralt{2})-\wt(\tour{1} \cup \tour{2})] = \Oh((\wt(\tour{1}) +\wt(\tour{2})) / r)~.\]
\end{theorem}

In the remainder of this section, we prove Theorem~\ref{thm:struct}. The proof is based on
the sparsity-sensitive patching technique of Kisfaludi{-}Bak et al.~\cite{focs21} that we discuss here for completeness.
The sparsity-sensitive patching transforms two noncrossing tours $\pi_1,\pi_2$ into an~$r$-simple pair~$(\pi'_1,\pi_2')$ of noncrossing tours
by patching the tours along each boundary~$F$ that does not satisfy the conditions of Definition~\ref{def:simple}. 
As shown later, by considering the boundaries one by one in non-increasing order of their length, % (that is, in increasing order of the level of the grid lines containing them), 
the patching of one boundary will never affect the crossings of boundaries already considered; thus, in the end, all boundaries will satisfy the desired conditions. 

Let~$F$ be the boundary whose turn is now to be patched. 
Let~$(\hat{\tour{1}},\hat{\tour{2}})$ be the pair of tours obtained from~$(\tour{1},\tour{2})$ by the previous patching steps (initially, $(\hat{\tour{1}},\hat{\tour{2}})=(\tour{1},\tour{2})$). 
By choosing the~$x$-axis parallel to~$F$ and orienting it appropriately, we can assume that~$F$ is horizontal with its open endpoint (if any) on the right side (in the direction of increasing~$x$-coordinate).
%As argued below, we can assume that~$\is((\hat{\tour{1}},\hat{\tour{2}}),F)/\is(({\tour{1}},{\tour{2}}),F)$ contains only 
Inductively, we assume that~$\is(({\tour{1}},{\tour{2}}),F) \subseteq \is((\hat{\tour{1}},\hat{\tour{2}}),F)$ and that the remaining crossing points, $H:=\is((\hat{\tour{1}},\hat{\tour{2}}),F)/\is(({\tour{1}},{\tour{2}}),F)$, lie together infinitesimally close to the right endpoint of~$F$ and to the right of all the other crossing points.
Let~$i$ be the level of the grid line containing~$F$.

%First, we group crossing points together that are close to each other.
%We start by considering~$\is(({\tour{1}},{\tour{2}}),F)$. 
First, we partition the crossings in~$\is(({\tour{1}},{\tour{2}}),F)$ (thus, without~$H$) into two sets~$G$ and~$N$. 
If their number, $k=|\is(({\tour{1}},\tour{{2}}),F)|$, is~$0$, we set~$G=\emptyset$ and~$N=\emptyset$.
Otherwise, let~$c_1< \ldots < c_k$ denote their $x$-coordinates. % of the $k=|\is(({\tour{1}},\tour{{2}}),F)|$ crossings. 
The \emph{proximity} of the $j$-th crossing~$x$ in this set is defined
as~$\pro(x) = c_{j}-c_{j-1}$ (for $j=1$, use $c_{0}=-\infty$). 
%Let~$i$ be the level of the grid line containing $F$ in the dissection.
We set $N$ as the set of
all crossings in~$\is(({\tour{1}},{\tour{2}}),F)$ with proximity at most $L/(2^i r)$,
and~$G$ as the set of the remaining crossings.
%Let $G$ be the set of the remaining crossings of $(\tour{1},\tour{2})$ with $F$, and 
%let~$g=\lfloor r^2/(10(|G|+1)) \rfloor$.
If~$|G|\le 1$, we set~$g=\lfloor r \log L \rfloor$, otherwise~$g=\lfloor r^2/(10(|G|+1)) \rfloor$. 

Next, based on~$N$,~$G$ and~$H$, we create a set~$\mathcal{S}$ of disjoint line segments as follows:
we connect each point in~$N$ to its left neighbor, and each point in~$G$ and all points in~$H$ to their closest portals in $\grid(F,g)$.
The union of all these connections yields a set,~$\mathcal{S}$, of at most~$|G|+1$ maximal segments (if not empty, the set~$H$ is connected via an infinitesimally short segment, possibly belonging to a longer segment with points from~$N$ and~$G$).

Instead of using the standard patching procedure of Arora~\cite{Arora98,Arora2007} (Lemma~\ref{lem:patch}), we apply Lemma~\ref{lem:patch2} to each line segment of $\mathcal{S}$ to obtain a new pair of tours $(\touralthat{1},\touralthat{2})$ crossing~$F$ at no more than~$|G|+1$ portals and each portal at most ten times. 

If~$|G|\le 1$, then we use at most two portals and they belong to~$\grid(F,\lfloor r \log L \rfloor)$, which satisfies Case~\ref{case:Arora} of Definition~\ref{def:simple}.
Otherwise, if~$|G|>1$, 
%Furthermore, if it crosses~$F$ at more than two portals, then all crossed portals belong to~$\grid(F,g)$ (note that the (sub)segment corresponding to~$H$ readily lies on the last portal of~$\grid(F,g)$). 
observe that,~$F$ is crossed~$m\le 10(|G|+1)$ times in total, which implies that~$g=\lfloor r^2/(10(|G|+1)) \rfloor$ satisfies the bound in Case~\ref{case:Sensitive} of  Definition~\ref{def:simple}.
Thus, we conclude with Lemma~\ref{lem:patch2} that the resulting pair~$(\touralthat{1},\touralthat{2})$ is noncrossing and satisfies the conditions of Definition~\ref{def:simple} for the boundary~$F$.
Note that the new line segments that we introduced for patching may cross other boundaries (perpendicular to~$F$) and thus %, by Lemma~\ref{lem:patch2}, 
introduce new crossing points on each of them infinitesimally close to~$F$, that is, infinitesimally close to their respective open endpoints.
Moreover, by our construction of the dissection, any affected boundaries are shorter and thus haven't been considered yet by our patching procedure. Conversely, we can conclude that~$(\touralthat{1},\touralthat{2})$ satisfies the conditions of Definition~\ref{def:simple} also for any boundaries considered so far (whose length is equal or larger than~$|F|$).

%Add different intro text ? we already finished with all F and focus on arbitrary ell immediately?
By Lemma~\ref{lem:patch2}, the total expected patching cost for~$F$ (that is, $\wt(\touralthat{1} \cup \touralthat{2})-\wt(\hat{\tour{1}} \cup \hat{\tour{2}})$) is proportional to $\wt(\mathcal{S})$. 
Note that the attribution of~$H$ is only infinitesimal to the total cost as all crossing points in~$H$ are infinitesimally close to each other and readily lie on the last portal of~$\grid(F,g)$ (furthermore, the cost can be charged to the patching cost of the boundary that caused the crossing points of~$H$). Since~$H$ has also no influence on the proximity of the other crossing points, we will bound the expected patching cost assuming that~$H$ is empty; that is, in the following cost analysis, we assume~$\is(\hat{\tour{1}} \cup \hat{\tour{2}})=\is({\tour{1}} \cup {\tour{2}})$. 
Following the analysis of Kisfaludi{-}Bak et al.~\cite{focs21}, we first bound this cost for~$F$ in terms of proximity. Subsequently, we 
show that each crossing point contributes in expectation~$\Oh(1/r)$ to the total cost. This fact allows us to bound the total cost in terms of the number of crossing points and, finally, in terms of~$\wt(\tour{1}) + \wt(\tour{2})$.
%bound the total patching cost along each grid line in terms of the number of crossing points, and, finally, we relate this number to~$\wt(\tour{1}) + \wt(\tour{2})$. 

%In the following, we first bound~$\wt(\mathcal{S})$ in terms of proximity, then we show that, in expectation, each crossing point contributes~$\Oh(1/r)$ to the total patching cost that we finally relate to~$\wt(\tour{1}) + \wt(\tour{2})$. 

%We first consider the case~$|G|=1$. 
%Let~$\ell$ be the grid line containing~$F$.
To simplify the discussion, we overestimate the costs and charge every crossing point in~$\is(({\tour{1}},{\tour{2}}),F)$ with the cost of being connected to the closest portal in~$\grid(F,\lfloor r \log L \rfloor)$ (independently of wherwhere the point really is). Note that this grid corresponds to the portal placement of Arora~\cite{Arora98,Arora2007}. Thus, following Arora's arguments~\cite{Arora98,Arora2007}, this expected (amortized) charging cost amounts to~$\Oh(1/r)$ for every crossing point and is thus within the desired bound (see above).
For the case~$|G|=1$, we are therefore left with bounding the total cost~$\sum_{x \in N} \pro(x)$ of connecting the points in~$N$ to their left neighbors. 
Since this sum is encompassed in a bound that we establish below for the case~$|G|>1$, it will immediately follow that we can charge this sum with the expected value~$\Oh(1/r)$ to each crossing point. Thus, from now on, we assume that~$|G|>1$.

For each point in~$G$, we pay no more than~$|F|/(2(g-1))$ to connect it to its closest portal in~$\grid(F,g)$. 
%Also by constrution, %Since the distance between the points in~$G$ is not smaller than the proximity threshold~$L/(2^i r)$, 
%we have~$|G|=\Oh(r)$. 
Since the definition of $G$ implies~$|G|=\Oh(r)$, the total connection cost of the points in~$G$ is bounded by
\begin{align}\label{eq:wtf}
	\sum_{x \in G}\frac{\wt(F)}{2(g-1)} 
%	&= \sum_{x \in G}\frac{\wt(F)}{2 \cdot (\lfloor \frac{r^2}{10(|G|+1)}\rfloor - 1)}
%	\\&\le \sum_{x \in G}\frac{\wt(F)}{2 \cdot (\frac{r^2}{10(|G|+1)} - 2)}
%	\\&= \sum_{x \in G}\frac{10 \cdot \wt(F) \cdot (|G|+1)}{2 \cdot (r^2 - 20(|G|+1))}
%	\\&= \sum_{x \in G}\frac{10 \cdot \wt(F) \cdot (|G|+1)}{2 \cdot (r^2 - \Oh(r))}
%	\\&= \sum_{x \in G} \Oh\left(\frac{\wt(F) \cdot |G|}{r^2}\right)
	= \Oh\left( \sum_{x \in G}\frac{L|G|}{2^i r^2}\right) 
	= \Oh\left(\frac{L|G|^2}{2^i r^2}\right)~.
\end{align}
%% NOTE: For r>= 81, we get the constant factor 20 (as then $|G| \le 1+ 2^i r |F|/L = 1 + 2r \le \frac{r^2}{40}$).
To further bound $|G|^2$ in terms of proximity, % the proximities of the vertices in $G$, 
let $\rho$ denote the point in~$G$ with the minimum $x$-coordinate. 
%Since all crossings in $G$ are in an interval of length $\wt(F)=2L/2^i$,~$\sum_{x \in G \setminus \{\rho\}} \pro(x) \leq 2L/2^i$ \importantquestion{KF: $2\cdot$? See definition of grid line levels!}. 
We apply Cauchy-Schwartz to the vectors $\big(\sqrt{\pro(x)}\big)_{x\in G\setminus \{\rho\}}$ and $\big(\sqrt{1/\pro(x)}\big)_{x\in G\setminus \{\rho\}}$, noting that $1/\pro(\rho)=0$ and~$\sum_{x \in G \setminus \{\rho\}} \pro(x) \leq \wt(F) = 2L/2^i$.
Therefore, we have 

\begin{equation}\label{eq:hmam}
	|G \setminus \{\rho\}|^2 \leq \left(\sum_{x \in G \setminus
    \{\rho\}} \pro(x) \right) \left( \sum_{x \in
    G}\frac{1}{\pro(x)}\right) \leq \frac{2L}{2^i}\sum_{x \in
    G}\frac{1}{\pro(x)}~.
\end{equation}

Since $|G|^2 \le 4 |G \setminus \{\rho\}|^2$ (as~$|G|>1$), we can combine \eqref{eq:hmam} with~\eqref{eq:wtf} and, altogether, obtain the bound 

\[
	\wt(\mathcal{S}) = \Oh \left( \sum_{x \in N} \pro(x) +  \left(\frac{L}{2^ir}\right)^2 \sum_{x \in G}\frac{1}{\pro(x)}\right)~.
\]
(The right side also bounds the connection cost for~$N$ of the case~$|G|=1$.)
Thus, each crossing point contributes a specific amount to the total bound depending on its proximity and the level~$i$ of the grid line containing~$F$. Recall that the proximity also depends on~$i$ as the level determines whether a crossing point is the left-most one of its boundary and therefore whether its proximity is fixed to~$\infty$. To remove this dependence on the level, we define, for each grid line~$\ell$, the \emph{relaxed proximity} of each crossing point $x\in\is((\tour{1},\tour{2}),\ell)$ as the distance to the left neighbor of~$x$ in $\is((\tour{1},\tour{2}),\ell)$, and~$\infty$ if there is no left neighbor. Thus, the relaxed proximity equals the old proximity for all but possibly the left-most point~$\rho$ of~$\is((\tour{1},\tour{2}),F)$. 
Similarly, let~$N'$ be the set of all crossing points in $\is((\tour{1},\tour{2}),F)$ with relaxed proximity at most~$L/(2^i r)$ and let~$G'$ contain all the other points. 
Thus, either~$N'=N$ and~$G'=G$ (if~$\pro'(\rho)=\pro(\rho)$), or~$N'=N\cup\{\rho\}$ and~$G'=G\setminus\{\rho\}$. Using that~$1/\pro(\rho)=0$, we have
\[
	\wt(\mathcal{S}) = \Oh \left( \sum_{x \in N'} \pro'(x) +  \left(\frac{L}{2^ir}\right)^2 \sum_{x \in G'}\frac{1}{\pro'(x)}\right)~.
\]
Note that~$x\in N'$ if and only if the level~$i$ is at most~$\theta(x) := \log (L/(r \cdot \pro'(x)))$. Thus, 
%sif we consider the total patching cost along a fixed grid line~$\ell$ of some level~$i$, 
the contribution of each crossing point~$x$ % \in\is((\tour{1},\tour{2}),F)$ 
to the total bound is~$\alpha_i(x) := \Oh(\pro'(x))$ if~$i\le \theta(x)$, and $\alpha_i(x) := \Oh\left((L/(2^ir))^2 / \pro(x')\right)$ otherwise. 

Now, for a fixed grid line $\ell$ and a crossing point~$x \in \is((\tour{1},\tour{2}),\ell)$, Lemma~\ref{lem:levelprob} allows us to bound the expected patching cost
due to $x$ by 
\[ 
    \sum_{i=0}^{1+\log L} \Pr[\ell \text{ has level } i] \cdot \alpha_i(x) =
\Oh\left(\sum_{i=0}^{\theta(x) } \frac{2^i}{L} \pro'(x) + \sum_{i=\theta(x)+1}^{1+\log L}
\frac{L}{2^i} \frac{1}{r^2 \pro'(x)}\right) = \Oh\left(\frac{1}{r}\right)~,
\]
%where $\theta:= \log \frac{L}{r \cdot \pro(x)}$ is the level of(recall that $x \in G$ when
%$\pro(x) > \frac{L}{2^ir}$ and $x \in N$ otherwise). 
where the right side follows by the convergence of sums of geometric progressions. % (and noting~$1/\pro(x)=0$ if~$x$ is the left-most crossing on its boundary).
%\importanttodo[inline]{KF: Before adding up, we need to bound the cost of~$F$ by the number of intersection points of the original instance! For this, we need to argue that the previous patchings introduced at most constant many new intersection points on~$F$!}
%
%%Now, if~$|G|=1$, the tota
%%\[ \wt(S) \le  \sum_{x \in N} \pro(x)  + \frac{\wt(F)}{2(g-1)}\]
%%The first term is bounded by \dots and yields by the discussion above an expected cost of \dots on every crossing point.
%%The second term  
%
%
Consequently, the expected total patching cost along~$\ell$ is bounded by $\Oh\left(|\is((\tour{1},{\tour{2}}),\ell)|/r \right)$.
%\[\mathbb{E}_{\mathbf{a}}[\wt(S)] = \Oh\left(|\is((\hat{{1}},{tour{2}}),F)|/r \right) ~.\]
%Observe that there is at most one boundary longer than~$F$ that might have caused additional crossing points on~$F$.
%
% only crossings points on~$F$ due to patching longer boundaries could 
%Observe that $\is((\hat{\tour{1}},\hat{tour{2}}),F)$ contains $\is((\tour{1},tour{2}),F)$ and up to~$O(1)$ crossings and that any  of
%The right side is at most $\Oh(1/r)$ by the convergence of sums of geometric progressions. 
%
Adding everything up, the total expected patching cost along all grid lines is at most 
\begin{align*}
%    \sum_{\ell}
%    \sum_{x \in \is((\tour{1},\tour{2}),\ell)}
%	\sum_{i=0}^{1+\log L} \Pr[\ell \text{ has level } i] \cdot
%	\alpha_i(x) 
%	\\&= 
	\sum_{\ell} \Oh(|\is((\tour{1},\tour{2}),\ell)|/r)  
	&= \sum_{\ell} \Oh(|\is(\tour{1}),\ell)|/r)+ \sum_{\ell}\Oh(|\is(\tour{2}),\ell)|/r) 
	\\&= \Oh((\wt(\tour{1}) + \wt(\tour{2}))/r)
\end{align*}
by Lemma~\ref{lem:crossingsvslength}, as required.

\subsection{EPTAS for Noncrossing Euclidean Tours}
\label{sec:alg-2noncross}

In this section, we are going to use our structure theorem to give an approximation algorithm for \BTSP. For a fixed $\eps >0$, our algorithm will provide a~$(1+\eps)$-approximation for this problem. 

The input consists of two terminal sets~$P_1,P_2 \subset \mathbb{R}^2$, each colored in a different color. 
A \emph{bounding box} of a point set is the smallest axis-aligned square containing that set. 
If the bounding boxes of~$P_1$ and~$P_2$ are disjoint, then we can treat them as two independent
(uncolored) instances of TSP and solve them using the algorithm of Kisfaludi{-}Bak et al.~\cite{focs21}. 
Thus, in the end, we assume that the two bounding boxes intersect. 
Let~$L$ be the side length of the bounding box of~$P_1 \cup P_2$ and observe that~$L \ge \opt$.
By scaling and translating the instance, we assume that~$L$ is a power of~$2$ and of order~$\Theta(n / \eps )$, and that the corners of the bounding box of~$P_1 \cup P_2$ are integer. 

\paragraph*{Perturbation} 
As mentioned in Section~\ref{sec:arorasTools}, we first perturb the instance such that all input terminals lie in~$\{0,\ldots,L\}^2$ on disjoint positions. 
For this, we follow related work~\cite{Arora98,red-blue,isaac15}, and move every terminal in~$P_1$ to the closest position in~$\mathbb{Z}^2$ with even~$x$-coordinate, and every terminal in~$P_2$ to a closest position in~$\mathbb{Z}^2$ with odd~$x$-coordinate. Thus, no terminals of the same color end up in the same position. If there are terminals of the same color on the same position, we treat them from now on as a single terminal.  Later, we argue that
there is an optimum solution 
to this perturbed instance that is only negligibly more costly
than an optimum solution to the original instance. 
From now on, we assume that~$P_1,P_2 \in \{0,\ldots,L\}^2$.

Given~$L$, the perturbed instance and a shift vector~$\mathbf{a}$, we construct a quadtree~$QT(P,\mathbf{a})$ as described in Section~\ref{sec:arorasTools}. Let~$D(\mathbf{a})$ be the corresponding dissection. 
Our idea is to compute the cost of a noncrossing \kt{noncrossing pair of tours or pair of noncrossing tours?} pair~$(\touralt{1},\touralt{2})$ of tours that solves the perturbed instance optimally and that crosses each boundary of the dissection only through carefully selected portals. 
%Later, we will see that this pair is also an optimum 
%a $(1/\eps)$-simple pair of tours that solves the perturbed instance optimally. 
Later, we show that the solution is not much more expensive than the optimum solution to the original instance.

Each boundary of the dissection is allowed to be crossed only through portals belonging to a so-called \emph{fine} portal set that we define as follows.
%\resolvedtodoNotDoneYet[inline]{KF: What do we mean by boundary here? Only the edges of the same level?}{Boundary is a line segment used for subdivision of the quadtree. The boundaries of a cell are the four boundaries that contain the cell's edges. However, when looking at a cell, we could just talk about edges instead of boundaries. When talking how to place portals on an edge of a cell, we could mention the boundaries as it happens w.r.t. boundaries.} 
%As stated above, our overall goal will be to compute the cost of two noncrossing tours, $\touralt{1}$ and $\touralt{2}$, that are allowed to cross the boundary of each cell only  %$\partial C$ only 
%in carefully selected portals. 
%There is a number of choices \importantquestion[inline]{KF: What do we want to say here? That in theory/related words, there are several possibilities how to place portals and we choose a single possibility for our paper (exactly as Arora did?) Or do we mean that we will use different placements where~$B$ means the respective placement that we use for the respective subproblem, and (for comparison/analysis only) we define also the placement, called~$A$, that Arora used.} for assigning a feasible set of portals $B$ to the boundary $\partial C$, where $B$ is the set of portals where $(\touralt{1},\touralt{2})$ crosses $F$. To be more precise, let $A = \bigcup_{F \subset \partial C}  \grid(F,10\log(n)/\eps)$ be the set of
%$\Oh(\log(n)/\eps)$ equidistant portals on every boundary $F$ of $\partial C$
%(i.e., $A$ is the portal placement according to the original Arora's
%procedure).
\begin{definition}\label{def:fine}
Let~$\mathbf{a}$ be a random shift vector. % and let~$F$ be a boundary of~$D(\mathbf{a})$.
%Furthermore, let $A = \grid(F,10\log(n)/\eps)$ be the set of
%%$\Oh(\log(n)/\eps)$ equidistant portals on the boundaries~$F_1,\dots,F_4$ (note that $A$ is the portal placement according to the original Arora's procedure\todo[inline]{KF: Reference!}).
Set~$B$ of portals is called \emph{fine for a boundary~$F$} of~$D(\mathbf{a})$ if 
	\begin{enumerate}[label=(\alph*)]
		\item \label{def:fineArora} $|B\cap F|\le 2$ and $B\cap F \subset \grid(F,\lfloor (\log L)/\eps\rfloor)$, or
        \item \label{def:fineSparsity} $B \cap F \subset \grid(F,1/(\eps^2 k))$ for some $k\ge|B\cap F|$. 
%        \importanttodo[inline]{KF: Note that this implies $|B \cap F| \le 1/(\eps^2 k) \le 1/(\eps^2 |B \cap F|)$ which implies~$|B \cap F| \le 1/\eps$ as desired. On the other hand, we have now more flexibility when given an~$r$-simple pair of tours. We just choose~$r=1/\eps$ and~$k=r^2/g$. If~$u$ is the number of used portals, we have~$u\le m$ and thus in total~$g\le r^2/m \le r^2/u$ and therefore~$k=r^2/g \ge r^2\cdot u/r^2 = u$. Thus, observing~$|B \cap F|=u$ shows that everything works fine now.} 
%        \importanttodo[inline]{KF: Rounding!}  \todo[inline]{KF: The upper bound of~$1/\eps$ on $|B\cap F|$ follows automatically! Write about it!} % for some $m_{F}$ such that $|B\cap F| = m_{F} \leq \sqrt{10}/\eps$. \infodone[inline]{KF: Increased upper bound to~$\sqrt{10}/\eps$ to ease the discussion of transforming a ~$r$-simple pair of tours into a fine-portal respecting one.} %\importantquestion[inline]{KF: Really equality?} %it crosses $F$ only through portals belonging to~$\grid(F, g)$, for~$g\le r^2/m$ where~$m$ is the number of times~$F$ is crossed.
	\end{enumerate}
%Let~$C$ be a cell of the quadtree $QT(P,\mathbf{a})$. 
%We need: input = fine portal set for a cell(!) ; fine-portal-respecting solution
%A portal set~$B$ is called (just) \emph{fine} if it is fine for every boundary of the quadtree and is contained in the union of all boundaries.
Set~$B$ of portals is called \emph{fine for a cell~$C$} of the quadtree $QT(P,\mathbf{a})$ if~$B$ is fine for each of the four boundaries of~$D(\mathbf{a})$ that contain a border edge of~$C$ and~$B$ is contained in the union of the four boundaries.  
A set of paths (open nor closed\kt{in sense of cycles. define it somewhere?}) is called \emph{fine-portal-respecting} if there is a portal set~$B$ that is fine for each boundary and the tours cross each boundary only through the portals of~$B$ and each portal at most ten times, and the paths are noncrossing.
%Finally,
%
%, for each boundary~$
%Let~$F_1, \dots, F_4$ be the four boundaries containing the four edges of~$C$, and let $A = \bigcup_{i=1}^{4} \grid(F_i,10\log(n)/\eps)$ be the set of
%$\Oh(\log(n)/\eps)$ equidistant portals on the boundaries~$F_1,\dots,F_4$ (note that $A$ is the portal placement according to the original Arora's procedure\todo[inline]{KF: Reference!}).
%A set of portals~$B \subset \partial C$ is called \emph{fine for~$C$} if, for~$i=1,\dots,4$, % every boundary~$F$ containing a border edge of~$C$ 
%	\begin{enumerate}[label=(\alph*)]
%		\item $|B\cap F_i|\le 2$ and $B\cap F_i \subset A\cap F$, or
%        \item $B \cap F_i \subset \grid(F_i,1/(\eps^2 m_{F_i}))$ for some $m_{F_i}$ such that $|B\cap F_i| \le m_{F_i} \leq 1/\eps$. %it crosses $F$ only through portals belonging to~$\grid(F, g)$, for~$g\le r^2/m$ where~$m$ is the number of times~$F$ is crossed.
%	\end{enumerate}
\end{definition}
Note that it is no coincidence that the definitions of fine-portal-respecting pairs of tours and~$r$-simple pairs of tours are very similar. Indeed, later we observe that any~$(1\/\eps)$-simple pair of tours is fine-portal-respecting. 
Also, note that the conditions of the definition imply that a fine portal set has no more than~$\Oh(1/\eps)$ portals.

Now, our overall goal is to compute the cost of an optimum fine-portal-respecting pair of tours solving the perturbed instance. \ktp{Can we underestimate the cost as said in Section~\ref{sec:prelim}?}
%As discussed in Section~\ref{sec:prelim}, we actually compute a number that is sufficiently
%
%
%%\todo{KF: unclear sentence.}Definition~\ref{def:fine} is designed as to guess precisely the set of portals
%%in which the desired pair of tours crosses a boundary of a cell. 
%%\todo{KF: Rephrase that sentence.} 
%%To give some intuition, options (i) and (ii) follow points (a) and (b) of Definition~\ref{def:simple}, with some additional constraints. In option (i) we enforce $p$ to be taken from $A$, as we often need
%%a perfect precision on boundaries \todo{KF: or on the crossing points? Why do we need that perfect precision?} that are crossed exactly once. In option (ii), we bound the number of crossings by $1/\eps$ in order to bound the number of subproblems in the dynamic programming. 
%%Towards the end of this section, we argue formally that it is sufficient to consider \emph{fine} sets of portals.
%%
%%Our goal is to compute the cost of a near-optimal fine-portal respecting noncrossing pair of tours.
%
%
%
%
%%In more detail, 
%
%Next, we construct a randomly shifted quadtree on the grid. Additionally\importantquestion[inline]{KF: Wait! I think we already defined the quadtree in this shifted way! However, maybe it makes the presentation indeed simpler to define it first without the shift and then to tell that we shift it. But we probably should tell about in also in the preliminaries.}, we can shift it by
%$1/2$ in each coordinate to guarantee that no input point intersects the boundary of the quadtree.
The strategy is the same as in the standard approximation
scheme for the traveling salesman problem~\cite{Arora98}. 
%Given~$L$, the perturbed instance and a shift vector~$\mathbf{a}$, we construct a quadtree~$QT(P,\mathbf{a})$ as described in Section~\ref{sec:arorasTools}. Let~$D(\mathbf{a})$ be the corresponding dissection. 
%
%Our overall goal is to compute the cost of an optimum %~$r$-simple 
%pair of noncrossing tours,~$(\touralt{1}, \touralt{2})$ that are allowed to cross each boundary of the dissection only through carefully selected portals, belonging to so-called \emph{fine} portal set. 
%For
%
We use the quadtree to guide our algorithm based on dynamic programming.
% to compute a number (as discussed in the preliminaries in Section~\ref{sec:prelim}) that is arbitrarily close to the cost of an optimum portal-respecting pair of tours solving the .
%in carefully selected portals. 
Starting from the lowest levels of the
quadtree, we %use dynamic programming to %dynamically 
compute partial solutions that we combine together to 
obtain solutions for the next higher levels. % until we end up with a solution to the overall problem.
More concretely, 
for each cell of our quadtree, we define a set of subproblems, % that we solve using dynamic programming in a bottom-up manner starting with the cells at the lower level. 
in each of which, we look for a collection of paths that connect neighboring cells in a prescribed manner
while visiting all terminals inside. Depending on an input parameter of the subproblem, the paths of each color will be disjoint or form a cycle. 
%By dynamic programming, we solve the subproblems going the quadtree bottom-up.
%Instead of looking for two noncrossing shortest tours that connect the points inside the respective cell, in each subproblem we rather look for a collection of paths that connect the neighboring cells in a prescribed manner
%while visiting all points inside.
%We solve 
%our dynamic programming algorithm works as follows. It iterates
%through the quadtree in a bottom-up fashion. 
%we start with the
%smallest quadtree cells on the lowest level and based on them, we compute the
%minimum solution one level higher. In subproblems that correspond to internal vertices of the quadtree, we are no longer searching for
%two noncrossing shortest tours that connect the points inside, but rather for a
%collection of paths that connect neighboring cells of a quadtree in the
%prescribed manner. 
%Let~$C$ be a quadtree cel and let $\partial C$ denote its
%boundary.
%For every quadtree cell $C$, let $\partial C$ denote its
%boundary.
%\todo[inline]{KF: Move the whole paragraph not related directly to defining subproblems to some other place in order to avoid the reading flow getting interrupted.}

% 
% 
% The Structural Theorem (Theorem~\ref{thm:struct}) guarantees the
% existence of some set of portals $B_1,B_2 \subseteq \partial C$ that will be traversed
% by tours $\pi'_1,\pi'_2$. 
To define our subproblems, fix some cell~$C$ and consider a fixed fine portal set $B$ for~$C$. %of portals on $\partial C$. %Assume that $B_1$ and $B_2$ are correctly guessed, i.e., $B_1$ and $B_2$ are precisely the points in which $\touralt{1}$ (respectively $\touralt{2}$) cross $\partial C$. 
Let~$\partial C$ be the set of border edges of~$C$, and let~$B' = B \cap \partial C$. % be the set of portals in~$B$ that lie on the border edges of~$C$. \importanttodo{KF: Define~$\partial C$ w.r.t. disjoint (open) border edges!}
We want to guess how exactly each portal in~$B'$ is crossed by the fine-portal-respecting pair~$(\touralt{1},\touralt{2})$ (recall that each portal can be crossed at most ten times according to Definition~\ref{def:fine}). 
For this reason, we define~$V_{B'}$ as a matrix of size $|B'| \times 11$ all whose entries are numbers from the set~$\{0,1,2\}$; 
%\todo{KF: Do we need the sentence with vector?} 
thus,~$V_{B'}[b]$ is a vector of size $11$ and, for~$b \in {B'}$ and $i \in \{ 1, \ldots, 11\}$, we have~$V_{B'}[b,i] \in \{ 0,1,2 \}$.
The interpretation of $V_{B'}[b,i]$ is simple as follows. If $l_b$ is the minimum $i$ such that $V_{B'}[b,i]=0$, then $l_b-1$ is the number of times the tours cross the portal~$b$. For $i < l_b$, the value~$V_{B'}[b,i]$ determines whether the $i$-th crossing of $b$ is due to $\touralt{1}$ (value $1$) or due to $\touralt{2}$ (value $2$). 
Let $\cpy({B'},V_{B'})$ be a new set of colored portals, that is obtained 
by subdividing each portal~$b \in {B'}$ in~$l_b$ shorter portals~$b_1,\dots,b_{l_b}$ (ordered along~$F$ in a globally fixed direction), and by setting the color of~$b_i$ to $\col{b_i}=V_{B'}[b,i]$. \kt{Should we say how we represent infinitesimality in our algorithm?}
For $j \in \{ 1,2 \}$, let ${B'}_j=\{ b \in \cpy({B'},V_{B'}) : \col{b}=j \}$. 
The subproblem for the cell $C$ is additionally defined by two perfect matchings,~$M_1$ on~${B'}_1$ and~$M_2$ on~${B'}_2$ whose union is noncrossing. 
For~$i\in\{1,2\}$, we say that a collection~$\mathcal{P}$ of~$|B'_i|$ \emph{realizes}~$M_i$ if for each $(p,q) \in M_i$ there is a path $\pi_i \in \mathcal{P}$ with $p$ and $q$ as endpoints.
The formal definition of a subproblem for a cell $C$ is as follows.

% For each boundary $F$ there is a unique maximum boundary $\ex(F)$ that is the boundary of a cell in the compressed quadtree that contains $F$. Note that when considering the cell $C$ and one of its
% boundaries $F$, we will place the portals according to the grids of $\ex(F)$, or
% potentially at exactly one point in $F$.
% Now, we define the subproblems in our approximation scheme.
\newcommand{\doCycle}{\mathrm{Cycle}}
\defproblem{Noncrossing Multipath Problem}
{%
A nonempty cell $C$ of the quadtree, 
a fine portal set $B$ for~$C$ with ${B'}:=B \cap \partial C$, 
a matrix $V_{B'} \in \{0,1,2\}^{|{B'}| \times 11}$, and two perfect matchings, $M_1$ on ${B'}_1=\{b \in \cpy({B'},V_{B'}): \col{b}=1 \}$, and~$M_2$ on ${B'}_2=\{ b \in \cpy({B'},V_{B'}): \col{b}=2 \}$, such that~$M_1 \cup M_2$ is noncrossing, and a (possibly empty) subset~$\doCycle$ of~$\{1,2\}$.}
{Find two %$r$-simple 
path collections
    $\mathcal{P}_{{B'}_1,M_1},\mathcal{P}_{{B'}_2,M_2}$ of minimum total length
    such that their union is fine-portal-respecting and that 
%     \importantquestion[inline]{KF: We might not be able to find an optimum solution but only a close-to-optimal one according to our definition, right?} that
    satisfy the following properties for all $i \in \{1,2\}$:
\vspace{-0.5em}
\begin{itemize}
    \setlength\itemsep{0.0em}
    \item every terminal in~$P_i \cap C$ is visited by a path from~$\mathcal{P}_{{B'}_i,M_i}$, 
    \item the paths in~$\mathcal{P}_{{B'}_i,M_i}$ are pairwise noncrossing and entirely contained in $C$, and 
    \item $\mathcal{P}_{{B'}_i,M_i}$ realizes the matching $M_i$ on ${B'}_i$, and
    \item if~$i \in \doCycle$, then the paths in~$\mathcal{P}_{{B'}_i,M_i}$ form a cycle and ${B'}_i=\emptyset$.
\end{itemize}
}

Though our subproblems are similar to those used in approximation schemes for TSP in the literature~\cite{Arora98,focs21},
there is one main difference:
%Similar definitions of subproblems commonly appear in the formulations of
%dynamic programming for traveling salesman approximation
%schemes~\cite{Arora98,focs21}.  The main difference here is that 
here, we have
two types of paths that correspond to $\pi_1$ and $\pi_2$ in the solution. 
Exactly as in Arora's approximation scheme~\cite{Arora98}, our dynamic programming fills a lookup
table with the solution costs of all the multipath problem instances that
arise in the quadtree. The details follow next.
% When the
% table is built, it is enough to output the entry that corresponds to the root
% of the quadtree. The number of non-empty cells in quadtree is
% $\Oh(n \log (n))$.  For each border edge $F$ of the cell $C$, we guess an integer $m_F \leq 1/\eps$
% that is the number of times the $\Oh(1/\eps)$-simple salesman tours
% cross it. Then, we guess a sets $B_1,B_2$ by selecting a sets of size $m$
% from $\bigcup_F \grid(\ex(F),r^2/m_F)$, where $\sum_F m_F=m$. There are at most
% $\prod_F \binom{r^2/m_F}{m_F} \le 2^{\Oh(r)}$ possible choices for the portal
% sets $B_1$ and $B_2$ by 
% Claim~\ref{portal-bin-ineq}, since the number of facets $F$ of $C$ is at most
% $4$.  Observe that the trivial bound on the number of perfect matchings
% on $m$ points is $2^{\Oh(m \log{m})}$.  Nevertheless in $d=2$ we could use that
% our tours are crossing-free and it was efficient to look for
% ``crossing-free matchings'' (and their number is at most $2^{\Oh(m)}$).

\paragraph*{Base Case} We start with the base case, where the cell $C$ is a
leaf of the quadtree and contains at most 
one terminal. Without loss of
generality, assume that if there is a terminal inside, then it is from $P_1$ (the algorithm for the other
case is analogous). 
We will generate all subproblems for~$C$ and solve each of them.
Consider every fine portal set~$B$ for~$C$ and, for~$B'=B\cap \partial C$, every possible matrix~$V_{B'}$. Each matrix defines an instance~$\cpy({B'},V_{B'})$ of at most~$10 |B|$ portals that are colored with~$1$ and~$2$. 
For each possible noncrossing pair of perfect matchings, $M_1$ on~${B'}_1$ and $M_2$ on ${B'}_2$, and each subset~$\doCycle\subseteq\{1,2\}$ that is compatible with~${B'}_1$ and~${B'}_2$ (in the sense that if~$i\in\doCycle$, then~${B'_i}=\emptyset$) any path collection realizing~$M_i$ can be connected in the portals together to a cycle without intersecting any other path collection realizing the other matching), we compute the cost of an optimum multipath solution. 
We use dynamic programming to enumerate all possible such pairs of perfect matchings together with their solutions.
To be more precise, let us fix $B$ and $V_{B'}$, then ${B'}_1$ and ${B'}_2$ are implied by the choice of $B$ and $V_{B'}$. Let $p\in P_1$ be the
only terminal inside $C$ (if it exists).  
We define~$\mathtt{BaseCase}$ as our lookup table as follows.
For every $X_1 \subseteq {B'}_1$, every~$X_2 \subseteq {B'}_2$, and every set~$X'$ that contains a pair of elements from~$X_1$ if~$p$ exists, and is empty otherwise, the entry~$\mathtt{BaseCase}[X_1,X_2,X']$ is a set containing every triple~$(M_1, M_2, \wt(M_1 \cup M_2)$, where~$M_1$ is a perfect matching on~$X_1$ with~$X'\subseteq M_1$,~$M_2$ a perfect matching on~$X_2$,~$M_1\cup M_2$ is noncrossing, and~$\wt(M_1 \cup M_2)$ is the cost of a minimum path collection realizing the matchings within the cell where~$p$ (if exists) is connected via the path realizing the pair in~$X'$. 
%
%Each triple consists of two perfect matchings $M_1$ and $M_2$ on $X_1$ and $X_2$ respectively, and the third position of the triple reveals the minimum cost of the path collections that realize $M_1$ (respectively$M_2$) on $B_1$ (respectively $B_2$), are entirely contained in $C$ and visit $p$ (if exists). 
Initially, if~$p$ does not exist, we set
$\mathtt{BaseCase}[\emptyset,\emptyset,\emptyset] = \{(\emptyset,\emptyset,0)\}$.
Otherwise, for every $a,b \in {B'}_1$, we set 
$\mathtt{BaseCase}[\{a,b\},\emptyset,\{(a,b)\}] := \{ (\{(a,b)\},\emptyset, \dist(a,p) + \dist(p,b)) \}$ (which intuitively means that $p$ is connected to the portals
$a,b \in {B'}_1$ and there are no other paths in~$C$).  
%\newcommand{\canconnectdirectly}{\mathrm{isfree}}
%We define a helpful operator that tells us whether we can connect a matching directly with one line segment, or whether we have to go around~$p$ (if it exists):
%for~$X'\subseteq X_1$ defined as above and any pair of portals~$(u,v)$ from ${B'}_1$ or ${B'}_2$, let~$\canconnectdirectly(X',(u,v))=1$ if~$X'$ is empty ($p$ does not exist) or, for~$\{a,b\}=X'$, the segment~$uv$ does neither cross the segment~$ap$ nor the segment~$pb$; otherwise~$\canconnectdirectly((a,b),(u,v))=0$.  
\newcommand{\pathcost}{\mathrm{pathcost}}
We define a helpful operator to determine the cost of connecting a pair of portals:
% returns the cost of a path realizing a matching in dependence of how~$p$  tells us whether we can connect a matching directly with one line segment, or whether we have to go around~$p$ (if it exists):
for~$X'\subseteq X_1$ defined as above and any pair of portals~$(u,v)$ from ${B'}_1$ or ${B'}_2$, let~$\pathcost(X',(u,v)):=\dist(u,v)$ if~$X'$ is empty ($p$ does not exist) or, for~$\{a,b\}=X'$, the segment~$uv$ does neither cross the segment~$ap$ nor the segment~$pb$; otherwise~$\pathcost(X',(u,v)):=\dist(u,p)+\dist(p,v)$.  
Next, for every $X_1 \subseteq
B_1$, $X_2 \subseteq B_2$, $X'\subseteq X_1$ with~$|X'|=2$ if~$p$ exists, and~$X'=\emptyset$ otherwise, we compute $\mathtt{BaseCase}[X_1,X_2,X']$ 
%as follows:
%For every~$\{u,v\}\subseteq X_1$, we consider~$(M_1,M_2,\wt(M_1\cup M_2)) \in \mathtt{BaseCase}[X_1 \setminus \{u,v\},X_2 ]
%by considering every possible pair~$(u,v)$ in~$X_1$ and~$X_2$, 
%Let
%\[\mathrm{Candidates}:=\left\{((u,v),M_1\cup\{(u,v)\},M_2,\wt(M_1\cup M_2))\mid \{u,v\}\subseteq X_1 \textrm{ and }(M_1,M_2,\wt(M_1\cup M_2)) \in \mathtt{BaseCase}[X_1 \setminus \{u,v\},X_2, X'] \right\}\]
with the following dynamic programming formula:
%\importantquestion[inline]{KF: Why does it work? Especially, why is it correct to just add~$+\dist(u,v)$? Can't it happen that the~$u$-$v$-path must go around another path that visits the terminal?\\Solution idea: We just move the terminal where we want. This again means, that in our perturbed instance we do not fix the locations of the terminals but only assign them to the grid cells and allow them moving for free within the grid cell? However, this still will cause problems in the base case if the quadtree cell of the base case is much larger than a grid cell. Then we can't just move the terminal for free anywhere we want. A possible answer: We define the quadtree differently in the sense that we always go down to grid cell size. But then maybe we have too many instances to solve. Another answer might be that we can charge the cost of moving the terminal anywhere by the cell's size that again is small with high probability.}
\begin{align*}
%x
    \mathtt{BaseCase}[X_1,X_2,X'] :=
		\bigg\{ ~~&\Big(M_1 \cup
            \{(u,v)\},M_2, \wt(M_1\cup M_2) + \pathcost(X',u,v)\Big)
\\ \Big|~&\{u,v\}\subseteq X_1 \textrm{ and } 
        \\&(M_1,M_2,\wt(M_1\cup M_2)) \in \mathtt{BaseCase}[X_1 \setminus \{u,v\},X_2 ]
        \text{ and } 
        \\&(u,v) \text{ is noncrossing with } M_1 \cup M_2
        ~\bigg\}
        \\\bigcup~&
        \\\bigg\{ ~~&\Big(M_1 ,M_2 \cup
                            \{(u,v)\}, \wt(M_1\cup M_2) + \pathcost(X',u,v)\Big)
        \\ \Big|~&\{u,v\}\subseteq X_2 \textrm{ and } 
                \\&(M_1,M_2,\wt(M_1\cup M_2)) \in \mathtt{BaseCase}[X_1 ,X_2 \setminus \{u,v\} ]
                \text{ and } 
                \\&(u,v) \text{ is noncrossing with } M_1 \cup M_2
                ~\bigg\}      
%        \bigcup_{\substack{u,v \in X_2\\u\neq v}}  \big\{ &\big(M_1 
%            ,M_2 \cup \{(u,v)\}, \wt(M_1\cup M_2) + \dist(u,v)\big)
%        \; \\ & \Big| \; 
%        (M_1,M_2,\wt(M_1\cup M_2)) \in \mathtt{BaseCase}[X_1 ,X_2 \setminus \{u,v\} ]
%        \text{ and } \\& (u,v) \text{ is noncrossing with } M_1 \cup M_2
%        \big\} 
\end{align*}
%\begin{align*}
%    \mathtt{BaseCase}[X_1,X_2] := 
%        \bigcup_{\substack{u,v \in X_1\\u\neq v}}  \big\{ &\big(M_1 \cup
%            \{(u,v)\},M_2, \wt(M_1\cup M_2) + \dist(u,v)\big)
%        \; \\ & \Big| \; 
%        (M_1,M_2,\wt(M_1\cup M_2)) \in \mathtt{BaseCase}[X_1 \setminus \{u,v\},X_2 ]
%        \text{ and } \\& (u,v) \text{ is noncrossing with } M_1 \cup M_2
%        \big\} \\
%        \bigcup\\
%        \bigcup_{\substack{u,v \in X_2\\u\neq v}}  \big\{ &\big(M_1 
%            ,M_2 \cup \{(u,v)\}, \wt(M_1\cup M_2) + \dist(u,v)\big)
%        \; \\ & \Big| \; 
%        (M_1,M_2,\wt(M_1\cup M_2)) \in \mathtt{BaseCase}[X_1 ,X_2 \setminus \{u,v\} ]
%        \text{ and } \\& (u,v) \text{ is noncrossing with } M_1 \cup M_2
%        \big\} 
%\end{align*}
\importanttodo[inline]{KF: This is the place to note that our computed cost is not the actual cost of a feasible solution but we should argue why we are arbitrarily close to it.}
Note that if~$1\in\doCycle$ and~$p$ exists, then we assume that~$p$ is visited by an inifitesimal short cycle of length~$0$.
For fixed $B$ and $V_{B'}$, this algorithm runs in $\Oh(2^{\Oh(|B|)})$ and computes the set of all perfect matchings and the corresponding partial solution costs  (cf. the PTAS of Bereg et al.~\cite{isaac15} to see how the base case is handled).

\paragraph*{Algorithm} 
%\todo[inline]{KF: A maybe more precise approach would be to say: look, here is a fixed subproblem for~$C$, and we solve it as follows (by combining all appropriate subproblems of its children that are compatible with the overal subproblem for~$C$). Currently, we just look at all combinations of subproblems, and we hope that we solve all subproblems of~$C$ without really proving that we really reach all subproblems of~$C$.}
To enumerate and solve all subproblems of a non-leaf cell~$C$, we enumerate all compatible subproblems of its four children~$C_1,\ldots,C_4$ in the quadtree, lookup their solutions and combine them to the implied subproblem for~$C$.
We do it as follows. % by solving all subproblems of~$C$ \enquote{at once}.
For $i \in \{ 1, \dots, 4 \}$, we iterate 
over every possible fine set~$B_i$ for~$C_i$ and set~${B'}_i= B_i \cap \partial C_i$, over every corresponding matrix $V_{{B'}_i}$. 
For each such pair $B_i,V_{{B'}_i}$, we iterate through all pairs of matchings $M^{(i)}_1$ and $M^{(i)}_2$, and every subset~$\doCycle_i \subseteq\{1,2\}$. 
For every such obtained quadruple~$\left((B_i,V_{{B'}_i},M^{(i)}_1,M^{(i)}_2,\doCycle_i)\right)_{i\in\{1,\dots,4\}}$ (each consisting of four quintuples), we first check 
whether~$B_1 \cup B_2 \cup B_3 \cup B_4$ restricted to the four boundaries around~$C$ implies a fine set~$B$ for~$C$, whether $V_{{B'}_1},\dots,V_{{B'}_4}$ imply, for~$B'=B \cap \partial C$, a consistent matrix~$V_{{B'}}$, and whether the portals, matchings and $\doCycle_1,\dots,\doCycle_4$ are \emph{compatible}. By compatible, we mean 
that (i) for every border edge shared by two neighboring cells (among~$\{C_1, \dots, C_4\}$), both cells define exactly the same portals of the same color,
(ii) $\doCycle_1,\dots,\doCycle_4$ are pairwise disjoint and if~$c\in \Cup_{i=1}^4 \doCycle_i$, then there are no portals of color~$c$ on~$\partial C_1, \dots, \partial C_4$, and (iii) for each color~$c\in\{1,2\}$, the graph in which the portals of color~$c$ are vertices and the matchings are edges is either a cycle or contains no cycles at all.
The last point implies that~$\doCycle$ for~$C$ contains not only all colors from~$\Cup_{i=1}^4 \doCycle_i$ but also all colors for which the graph was a cycle. 
%  no cycle if~$i\not\in\doCycle$, and is a single cycle otherwise. %  if~$i\n\doCycle$, and has no cycles at all otherwise.  no cycle if cycle of color~$i$ if~$i\not\doCycle)  is realizing the matchings with paths, te  together that share the same endpoints results in a collection of paths without cycles that have their endpoints~$\partial C$
%If 
%
%their matchings\todo{KF: This definition is not precise enough! 1) Matchings from the same or different sets should have the same endpoints? 2) When should they have the same endpoints? 3) What does it mean to have the same endpoint? Just same portal or even same slot within the same portal?} have the same endpoints on
%their shared boundary, and (ii) combining matchings sharing the same endpoints results in
%a set of paths with endpoints in $\partial C$\todo{KF: Unclear: What does combining matchings mean? How can a matching suddenly become a path? Why do we want paths?}.
Next, if the above conditions hold, we \emph{join} the matchings by contracting degree two vertices in the graphs above.
For graphs that were no cycles, this operation results in (new) matchings with endpoints on~$\partial C$ in the portals $\cpy(B',V_{{B'}})$. 
Finally, if the resulting perfect matchings~$M_1$ and~$M_2$ are noncrossing, we obtained a complete description of a subproblem for~$C$. 
We sum up the solutions costs of the respective subproblems of the four children and store them in our lookup table if this cost is the best found so far for this subproblem for~$c$.

For every quadruple~$\left((B_i,V_{{B'}_i},M^{(i)}_1,M^{(i)}_2)\right)_{i\in\{1,\dots,4\}},\doCycle_i$, it takes time polynomial in~$1/\eps$ to check whether it defines a valid subproblem for~$C$ (and to compute a solution to it). 

After completing the lookup table, our algorithm returns the cost for any subproblem for the root cell of the quadtree with~$\doCycle=\{1,2\}$.

\paragraph*{Analysis}\todo{KF: Check al formulas whether we use/don't use parentheses after $\log$ and after $\polylog$.}
Next, we show that this algorithm runs in~$2^{\Oh(1/\eps)} n \polylog(n)$ time.  
For each boundary, if we want to place a set of~$k$ portals, we have~$\binom{\Oh((\log L)/\eps)}{k}=\binom{\Oh((\log n)/\eps)}{k}$ possibilities in Case~\ref{def:fineArora} of Definition~\ref{def:fine} (assuming~$k\le 2$), and, for every feasible~$k'\ge k$ such that~$k \le 1/(\eps^2 k')$, we have~$\binom{\frac{1}{\eps^2 k'}}{k}$ possibilities in Case~\ref{def:fineSparsity}.
Thus, for every cell, the number of possible fine portal sets is
%\begin{displaymath}
%    \left(
%        \sum_{k=2}^{1/\eps} \binom{\frac{1}{\eps^2 k}}{k} + \sum_{k=0}^{2} \binom{\Oh(\log(n)/\eps)}{k}
%    \right)^4,
%\end{displaymath}
%\importanttodo[inline]{KF: Alternative:
%\begin{displaymath}
%    \left(
%        \sum_{2 \le k \le k' \textrm{ s.t. } k\le1/(\eps^2 k')} \binom{\frac{1}{\eps^2 k'}}{k} + \sum_{k=0}^{2} \binom{\Oh(\log(n)/\eps)}{k}
%    \right)^4,
%\end{displaymath}
\begin{displaymath}
    \left(
%        \sum_{k=0}^{2} \binom{\Oh((\log n)/\eps)}{2} + \sum_{k=2}^{1/\eps} \sum_{k'=k}^{1/(\eps^2k)} \binom{\frac{1}{\eps^2 k'}}{k} 
	\sum_{k=0}^{2} \binom{\Oh((\log n)/\eps)}{2} + \sum_{2 \le k \le k' \textrm{ s.t. } k\le1/(\eps^2 k')} \binom{\frac{1}{\eps^2 k'}}{k}
    \right)^4 = 2^{\Oh(1/\eps)} \polylog (n)~,
\end{displaymath}
using $\sum_{2 \le k \le k' \textrm{ s.t. } k\le1/(\eps^2 k')} \binom{\frac{1}{\eps^2 k'}}{k} = 2^{\Oh(1/\eps)}$ (see Claim~3.4
by Kisfaludi{-}Bak et al.~\cite{focs21}).
%Since~$k \le 1/\eps$ 
Note that the number of possible noncrossing matchings on~$k$ fine portals is known to be bounded by the~$k$-th Catalan number whose value is of order of~$\Oh(2^k)$.
Since~$k\le 1/\eps$ (as noted above) and given that the number of cells in the quadtree is $\Oh(n \log n)$, we conclude that
the number of states in our dynamic programming algorithm is bounded by
$2^{\Oh(1/\eps)} n \polylog(n)$. Observe that to get one level up in the dynamic
programming, we combine the solutions computed for the children by
iterating through all the states of the children cells. This iteration takes
$2^{\Oh(1/\eps)}\polylog (n)$ time as there are~$2^{\Oh(1/\eps)}\polylog (n)$ states for each child and
it takes time polynomial in~$1/\eps$ to check if the states are compatible. 

Finally, we analyze the approximation ratio of our algorithm. \todo{KF: We should give an overview of how we do it: First we analyze the cost of OPT-perturbed, then of OPT-Arora-portals respecting, then OPT-our-portals-respecting, \dots} Assume that
$\pi_1,\pi_2$ is an optimal solution of cost~$\opt$.
First, we argue that snapping points to~$\mathbb{Z}^2$
perturbs the solution by at most $\Oh(\eps\cdot\opt)$. 
Consider a point $p \in P_1$ (the case when $p \in
P_2$ is analogous) that got perturbed. Let $c$ be the position to which $p$ is moved. We
are going to add a segment $cp$ and $pc$ (infinitesimally close to each other) to the curve $\pi_1$ in order to guarantee that the point $p$
is visited after the perturbation. This, however, may cause that $\pi_2$ becomes
intersected. To avoid that, we are going to use Lemma~\ref{lem:patch2} on both
curves to guarantee that segment $cp$ is crossed~$O(1)$ times by the curve $\pi_2$. 
We route all these crossing points around~$cp$ in~$O(1)$ \enquote{layers}. 
%only through point $c$ \todo{KF: 1) We don't need that it goes through~$c$. 2) The lemma statement (currently) gives no guarantee that all the crossing points lie at exactly~$c$!}. 
This patching increases the total lengths of the curves by $\Oh(|cp|)$ which
is bounded by~$O(1)$.
Since there are only $\Oh(n)$
nonempty cells, the snapping to the grid increases the cost of the
solution by at most $\Oh(n) = \Oh(\eps L) = \Oh(\eps \opt)$. 

Thus, there exists a solution of cost~$\opt+\Oh(\eps\opt)$ to the perturbed instance. Without loss of generality, we can assume that it consists only of line segments with all endpoints lying in an infinitesimal neighborhood around the terminals. Since the terminals have integer coordinates, we can apply Theorem~\ref{thm:struct} for~$r=1/\eps$ and obtain a new solution  $\pi_1^A$ and $\pi_2^A$ that (i) is~$1/\eps$-simple, and (ii) of length bounded by~$\opt+\Oh(\eps\opt)$.

We claim that this solution is also fine-portal-respecting. If~$(\pi_1^A,\pi_2^A)$ uses~$u$ portals on a boundary~$F$ from~$\grid(F,g)$, then we set~$k=r^2/g$ (recall~$r=1/\eps$) in Case~\ref{def:fineSparsity} of Definition~\ref{def:fine} and obtain exactly the same grid for our fine portals. It remains to observe~$k\ge u$ as required by the definition. % enabling us to select the same portals. f~$u$ is the number 
Since~$u$ is not larger than the number~$m$ of times~$F$ is crossed, Definition~\ref{def:simple} implies~$g\le r^2/m \le r^2/u$ and therefore~$k=r^2/g \ge r^2\cdot u/r^2 = u$. 
%It's also not hard to see that, for~$r=1/\eps$, an~$r$-simple pair of tours is acutally a fine-portal respecting pair of tours.
%
%
%\importantquestion[inline]{KF: I don't understand the following paragraphs. Do we want to say that our DP outputs an \enquote{optimal}~$r$-simple pair of tours which, by our structure theorem, is close to (unconstrained) $\opt$?}
%Now, in order to use Theorem~\ref{thm:struct}, we observe that an  \importantquestion[inline]{KF: We forgot to say what the~$r$ is for our structure theorem!?} we need to take care of the case
%when the number of crossing is constant \todo{KF: why being constant is an issue?} or $\Omega(1/\eps)$. To achieve that, we
%use a structure theorem of Arora \importantquestion{KF: Which one? We should cite it and discuss somewhere!} combined \importantquestion{KF: In what way combined?} with our patching procedure from
%Lemma~\ref{lem:patch2}. 
%This guarantees us that there exists two tours $\pi_1^A$
%and $\pi_2^A$ that (i) are noncrossing (ii) their weight is larger by
%$\Oh(\eps\cdot (\wt(\pi_1) + \wt(\pi_2)))$ than $\wt(\pi_1) + \wt(\pi_2)$ an
%(iii) cross quadtree box at most $\Oh(1/\eps)$ times and only through two of
%$\Oh(\log(n)/\eps)$ portals.

%\todo{KF: We could use our term \enquote{$r$-simple} somewhere here (if it makes the description easier).}
Finally, recall that our algorithm computes the cost of an optimum fine-portal-respecting pair of tours. By the discussion above this cost is bounded by~$\wt(\pi_1^A)+\wt(\pi_2^A) = \opt + \Oh(\eps\opt)$ which concludes the proof of Theorem~\ref{thm:eptas-noncrossingtsp}.

It is easy to see that the algorithm can be derandomized by trying all
possibilities for $\mathbf{a}$ (but the cost increases by a polynomial in $n$ factor).

% $\pi'_1,\pi'_2$. Thus, our solution is bounded by the 
% that satisfies \importantquestion{KF: In what sense? That it visits the terminals~$P_1$ and~$P_2$? Or that it is~$r$-simple (but then why refer to the theorem and not to the definition of~$r$-simpleness?} the condition of Structural
%Theorem~\ref{thm:struct} for $\pi^A_1,\pi^A_2$. It means that
%$\mathbb{E}_\mathbf{a}[\wt(\pi^A_1) + \wt(\pi^A_2)) - \wt(\pi'_1) +
%\wt(\pi'_2))] = \Oh(\eps\cdot (\wt(\pi_1) + \wt(\pi_2)))$. Therefore
%$\wt(\pi'_1)+\wt(\pi'_2) = (1+\Oh(\eps)) (\wt(\pi_1)+\wt(\pi_2))$.  This
%concludes the proof of Theorem~\ref{thm:eptas-noncrossingtsp}.

%\todo[inline]{KF: The above argument was not clear enough. I think we want to show the following:
%\\1) By our DP, our solution $S_1$ is (almost) as good as the best solution~$S_2$ among all fine-portal respecting solutions.
%\\2) Automatically,~$S_2$ is the same as the best solution~$S_3$ among all solutions that are~$r$-simple for~$r=1/\eps$. 
%[Proof:
%\\Case (a) Holds by definition.
%\\Case (b) also holds: let~$F$ be a boundary with case (b).
%We just choose~$k=r^2/g$ and get~$\grid(F,g)=\grid(G,r^2/k)$. If~$u$ is the number of used portals, we have~$u\le m$ and thus in total~$g\le r^2/m \le r^2/u$ and therefore~$k=r^2/g \ge r^2\cdot u/r^2 = u$. Thus,~$k\ge u$ as required.]
%\\3) By our structure theorem,~$S_3$ is almost as good as the best unconstrained solution~$S_4$ for the perturbed instance that WLOG uses $\mathbb{R}^2$.
%\\4) By our observation,~$S_4$ is almost as good as $\opt$ (by transforming $\opt$).}

\section{Red-Blue-Green Separation}
\label{sec:rgbseparation}
\todo[inline]{KF: The problem has been already defined; thus we should remove this paragraph or rephrase it in the style \enquote{Recall that \dots}}
We present an EPTAS for the Euclidean \RBGS{} problem. In this problem, we
are given three sets of different points $R,G,B \subseteq \mathbb{R}^2$\note{KF: $\mathbb{R}$ is correct as only in the perturbed instances we have integer coordinates. The input is in reals.}. We say
that points in $R$ are ``red'', points in $G$ are ``green'' and points in
$B$ are ``blue''. We desire two simple noncrossing polygons $P_1$ and $P_2$
of smallest total length such that all classes of points are separated by
them, that is, for any two points~$p_i$ and~$p_j$ of different color, 
%points $p_r \in R$, $p_b \in B$ and $p_g \in G$ 
any path
from $p_i$ to $p_j$ must cross $P_1$ or $P_2$.\todo{KF: The sets are non-empty, right? Otherwise, we have as a special case the Red-blue-separation problem with two(!) cycles (one color can reside in two different cycles).}
\todo{KF: Does the connectivity requirement automatically follow? Independently of the answer, we should mention the connectivity requirement!}

\importantquestion[inline]{KF: Shouldn't this part be moved to the introduction? Is it even not already there?}
The problem is NP-hard\todo{KF: citation missing!}. To the best of our knowledge, only a special case\todo{KF: If we want it really to be a special case of our problem, then we formally have to allow that the colored sets can be empty and consequently require either that all points of the same color are in the same region or that the number of cycles is equal to the number of non-empty sets minus $1$.} 
of two\resolvedtodoNotDoneYet{KF: What "two"? two-colored}{Yes.}
red-blue separation has been previously considered. For the red-blue separation problem, Mata
and Mitchell~\cite{mata-michell} gave an $\Oh(\log n)$-approximation algorithm. Later, Arora
and Chang~\cite{red-blue} presented a PTAS for the red-blue separation problem. 

Our contribution is twofold. First, we generalize the result of Arora and
Chan~\cite{red-blue} to three colors. Second, we give an EPTAS 
with~$2^{\Oh(1/\eps)}n \polylog(n)$ running time.

%\begin{theorem}[Red-Blue-Green separation]
%    \label{thm:rgb-ptas}
%    Euclidean \RBGS{} admits an approximation scheme with
%    $2^{\Oh(1/\eps)} n \polylog(n)$ running time.
%\end{theorem}
\thmPTASrbg*

\todo{KF: Rephrase that sentence as in my opinion it sounds strange to say Now, \dots}
Now, we discuss the innovations that enabled us to generalize the result of
Arora and Chang~\cite{red-blue}. Their algorithm is based on the quadtree
framework that was introduced\todo{KF: added citations?} to give a PTAS for Euclidean TSP. 
We replace it
with a more efficient \emph{sparsity-sensitive patching} already introduced in
Section~\ref{sec:noncrossing-tsp}. Apart from minor modifications, this framework and the
proof of the existence of an EPTAS is almost exactly the same as in
Section~\ref{sec:noncrossing-tsp}. The main innovation is the introduction
of a new patching procedure that allows us to extend the result of Arora and
Chang~\cite{red-blue} to three colors. In the next subsection, we analyze
the red-blue-separation problem with the sparsity-sensitive patching framework. Next \todo{KF: next proper section?} we
will show our patching procedure for three colors.

\subsection{EPTAS for Red-Blue-Green Separation}\todo{KF: Since we have only one subsection, I think we should either split it in two subsections or merge it with its parent section.}

As a first step, we guess a topology of the optimum solution. We need to guess
whether one of the polygons is fully contained within the other or whether they are
separated\todo{KF: \enquote{separated} seems to me to be the wrong word. \enquote{disjoint} would be better, but one could misunderstand that we mean that only the border are disjoint. \enquote{interiors disjoint} is better but could suggest that their borders might touch.}. Moreover we need to guess a color of the points contained within each
polygon. Without loss of generality, let us assume that $P_1 \setminus P_2$ contains red
points and $P_2 \setminus P_1$ contains green points. Our algorithm consists of
several steps that replicate the steps from Section~\ref{sec:noncrossing-tsp}
(and mostly \todo{KF: In what sense mostly? That most of our steps actually come from that paper?} from~\cite{red-blue}).

\paragraph*{Perturbation} \todo{KF: We repeat the same text from Section 3! We should somehow refer to that section and only highlight the differences.} This step is analogous to the perturbation step introduced
in~\cite{red-blue} and the goal is to guarantee three properties: each point has
integral coordinate, the maximum internode distance is $\Oh(n/\eps)$ and
the distance between each pair of points is at least~$8$\todo{KF: Do we really need these three properties?}. 
The perturbation procedure
described in~\cite{red-blue} places a grid of \emph{small} granularity on the
instance and moves each point to the closest grid point. To separate points of
the same color we move red points to the closest north-east grid point $g_r$
with $x(g_r) \equiv 1 \; (\text{mod } 4)$, blue to north-west grid point $g_b$
with $x(g_b) \equiv 2\;(\text{mod } 4)$ and green points to the closest south
east grid point $g_g$ with $x(g_g) \equiv 3 \; (\text{mod }4)$.  Later, we argue
that the optimum solution to this perturbed instance is only negligibly more
costly than the solution to original instance.

Arora and Chang~\cite{red-blue} define the \emph{red bounding box} as a \importantquestion{KF: the smallest axis-aligned?} square
of length $L_r$ that contains all the red points such that the red bounding box
contains any polygon that separates red and blue points. They need to consider
it to have a lower bound on optimum, because it could be the case that blue
points are spread very far from red points. Because we assume that $P_1$
contains all the red points and $P_2$ contains all the green points we need to
also define \emph{green bounding box} of side-length $L_g$ that contains all
green points. Observe, that if green-bounding-box does not intersect
red-bounding-box then we can assume that we are given two independent dent
instances of instance with two colors (and use $L_r$ and $L_g$ as a lowerbounds
on $\opt$ in these instances). On the other hand, if red and green bounding boxes
intersect, we have a lower-bound $\opt > \min\{L_r,L_g\}$ and we
can use $L = \min\{L_r,L_g\}$ to lower-bound the $\opt$ and determine the granularity
(see~\cite{red-blue} for details). 

Next, in the step 2 we construct a randomly shifted quadtree. This step is
identical to the quadtree construction in Section~\ref{sec:noncrossing-tsp}
(note that no input point lies on a boundary of a quadtree after this procedure). We
draw a random $1 \le a_1,a_2 \le L$ \todo{KF: a random what? shifting vector?} and create a dissection $D(a_1,a_2)$.
\paragraph*{Structure theorem}
Our structure theorem here is completely analogous to the structure theorem introduced in Section~\ref{sec:alg-2noncross}. \todo{KF: Add some words that now instead of tour pairs we have polygon pairs?}
\importanttodo[inline]{KF: Rephrase definition to make it consistent with new definiton for TSP.}
\begin{definition}[$r$-simple pair of polygons]
Let $\mathbf{a}$ be a random shift vector. A pair of polygons $(P_1,P_2)$ is \emph{$r$-simple} if it is noncrossing and, for any boundary~$F$ of the dissection $D(\mathbf{a})$ crossed by the pair $(\bd P_1, \bd P_2)$, 
%either
	\begin{enumerate}[label=(\alph*)]
		\item it crosses $F$ entirely through one or two portals belonging to~$\grid(F, \lfloor r\log L \rfloor)$, or %(and only the interior of~$F$), or
        \item it crosses $F$ entirely through portals belonging to~$\grid(F, g)$, for~$g\le r^2/m$ where~$m$ is the number of times~$F$ is crossed.
	\end{enumerate}
	Moreover, for any portal~$p$ on a grid line $\ell$, the pair~$(\bd P_1, \bd P_2)$ crosses~$\ell$ at most ten times through~$p$. 
\end{definition}
\begin{theorem}[Structure Theorem]
    \label{thm:str-polygon}
    Let $\mathbf{a}$ be a random shift vector, 
	 and let $(P_1,P_2)$ be a pair of simple noncrossing polygons
    consisting only of line segments whose endpoints lie in an infinitesimal neighborhood of~$\mathbb{Z}^2$. 
    For any large enough integer $r$, there is an $r$-simple pair
    $(\hat{P}_1,\hat{P}_2)$ of noncrossing simple polygons %visiting $P_1$ and $P_2$ respectively 
    whose borders that differ from~$(\bd P_1, \bd P_2)$ only in an infinitesimal neighborhood around the grid lines
    and that satisfies 
    $$\mathbb{E}_{\mathbf{a}}[\wt(\bd \hat{P}_1) + \wt(\bd \hat{P}_2)-\wt(\bd P_1) - \wt(\bd P_2)] = \Oh((\wt(\bd P_1) +  \wt(\bd P_2) )/ r).$$

    Moreover if $(P_1,P_2)$ are separating points $R,G,B$ then so are $(\hat{P}_1,\hat{P}_2)$.
\end{theorem}

The proof of Theorem~\ref{thm:str-polygon} follows the framework used in the
Theorem~\ref{thm:struct}.

\paragraph*{Dynamic Programming} Now we are going to briefly sketch the
dynamic programming assuming Theorem~\ref{thm:str-polygon}. The dynamic
programming will enable us to exactly \todo{KF: In what sense "exactly"?} find two $\Oh(1/\eps)$-simple and noncrossing
polygons in the given quadtree. This fact combined with the properties
of Theorem~\ref{thm:str-polygon} will guarantee that these polygons are
$(1+\eps)$ approximation.

At a high level, the states of the dynamic programming are very similar to
ones defined in Section~\ref{sec:alg-2noncross}. The subproblems of the dynamic
programming are defined as follows:

\defproblem{Separating $s$-simple polygons}
{A nonempty cell $C$ in the shifted quadtree,  a fine portal set $B \subseteq
\partial C$, a matrix $V_B$, and (ii) two perfect noncrossing matchings $M_1,M_2$ on $B_1=\{b \in \cpy(B,V_B): \col{b}=1 \}$ and $B_2=\{ b \in \cpy(B,V_B): \col{b}=2 \}$
respectively (iii) a $3$-coloring which indicates whether each region
    defined by matchings $M_1$ and $M_2$ is inside $P_1 \setminus P_2$, inside
$P_2\setminus P_1$ or neither}
{Find two $r$-simple path collections
    $\mathcal{P}_{B_1,M_1},\mathcal{P}_{B_2,M_2}$ of minimum total length that
    satisfies the following properties for all $i \in \{1,2\}$
\vspace{-0.5em}
\begin{itemize}
    \setlength\itemsep{0.0em}
    \item All points of class $R$ (respectively $G,B$) in cell $C$ are inside the region colored $R$ (respectively $G,B$),
    \item Paths of $\mathcal{P}_{B_i,M_i}$ are entirely contained in cell $C$
    \item $\mathcal{P}_{B_i,M_i}$ realizes the matching $M_i$ on $B_i$.
\end{itemize}
}

The size of the lookup table is the number of subproblems. Number of possible colourings (iii) for a fixed matchings and
portals is $2^{\Oh(1/\eps)}$. We showed in Section~\ref{sec:alg-2noncross} that the number of fine sets $B$, the corresponding matrices $V_B$ and non-crossing matchings $M_1$ and $M_2$ ((i) and (ii)) is $2^{\Oh(1/\eps)} n \polylog(n)$. Thus, the total number of states for the dynamic programming here is also $2^{\Oh(1/\eps)} n \polylog(n)$.

% Number of cells in the quadtree is $\Oh(n \log {n})$.
% The number of possible matchings is $2^{\Oh(r)}$ and the number of possible
% choices for portal selection is $(\sum_{k \in \{1,\ldots,r\}} \binom{r^2/k}{k} +
% \log{n})^4 = 2^{\Oh(r)} \log^4{n}$. Therefore the total number of states is
% $2^{\Oh(r)} \cdot n \log^5{n}$.

Each subproblem is solved by combining the previous one in the usual bottom-up
manner. For a fixed state in the quadtree cell we iterate through all the
solutions in 4 children in the quadtree. For each tuple we check if
selected portals match, if corresponding matchings are connected and whether
coloring across cell is consistent. If so, we remember the one with the smallest
total length. 
\importanttodo{KF: Analogous to my comment in TSP section, each 
subproblem should, for each color, also store the information whether it fully 
contains the polygon of that color.}

It remains to describe the base case (which is analogous to the dynamic
programming procedure described in Section~\ref{sec:alg-2noncross}). For a quadtree
with a single node. Recall that all the input vertices are infinitesimally close
to the corner of the cell. If the cell is empty, the solution is valid if the
pairs of connected portals can be connected with a noncrossing matching. If the
cell contains single point, we need to check that the region selected in the
state matches the color of the point. Otherwise, our polytopes may need to
``bend''\importantquestion{KF: And how do we compute the correct bend in polynomial time?} 
on the points in order to separate the region. Observe that if the
points of different colour have distance 
\importantquestion{KF: How can this happen? By perturbation (and alone, our assumption of 
the base case), we have at most one terminal in the cell!} 
$0$, then all we need to do is make
polygon go through it.

Now we argue that this procedure returns a $(1+\eps)$ approximation to the
original problem. During a perturbation step we shifted the points of the
different colors to the different corners of the grid of granularity $\Oh(\eps
\opt/n)$. Consider a point $p$ inside a grid cell $C$. Let $c \in C$ be
the corner of the grid towards which the point $p$ is moved to by a snapping
procedure.  We add to optimum solution to the original instance the segments
$cp$ and $pc$. This ensures that point $p$ is inside a proper polygon.  Note,
that this may create an additional crossings. To avoid that, we use
Lemma~\ref{lem:patch2} on segment $cp$. This guarantees that there exist two
noncrossing polygons that does cross segment $p c$ only on point $c$. This
operation increases the lengths of $P_1$ and $P_2$ by at most $\Oh(|p c|)$.
In total this is bounded by the side-length of the grid which is $\Oh(\eps
\opt/n)$.  There are at most $\Oh(n)$ cells on which we do this
operation. Hence after perturbation the length increases by at most $\Oh(\eps
\opt)$.

The analysis of the approximation factor due to the dynamic programming follows
analogously to the analysis in Section~\ref{sec:noncrossing-tsp}.

\section{Approximation Schemes in Planar Graphs}
\label{sec:ptas-graphs}

In this section, we complement our results\todo{KF: which ones? In what sense do we complement them? By giving PTASes or by giving hardness proofs? The second part of the sentence suggests the former.}, and demonstrate that our new patching lemma (Lemma~\ref{lem:patch2})
gives approximation algorithms in planar graphs.

\thmPTAStsp*
%\begin{theorem}
%    \label{thm:planar-graphs}
%    \BTSP in plane unweighted graphs admits an
%    $(1+\eps)$ approximation in $f(\eps) n^{\Oh(1/\eps)}$ time.
%\end{theorem}

We use a framework proposed by Grigni et al.~\cite{papa95} that enabled them to
give a PTAS for TSP in planar graphs. Grigni et al.~\cite{papa95} proposed the
following binary decomposition tree $\Tt$ of planar graph $G$. The decomposition
is parametrized with $f,d \in \nat$ that intuitively control approximation and
depth of decomposition. Each node of $\Tt$ contains a subgraph of $G$. At the
root of $\mathcal{T}$ is $G$ itself. For each node $H \in \mathcal{T}$,  there
exist $\Oh(f)$ vertex disjoint paths in $H$ (we call these paths \emph{portal
paths}). Removal of this paths partitions $H$ into $H_1$ and $H_2$, such that
$H_1,H_2 \subseteq H$, and $|H_1|,|H_2| <\frac{5}{6} |H|$. These subgraphs
$H_1,H_2$ are then put recursively as two children of $H$ in the decomposition
$\Tt$.  The construction continues until at the leaves of $\mathcal{T}$ are
subgraphs of size $s = \Oh(f^2)$. The depth of tree $\mathcal{T}$ is $d$ and the
total length of all \emph{portal paths} in $\Tt$ is $\Oh(\frac{d}{f} |G|)$ (see
Fig.~\ref{fig:portal-paths} for one step of decomposition).

\begin{figure}[ht!]
    \centering
    \includegraphics[width=0.3\textwidth]{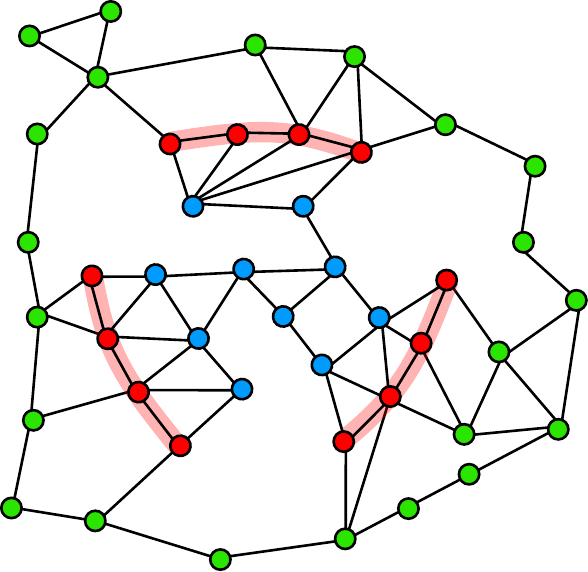}
    \caption{Figure depicts the graph $H$ as a node in a decomposition of Grigni et
    al.~\cite{papa95}. We have three portal paths (drawn in red) for a
    balanced planar separator that partitions graph $H$ into $H_1$ (drawn in blue) and
graph $H_2$. The graph is recursively decomposed until subgraph is of size $s$.
The patching lemma is applied on portal-paths to guarantee that solution is
crossing $H_1$ and $H_2$ few times.}
    \label{fig:portal-paths}
\end{figure}

\begin{lemma}[cf,~\cite{papa95}]
    \label{lem:decomposition}
    Given a planar graph $G$, in polynomial time one can compute a decomposition
    $\mathcal{T}$ with the following properties parameters $f :=
    \Theta((\log n)/\eps)$ and $d := \Theta(\log n)$. 
    \todo[inline]{KF: Write the time more accurately, say polynomial in what parameter(s)? That~$\eps$ treated as a constant?}
\end{lemma}

We build an approximate solution to \BTSP with a dynamic
programming in the similar fashion as~\cite{papa95}. For a sake of presentation
let us focus For each graph $H \in \mathcal{T}$ and information about
connectivity of cycles we store a partial
solution in the dynamic programming table. More precisely, each entry of dynamic programming
contains a following states: (i) a current node $H$ from $\Tt$, (ii) set $\Pp$ of at most
$\Oh(f)$ portal paths in $H$, (iii) set $\Cc(H,\Pp)$ of configurations on
portal paths in node $H$ that intuitively stores the
information about the connectivity of the solution and how it looks ``outside''
$H$. More precisely, in $\Cc(H,\Pp)$ we store a set of $\Oh(|\Pp|)$ ``noncrossing pairings''
between portal paths. Each pairing has either blue or red color and is
connecting exactly two portal paths. Additionally, a single portal path can be a part of at
most $\Oh(1)$ pairings. Moreover, pairing can be represented as the planar
graph (i.e., pairings are noncrossing). 

\begin{remark}
The description of the subproblem is exactly the same as in
Section~\ref{sec:noncrossing-tsp}, however we cannot artificially increase the
number of portals (hence each portal would need to be a part of $\Oh(1)$ number
of portals).
\end{remark}

To solve the problem, we build our dynamic table bottom up. We start with the
description of base case.

\paragraph*{Base Case}

Let $H$ be a leaf of $\mathcal{T}$. Based on the planar separator theorem, we give an $(1+\eps)$-approximation
algorithm in $f(\eps) \cdot 2^{\Oh(\sqrt{|H|})}$ time to solve base case. On the
input, we are given a set $\Pp$ of portal paths and a information
about the connectivity~$\Cc(H,\Pp)$ between them. First, we compute a balanced
separator of $H$ of size $\Oh(\sqrt{H})$~\cite{planar-sep}. For each vertex on
the separator, we guess how many times the cycle of each color is crossing it. By
Lemma~\ref{lem:patch2}, we have a guarantee that there exists a feasible solution that
crosses each vertex at most $\Oh(1)$ times. After that, we recurs into
subgraphs determined by the planar separators until the size of the graphs is
$\Oh(1/\eps)$ (where we use an exhaustive algorithm in $f(\eps)$ time). The
running time is:

\begin{displaymath}
    T(|H|) = 2^{\Oh(\sqrt{|H|})} \cdot T(2|H|/3) \le f(\eps) 2^{\Oh(\sqrt{|H|})}
    .
\end{displaymath}

We need to analyse the approximation factor of our algorithm. Observe, that on a
single call
our patching cost is $\Oh(\sqrt{H})$, because we use Lemma~\ref{lem:patch2}
to bound number of crossings of vertices on the separator. Therefore, the
approximation ratio is $E(|H|) = \Oh(\sqrt{|H|}) + 2T(2|H|/3)$ and $E(1/\eps) =
0$. This recursive equation results in the cost $E(|H|) = \Oh(\eps|H|)$, which concludes the
correctness of the base case.

\paragraph*{Combination of Subproblems}

We combine two solutions from the lower levels of $\Tt$ by checking if the
states are \emph{compatible} and selecting the solution of the lowest cost. We
define the states to be compatible as in
Section~\ref{sec:noncrossing-tsp}. For the running time, it is enough to compute
the number of states in the dynamic programming table. The main contribution comes from the number of
possibilities for $\Cc(H,\Pp)$. Note, that the number of noncrossing pairings is
$2^{\Oh(|\Pp|)} \le n^{\Oh(1/\eps)}$ (because each pairing can be represented
as \todo{KF: a?} noncrossing matching on $\Oh(|\Pp|)$ vertices). This matches the runtime \todo{KF: use consistently the term running time everywhere} 
of the computation of the base case.

\paragraph*{Approximation Ratio} 

It remains to bound the approximation cost. The approximation error comes
with the assumption that our solution is crossing each portal path $\Oh(1)$
number of times and with \todo{KF: I don't understand. Is this also an assumption?} 
solving each base case. For the base case error observe that
the total size of each leaves in $\Tt$ is $(1+\Oh(1/\eps)) n$ and in the base
case we incur $\Oh(\eps)$ error per vertex.  Therefore, we are left to
analyse the approximation incurred by invoking patching lemma. Observe, that
Lemma~\ref{lem:patch2} given an optimum tour $\pi$, guarantees that there exists
a tour $\pi'$ that crosses each portal path $\Oh(1)$ and the cost of $\pi'$ is only
constant times longer than the total length of all portal paths. The
decomposition lemma (see Lemma~\ref{lem:decomposition}) guarantees that the
total length of portal paths is $\Oh(\eps n)$. Because optimum is at least
$n/2$, we can select constant in front of $\eps$ to guarantee $(1+\eps)$
approximation. This concludes the correctness of our algorithm.

\begin{remark}
    Essentially, the same arguments can be used to give a polynomial time
    approximation scheme for \BST\resolvedtodoNotDoneYet{KF: in planar graphs, right?}{Yes, and add this info everywhere where it is missing.}. \todo{KF: Is the problem already defined, or is it defined in the last section only? Then we should give some reference?} 
    The main difference is that we need\todo{KF: can use?} to use
    a significantly simpler patching Lemma\todo{KF: Lemma with small letter, check also other such occurrences.} for two colored trees
    of Bereg et al.~\cite{isaac15} that allows additional Steiner vertices in portals.
\end{remark}

\thmPTASbst*

\section{NP-hardness of Bicolored Noncrossing Spanning Trees}
\label{sec:bst-np}

In this section we consider the following problem. \todo{KF: We already considered it before in a remark concerning a PTAS. We should say that we consider it here with respect to a hardness result.}

\BST{}

INPUT: a plane graph $G$ and a bipartition $(B,R)$ of $V(G)$, and an integer $k$.

PROBLEM: Are there two noncrossing trees $\Tt_B$ and $\Tt_R$ such that $\Tt_B$
spans $B$, $\Tt_R$ spans $R$ and $||\Tt_B|| + ||\Tt_R|| \le k$.

\thmNPbst*
%\BST{} is NP-complete/hard? in plane graphs.

\begin{proof}

The problem is in NP, since given the trees as a witness, it is easy to check
that they form a solution. Let us show that it is NP-hard.

We reduce from {\sc (Unweighted) Steiner Tree} in planar
graphs~\cite{steiner-tree-np-hard}. Consider an instance of {\sc Steiner Tree}, that is a graph $G$, a set $P$ of vertices called \emph{terminals}, and an integer $k$. We can assume that $G$ is connected.

Consider a planar embedding of $G$ in the Euclidean plane. Let $G'$ be the graph obtained from $G$ by subdividing every edge once. One can easily get an embedding of $G'$ by adding the new vertices anywhere on their edges (besides their endpoints).

Let $f$ be a face of $G$ of degree $d$, and $C_f$ be the closed curve bounding
$f$. Let $C'_f$ be a cycle obtained from $C_f$ as follows: 
\begin{enumerate}[nosep=0pt, label=(\roman*)]
\item Choose a point in $C_f$ that is not a vertex. Follow the curve $C_f$, and every time you encounter a vertex, add a new copy of that vertex. Let $v_1$, \dots, $v_d$ be the vertices created this way, in this order.
\item For $i \in \{1,\ldots,d-1\}$, let $v_i v_{i+1}$ be an edge. Also let $v_1 v_d$ be an edge.
\item For $i \in \{1,\ldots,d\}$, add one edge connecting $v_i$ to the vertex of $C_f$ it was created as the copy of.
\end{enumerate}
Note that for any face $f$, $C'_f$ can be embedded inside $f$. 

Let $G''$ be the graph $G'$ such that, for every face $f$ of $G'$, the cycle $C'_f$ is added to the graph. Note that $G''$ is still a planar graph.  Note that every vertex of $G'$ of degree $d$ has exactly $d$ copies in $G''$.
Let $B=P$ and $R = V(G'')\setminus P$. Let $\ell = 2k + |G''| - |P| - 1$. This concludes the description of our reduction.

\paragraph*{Correctness}  Let us show that $(G,P,k)$ is a
yes-instance to {\sc Steiner Tree} if and only if $(G'',B,R,\ell)$ is a yes-instance to \BST{}.

Suppose that $(G,P,k)$ is a yes-instance to {\sc Steiner Tree}. Let $\Tt$ be a Steiner tree of size $k$ in $(G,P)$. 
Let $\Tt'$ be $\Tt$ with every edge subdivided once in $G'$. Note that $\Tt'$ is a tree using only vertices of $G'$ spanning $B = P$ in $G''$. 

Consider the dual $G^\star$ of $G$, and a spanning tree $\Ss$ of $G^\star$ that does not
intersect $\Tt$. Let us build a tree $\Ss'$ in $G''$ as follows:
\begin{enumerate}[nosep=0pt, label=(\roman*)]
\item Let $V(\Ss') = R$
\item For every face $f$ of $G'$ of degree $d$, for $i \in \{1,\ldots,d-1\}$, add
    $v_iv_{i+1}$ to $\Ss'$.
\item For every edge $e$ of $\Ss$, let $u_e$ be the vertex of $G'$ subdividing
    the dual edge of $e$ in $G$, and let $u_e'$ and $u_e''$ be the two copies of
    $u_e$ in $G''$. Add $u_eu_e'$ and $u_eu_e''$ to $\Ss'$.
\item For other vertex of in $V(G') \setminus P$, add one edge connecting it to one of its copies.
\end{enumerate}

Note that $\Ss'$ is well defined since every edge that we added to $\Ss'$ has both of its endpoints in $R$ (either it is not a vertex of $G$ or it is in $V(G) \setminus P$).

\begin{claim}
$\Ss'$ is a tree.
\end{claim}
\begin{proof}
Let us first prove that $\Ss'$ is connected.
For every face $f$ of $G'$, $\Ss'[V(C'_f)]$ is a path. Since $\Ss$ is a tree of
$G^\star$, it connects every face of $G$. Thus $\Ss'$ connects all of the $\Ss'[V(C'_f)]$, and thus every vertex of $V(G'')\setminus V(G')$. Lastly, every vertex of $V(G') \setminus P$ is connected to at least one of its copies.

 Now, let us prove that $\Ss'$ has no cycle. The fourth point in the
 construction of $\Ss'$ adds vertices of degree $1$ in $\Ss'$, so they can not
 be in a cycle and we do not need to consider them. Consider for contradiction a
 cycle $C$ in $\Ss'$. Since for every face $f$ of $G'$, $\Ss'[V(C'_f)]$ is a path,
 if we contract each $C'_f$, the graph resulting from $C$ still has a cycle. By
 definition, this cycle is a subdivision of a cycle in $\Ss$, a contradiction.
\end{proof}

Since $\Ss'$ is a tree, it has exactly $|R| - 1 = |G''| - |P| - 1$ edges. The only
vertices of degree at least $2$ in $\Ss'$ that are vertices of $G$ are vertices
subdividing dual edges of edges of $\Ss$. Since $\Ss$ does not cross $\Tt$, it
follows that $\Ss'$ does not cross $\Tt'$. Lastly, $|\Tt'| = 2|\Tt| \le 2k$.
Therefore $\Ss'$ and $\Tt'$ are witnesses that $(G'',B,R,\ell)$ is a yes-instance to \BST{}.

Now assume that $(G'',B,R,\ell)$ is a yes-instance to \BST{}. Let $\Tt'$ be the
tree covering $B$ and $\Ss'$ be the tree covering $R$. Since $\Ss'$ needs to cover
$R$, it follows that $||R|| \ge |R| - 1 = |G''| - |P| - 1$, therefore $|\Tt'| \le 2k$. 

Let us modify $\Tt$ as follows: for every copy of a vertex in $\Tt'$, replace it
by the vertex of $G'$ it is a copy of (both as a vertex and as an endpoint of
its edges). If this would lead to loops, remove them, and if this would lead to
multi-edges, replace it with a single edge. This leads to a subgraph $H$ of
$G'$, since two copies of vertices are adjacent only if they are copies of
adjacent vertices. This may lead to reducing the number of edges of $\Tt'$ (since we may remove some loops or multi-edges) but may not increase it. The graph $H$ is a connected graph spanning $B$.
 In $H$, contract every vertex of $V(G') \setminus V(G)$ to one of its neighbors
 in $H$. This divides the number of edges of $H$ by at least two. Lastly, take a
 spanning tree $\Tt$ of $H$. We have $|\Tt| \le \frac{|\Tt'|}{2} \le k$. Therefore $(G,P,k)$ is a yes-instance of {\sc Steiner Tree}.

This concludes the correctness proof of our reduction.
\end{proof}

\section{\MSP~on the Plane}
\label{sec:msp-np}

In this section we consider the following problem.

\MSP{} \todo{KF: We should have the same style of problem definition everywhere I guess.}

INPUT: a set of pairs of points on the plane, a multi-set of pairs of points $P \subseteq \mathbb{R}^2$ called terminals, and a real number $\ell \in \mathbb{R}$. 

PROBLEM: Are there $|P|$ noncrossing paths, each one linking \todo{KF: connecting?} a different pair of points in $P$, such that their total length is a most $\ell$?

\thmNPpaths*
%\begin{theorem} \label{thm:MSPhard}
%Euclidean \MSP{} is NP-hard.
%\end{theorem}

In the rest of this section, we will describe a proof of Theorem~\ref{thm:MSPhard}.
\todo{KF: I guess we should (again) mention the issue with the paper of Fekete et al.}

\subsection{Preliminary Remarks} \label{prem_rem}

First note that we can consider instances where several terminals are identical.
As in our definition section \todo{KF: refer and rephrase \emph{definition section}} 
we require paths of different colors to be
disjoint, also the terminals need to be disjoint. \todo{KF: I don't understand the following sentence.} An actual instance to consider
would be to actually consider disjoint terminals, but infinitesimally close to
each other, and infinitesimally close to the positions described in the
remainder of this proof. In such an instance, the solutions would be arbitrarily
close to the one \todo{KF: But such a solution does not exists as it would contain non-disjoint paths!} 
for the instance where terminals have the same position.
Therefore, we can select the terminals in such a way that the optimal solution
will still be as described, and our NP-hardness reduction will still work. 

The next remark is that an optimal solution never uses a non-terminal as an
inflexion point. If it did, then among the paths using non-terminal $v$, choose
the one with the sharpest angle, say $p$.  By shifting one of the paths $p$ that
visits $v$ slightly toward the interior of the angle (keeping the other
inflexion points identical), one would reduce the length of $p$. If this does
not lead to a feasible solution (because it needs to be noncrossing), it means
that other paths used $v$ with the exact same angle, and we can change the
inflexion point $v$ for all of them similarly, and get a better solution.

\subsection{General Strategy} \label{gen_strat}

We will reduce from {\sc Max 2-SAT}. In this problem, an instance is a formula
in conjunctive normal form where every clause has exactly two literals, and we
want to know the largest number of clauses we can satisfy. This problem is known
to be NP-hard~\cite{complexity-book}.

Let $k$ be the number of clauses and $n$ be the number of variables in the instance.

We want to partition pairs of terminals into four levels (or groups) such that
these levels essentially do not interfere with each other. More formally, we
want that in an optimal solution, all paths of level $i$ are exactly those of an
optimal solution of the instance where terminals of higher levels are
removed.

Let us assume that we have an instance of the construction that will be defined in the following section, with $n_i$ pairs of terminals of level $i$ for $i \in \{1,2,3,4\}$. Assume that all of the $n_i$'s are polynomial in term of the size of the original instance of {\sc Max 2-SAT}.

In Section~\ref{my_analysis}, we make sure that there exists a constant $c_4 >
0$ such that any path of level 1, 2, or 3 varying from the optimal solution
where vertices of level 4 are removed will increase the sum of the lengths of
those paths by at least $c_4$. Note that, as argued in the previous subsection,
we only need to consider changes that make a path go around some terminal, as
other changes can not lead to an optimal solution.

Now we analyse cost incurred by paths of level $4$ by making them avoid
paths on the remaining levels.
Note that, each path of level 4 may be increased by at most a constant $C_4$ (from those of a solution without paths of level 1, 2, or 3). Namely, $C_4$ is at most twice the sum of the lengths of all of the paths of level 1, 2, or 3 in an optimal solution where terminals of level 4 are removed, in order to go around those paths. Now, we will take several copies of each pair of terminals of level 1, 2, and 3. Note that taking copies of already existing paths of levels 1, 2, and 3 may not further perturb the paths of level 4, that already avoid them. Thus, by taking at least $\lceil\frac{C_4n_4}{c_4}\rceil + 1$ copies of each pair of terminals of level at most 3, we make sure that it is more costly for paths of level at most $3$ to be altered in any way than for paths of level $4$ to avoid those paths. By doing that, we make sure that every optimal solution of the whole instance coincides, on paths of levels at most $3$, with an optimal solution of the instance where terminals of level 4 are removed.
Let $n_1' = (\lceil\frac{C_4n_4}{c_4}\rceil + 1)n_1$, $n_2' = (\lceil\frac{C_4n_4}{c_4}\rceil + 1)n_2$, and $n_3' = (\lceil\frac{C_4n_4}{c_4}\rceil + 1)n_3$ be the new number of terminals of levels 1, 2, and 3, respectively.

Now consider an optimal solution of the instance where terminals of level 4 are removed.
Similarly to what is above, in Section~\ref{my_analysis}, we will make sure that there exists a constant $c_3 > 0$ such that any path of level 1 varying from the optimal solution where vertices of level 3 are removed will increase the sum of the lengths of those paths by at least $c_3$. By being made to avoiding paths of levels 1 and 2, each path of level 3 may be increased by at most a constant $C_3$ (from those of a solution without paths of level 1 or 2). By taking at least $\lceil\frac{C_3n_3'}{c_3}\rceil + 1$ copies of each pair of terminals of level 1 and 2, we make sure that it is more costly for paths of level $1$ and $2$ to be altered in any way than for paths of level $3$ to avoid those paths. 

We do that one last time, increasing the number of copies of pairs of terminals of level 1 according to a constant $c_2$. Note that the number of terminals of each level is still a polynomial in terms of the size of the original instance of {\sc Max 2-SAT}.

The idea is that paths of level 1 up to 3 will depend only on whether each variable is positive or negative in a solution, and always have exactly the same sum of lengths. The paths of level 4 will be used for each of the clauses, and will be slightly longer for the clauses that are not validated (and always have the same length otherwise). That way, the sum of the lengths will be an affine function of the number of non-validated clauses.

\subsection{Description of the Gadgets} \label{my_analysis}

We will need a series of gadgets. 
\begin{figure}[h!]
    \begin{center}
      \begin{tikzpicture}[scale = 1.5]
        \coordinate (u) at (0,0);
        \coordinate (x) at (1,1);
        \coordinate (w) at (1,-1);
        \coordinate (v) at (2,0);

	\draw (u) node [left] {$u$} ;
	\draw (v) node [right] {$v$} ;
	\draw (w) node [below] {$w$} ;
	\draw (x) node [above] {$x$} ;

        \draw[blue, very thick] (w) -- (x);
        \draw[red,thick] (w) -- (v);
        \draw[red,thick] (w) -- (u);
        \draw[red, dashed,thick] (x) -- (v);
        \draw[red, dashed,thick] (x) -- (u);
      \end{tikzpicture}
\end{center}
\caption{Gadget $G1$. The level 1 path $wx$ (in blue) has to be a straight line. Therefore the level 2 path $uv$ needs to eitheir go above it (dashed red line) or below it (full red line). \label{G1}}
\end{figure}
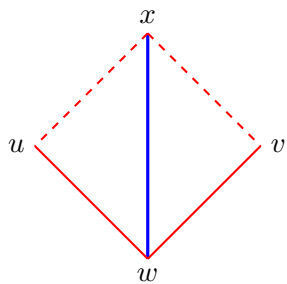

\paragraph*{First Gadget ($G1$) (Fig.~\ref{G1}):} \todo{KF: Use subscript, i.e., $G_1$ etc.?}

A pair of level 2 terminals $u$ and $v$, placed horizontally at a distance $1 \le d \le 2$, say at coordinates $(0,0)$ and $(d,0)$ respectively; and a pair of level 1 terminals $w$ and $x$ at distance $1$ above and below the axis $(u,v)$ at mid-distance, say $w$ at coordinates $(\frac{d}{2}, -1)$ and $x$ at coordinates $(\frac{d}{2},1)$.
Note that we will be able to vary $d$ for each instance of the gadget.

Assume that the closest terminal to the segment $[w,x]$ that is not on $[w,x]$ is at distance at least $\frac{1}{4}$ from $[w,x]$. We will make sure that every pair of terminals of level 1 is part of a gadget $G1$. Thus they do not interfere with each other, and can all be straight lines together.
We know, because of the discussion in Section~\ref{prem_rem}, that only
terminals are used as inflexion points. Therefore, if we do not use a straight
line between $w$ and $x$, we know that we must at least use the closest
terminal, which will be at distance at least $\frac{1}{4}$ from the segment
$[w,x]$. Therefore in that case, we would increase the length of this path by at
least $2\sqrt{\frac{1}{16} + 1}  - 2 > 0.06$. Note that the important part is
that this is a positive constant. In the similar reasoning later on, we will not
be explicit about the constants.
We can take $c_2 = 0.06$, $c_3 \le 0.06$, and $c_4 \le 0.06$, and we know that if any level 1 path is not a straight line, it will worsen the solution by at least $c_2$, $c_3$ and $c_4$, as required in Section~\ref{gen_strat}. Therefore we can assume that paths of level $1$ are straight lines, and other paths go around them.
In particular, in one instance of gadget $G1$, the path $uv$ can either go above or below $[w,x]$, as in Fig.~\ref{G1}. Note that the two possibilities have the same length.

%Informally, in this first gadget, we take $a$ large enough so that, in the optimal solution, the paths corresponding to those vertices will have to be straight paths between their terminals. 

%More formally, assume that the only path that interferes with the paths $[w,x]$ is that between $u$ and $v$, and that the closest terminal to $[w,x]$ that is not on $[w,x]$ is at distance at least $\frac{1}{4}$.
%We know, because of the discussion above the statement of the Themorem, that only terminals are used as inflexion points. 
%since these paths will only interfere with the path between $u$ and $v$, and this one is lengthened by at most $4$ (going around the segment $[w,x]$). However, if we do not use a straight line between $w$ and $x$, we know that we must at least use the closest terminal, which will be at distance at least $\frac{1}{4}$ from the segment $[w,x]$. Therefore in that case, we would augment the length of this path by at least $2\sqrt{\frac{1}{16} + 1}  - 2 > 0.06$. The important part is that the length of the path is increased by at least a positive constant. As we would lose as much for each instance of the pair $(w,x)$, we lose $a$ times some poisitive constant (here more than $0.06$). Therefore if $0.06 \times a \ge 4$, i.e. if $a \ge 67$, we know that the paths from $w$ to $x$ will go straight from $w$ to $x$.
%Note that this is true also if we take several copies of the gadget $G1$ with the same positions for the terminals: every weight gain and loss is then multiplied by the number of copies, and the difference is still positive when $a \ge 67$.

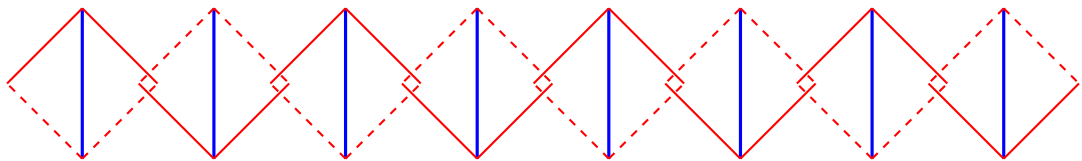
\begin{figure}[h!]
    \begin{center}
      \begin{tikzpicture}
	\def\n{8}

	\foreach\x in {1,...,\n}{ 	\coordinate (u\x) at (1.75*\x,0);
        						\coordinate (x\x) at (1.75*\x+1,1);
        						\coordinate (w\x) at (1.75*\x+1,-1);
        						\coordinate (v\x) at (1.75*\x+2,0);

        						\draw[blue, very thick] (w\x) -- (x\x);
						\pgfmathparse{int mod(\x,2)}
  						\ifnum\pgfmathresult=0{
                                \draw[red,thick] (w\x) -- (v\x);
                                \draw[red,thick] (w\x) -- (u\x);
                                \draw[red, dashed,thick] (x\x) -- (v\x);
                                \draw[red, dashed,thick] (x\x) -- (u\x);}
						\else{
                                \draw[red, dashed,thick] (w\x) -- (v\x);
                                \draw[red, dashed,thick] (w\x) -- (u\x);
                                \draw[red,thick] (x\x) -- (v\x);
                                \draw[red,thick] (x\x) -- (u\x);
						}
						\fi}

      \end{tikzpicture}
\end{center}
\caption{Variable gadget. The levels are such that the level 2 paths (in red) either all go along the full lines, which corresponds to the variable being true, or all go around the dashed lines, which corresponds to the variable being false. \label{Gvar}}
\end{figure}

\paragraph*{Variable Gadget (Fig.~\ref{Gvar}):}

On a line, put enough, say $20(k+1)$, instances of $G1$, each initially with $d
= \frac{3}{2}$. Align them so that the terminal $u$ for the first gadget is at
coordinates $(0,0)$ for the variable gadget; then for each $i\in\{2,...,10c\}$,
the terminal $u$ for the $i$-th instance of $G1$ is $\frac{1}{4}$ to the left of
the terminal $v$ for the $(i-1)$-st instance of $G1$. Every other terminal is translated similarly to the right.

Note that since $d \ge 1$, we add no terminal at distance less than $\frac{1}{4}$ from $[w,x]$ in any instance of $G1$.

Informally, in this gadget, if in one instance of $G1$ we go around the segment $[w,x]$ via the top side, then in the next gadget, we will go around the segment $[w,x]$ via the bottom side, and vice versa. We want the path from $u$ to $v$ either to be exactly the segment $[u,x]$ followed by $[x,v]$, or the segment $[u,w]$ followed by $[w,v]$.
To ensure this, we further require that no terminal on the segment $[u,v]$ is closer to $u$ or $v$ than $\frac{1}{8}$, and that no other terminal is closer to the two possible paths from $u$ to $v$ described above than $\frac{1}{8}$.

If the path from $u$ to $v$ does not behave as stated above, we lose at least some positive constant. Therefore, as in the previous analysis, we can take $c_3$ and $c_4$ to be lower than this constant, and make sure that every path of level $2$ behaves as we want.

%In the following, for a given gadget $G1$, there may be at most one pair of terminals whose shortest path intersects a line in $G1$, and it will intersect at most two copies of $G1$. 
%Therefore we may loose at most a constant distance by going around these gadgets, so if we take enough copies (a constant number) of the variable gadget at the same coordinates, the variable gadget will behave as described above.

In other words, in a variable gadget, either in the first copy of $G1$, we go around the segment $[w,x]$ via the top side, or we go around this segment via the bottom side. In the first case, we will say that the corresponding variable is positive, and in the second case, we will say that it is negative. The rest of the variable gadget is completely defined by whether the variable is positive of negative (we alternate for each gadget $G1$). Note that both cases, the variable being positive or negative, yield the same path lengths for the variable gadget. See Fig.~\ref{Gvar}.

We say that a gadget $G1$ is \emph{positive} if we go around $[w,x]$ on the same side as the first one of the gadget, and \emph{negative} otherwise. For a variable gadget, the copies of $G1$ will thus alternate between positive and negative.

In variable gadgets, we can modify the value of $d$ for each copy of $G1$ independently, thus increasing or decreasing its length by at most $\frac{1}{2}$.
By doing so, over eight copies of $G1$, one may move the following copy of $G1$
by up to four either to the right or to the left, and go back to the normal positions of the copies of $G1$ eight copies later. This will enable us to assume that some specific vertical positions will correspond to any arbitrary position according to a copy of $G1$, as long as they have at least eighteen copies of $G1$ between them. We can even choose if the corresponding copy of $G1$ will be positive or negative.

\paragraph*{Gadget $G2$:} 

The gadget $G2$ is very similar to the gadget $G1$, except that it is turned by 90 degrees, and the levels are increased by 1. It consists of a pair of level 3 terminals $u$ and $v$, placed vertically at a distance $1 \le d \le 2$, say at coordinates $(0,0)$ and $(0,d)$ respectively; and a pair of level 2 terminals $w$ and $x$ at distance $1$ left and right the axis $(u,v)$ at mid-distance, say $w$ at coordinates $(-1, \frac{d}{2})$ and $x$ at coordinates $(1, \frac{d}{2})$.

This gadget behaves very similarly to the gadget $G1$. Just like gadget $G1$, we will combine several of them together, and make (vertical) chains. Informally, it will be used to transmit vertically the information of whether a given variable is true or false.

\paragraph*{Fork Gadget (Fig.~\ref{Gfork})}
\begin{figure}[h!]
    \begin{center}
      \begin{tikzpicture}[scale = 1.5]
	\def\n{2}

	\foreach\x in {-1,0}{ 	\coordinate (u\x) at (1.75*\x,0);
        						\coordinate (x\x) at (1.75*\x+1,1);
        						\coordinate (w\x) at (1.75*\x+1,-1);
        						\coordinate (v\x) at (1.75*\x+2,0);

        						\draw[blue, very thick] (w\x) -- (x\x);
						\pgfmathparse{int mod(\x,2)}
  						\ifnum\pgfmathresult=0{
                                \draw[red,thick] (w\x) -- (v\x);
                                \draw[red,thick] (w\x) -- (u\x);
                                \draw[red,dashed,thick] (x\x) -- (v\x);
                                \draw[red,dashed,thick] (x\x) -- (u\x);}
						\else{
                                \draw[red,dashed,thick] (w\x) -- (v\x);
                                \draw[red,dashed,thick] (w\x) -- (u\x);
                                \draw[red,thick] (x\x) -- (v\x);
                                \draw[red,thick] (x\x) -- (u\x);
						}
						\fi}

	\coordinate (y) at (0.125,0);
	\coordinate (z) at (0.125,-2);

        \coordinate (u) at (0.125,-1.75);
        \coordinate (x) at (1.125,-2.75);
        \coordinate (w) at (-0.875,-2.75);
        \coordinate (v) at (0.125,-3.75);

    \draw[black!30!green,thick] (y) -- (u0) -- (z);
    \draw[black!30!green,dashed,thick] (y) -- (v-1) -- (z);

        \draw[red, very thick] (w) -- (x);
        \draw[black!30!green,dashed,thick] (w) -- (v);
        \draw[black!30!green,dashed,thick] (w) -- (u);
        \draw[black!30!green,thick] (x) -- (v);
        \draw[black!30!green,thick] (x) -- (u);

      \end{tikzpicture}
\end{center}
\caption{Fork gadget. The previous gadget $G1$ is also represented.
The paths in blue have level 1, the paths in red have level 2, and the paths in green have level 3. The paths that are not thick either all need to follow the plain lines, or all need to follow the dashed lines. \label{Gfork}}
\end{figure}
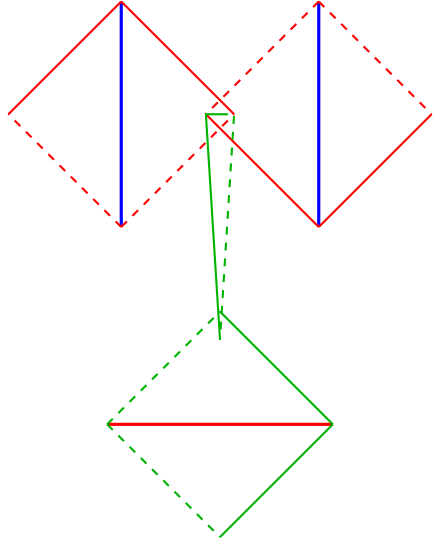

Another gadget we need is the fork gadget. This gadget is based on a $G1$ gadget that is part of a variable gadget, and adds to it\importantquestion{KF: What does it add?}. We will consider the coordinates according to that $G1$ gadget. We will further assume that the $G1$ gadget is not the first of its clause (therefore it has the point $v$ from the previous gadget at coordinates $(\frac{1}{4}, 0)$, that we will call $v'$).
Let $y$ at coordinates $(\frac{1}{8},0)$, and $z$ at coordinates $(\frac{1}{8}, -2)$, be a pair of level 3 terminals. Note that $y$ is at distance $\frac{1}{8}$ from vertex $u$ and from $v'$, and that $z$ is far enough away from any terminal not to make conflicts with the level system. 

As a level 3 path, the path $yz$ will need to go around the level 2 paths of the variable gadget. It will therefore need either to go through vertex $v'$ (in case our gadget $G1$ is positive) or through vertex $u$ (in case it is negative). It will branch with a gadget $G2$, whose vertex $u$ will be placed at coordinates $(\frac{1}{8}, -\frac{7}{4})$. Informally, it transmits the information whether the variable gadget is true or false from the horizontal variable gadget to a vertical chain of $G2$ gadgets. See Fig.~\ref{Gfork} for an illustration of this gadget.

We will make sure that no other terminal is close enough to interfere, and thus that any change in the level 3 paths would result in more loss that whatever can be gained from the level 4 paths.
Therefore, as previously, we can assume that the paths of level 3 will behave as mentioned in the previous paragraph.

\paragraph*{Crossing Gadget (Fig.~\ref{Gcross}):}
\begin{figure}[h!]
    \begin{center}
      \begin{tikzpicture}[scale = 1.5]

	\foreach\x in {-1,0,1}{ 	\coordinate (u\x) at (1.75*\x,0);
        						\coordinate (x\x) at (1.75*\x+1,1);
        						\coordinate (w\x) at (1.75*\x+1,-1);
        						\coordinate (v\x) at (1.75*\x+2,0);

        						\draw[blue, very thick] (w\x) -- (x\x);
						\pgfmathparse{int mod(\x,2)}
  						\ifnum\pgfmathresult=0{
                                \draw[red,thick] (w\x) -- (v\x);
                            \draw[red,thick] (w\x) -- (u\x);}
						\else{
                                \draw[red,thick] (x\x) -- (v\x);
                                \draw[red,thick] (x\x) -- (u\x);
						}
						\fi}

	\coordinate (y) at (1,2);
	\coordinate (z) at (1,-2);

    \draw[black!30!green,thick] (y) -- (v-1)--(u0) -- (z);
    \draw[black!30!green, dashed,thick] (y) -- (u1) -- (v0) -- (z);

      \end{tikzpicture}
\end{center}
\caption{Crossing gadget. The $G1$ gadgets before and after the current one are also represented.
The paths in blue have level 1, the paths in red have level 2, and the paths in green have level 3. 
Only one of the two possibilities is represented for the variable gadget, the other one is the symmetric one. The four possible green paths have the same length.
 \label{Gcross}}
\end{figure}
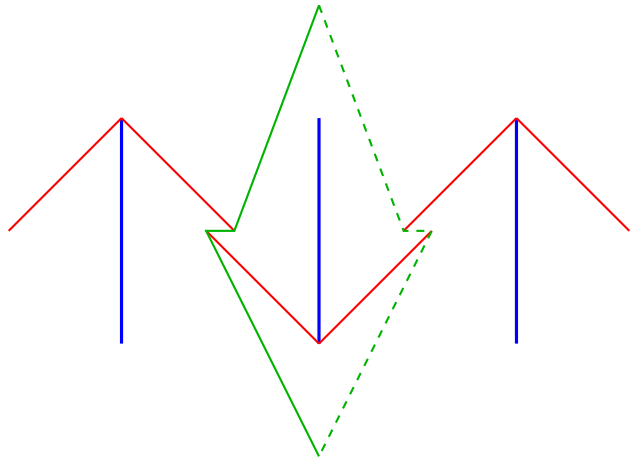

Our vertical chains of $G2$ gadgets need to be able to cross our horizontal chains of $G1$ gadgets. To do this, we will use the crossing gadget. It essentially behaves as a $G2$ gadget that uses a level~2 path from a $G1$ gadget that is part of a variable gadget as its middle path $wx$.

More formally, from a gadget $G1$ that is part of a variable gadget, we add a pair of level 3 vertices at coordinates $(\frac{d}{2}, 2)$ and $(\frac{d}{2}, -2)$. We then branch them to $G2$ gadgets normally (one with its vertex $u$ at coordinates $(\frac{d}{2}, -\frac{7}{4})$, and one symmetric above).

The two new vertices are far enough not to interfere with anything, and, as before, we can assume that the level 3 path is one of four paths with the same length, going either left or right of the level 2 path, independently from whether the $G1$ gadget is positive or negative. Therefore this gadget indeed enables our vertical and horizontal chains to intersect without interfering.

\paragraph*{Clause Gadget (Fig.~\ref{Gclo}):}

We are now ready to describe our clause gadgets. We will put them around some
horizontal position, far enough from each other. Let us take the $i$-th clause to
be placed at the middle of the $(10 + 20i)$-th gadget $G1$ of the variable gadgets (before changing the values of $d$). As argued above, this leaves enough space between two clause gadgets to make one coincide with any point in either a positive or a negative instance of $G1$.

Now consider the two variables that appear in the clause, $v_i$ and $v_j$. 
Let us assume without loss of generality that the gadget corresponding to $v_i$
is above our position. 
Let us modify the values of $d$ for the variable gadget corresponding to $v_i$
(resp. $v_j$) such that the vertical position corresponds to $\frac{1}{8}$ for
the coordinate of some gadget $G1$, that is positive if $v_i$(resp. $v_j$)
appears positively in the clause, and negatively otherwise. Now this vertical position is correct to have a fork gadget with the variable gadgets of $v_i$ and $v_j$ (the one in the bottom being the symmetric of the one presented above), and to have a crossing gadget with every other variable gadget. 

Now as we mentioned in the part about gadget $G2$, we chain the $G2$ gadgets similarly to how we chain the gadgets $G1$ to get a variable gadgets. Also similarly to the variable gadgets, by modulating the values of $d$ and taking the variable gadgets sufficiently far apart, we can make our $G2$ gadgets coincide with the other gadgets. 

The last part is to place, between two variable gadgets, for instance right below the gadget for variable $v_i$, the clause gadget, that will cost more if the clause is not verified by the solution. For this clause gadget, we take a pair of level 4 terminals $u$ and $v$ at coordinates $(0,0)$ and $(0,-2)$ respectively. 
Between them we put not one but two pairs of level 2 terminals, at vertical positions $-\frac{3}{4}$ and $-\frac{5}{4}$, respectively. For the top pair, we put the left vertex at coordinates $(-1,-\frac{3}{4})$, and the right one at coordinates $(\alpha, -\frac{3}{4})$. For the bottom pair, we put the left vertex at coordinates $(-\alpha,-\frac{5}{4})$, and the right one at coordinates $(1,-\frac{5}{4})$. We do this for a value of $0 <\alpha<1$ such that the three black paths depicted in Fig.~\ref{Gclo} have the same length. This value of $\alpha$ exists by continuity.

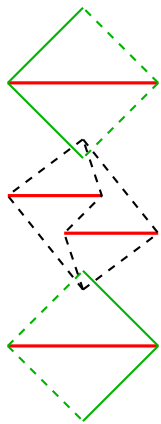
\begin{figure}[h!]
    \begin{center}
      \begin{tikzpicture}

	\foreach\x in {-1,1}{ 	\coordinate (u\x) at (0,1.75*\x);
        						\coordinate (x\x) at (1,1.75*\x+1);
        						\coordinate (w\x) at (-1,1.75*\x+1);
        						\coordinate (v\x) at (0,1.75*\x+2);

        						\draw[red, very thick] (w\x) -- (x\x);
  						\ifnum\x=1{
                                \draw[black!30!green,thick] (w\x) -- (v\x);
                                \draw[black!30!green,thick] (w\x) -- (u\x);
                                \draw[black!30!green, dashed,thick] (x\x) -- (v\x);
                            \draw[black!30!green, dashed,thick] (x\x) -- (u\x);}
						\else{
                                \draw[black!30!green, dashed,thick] (w\x) -- (v\x);
                                \draw[black!30!green, dashed,thick] (w\x) -- (u\x);
                                \draw[black!30!green,thick] (x\x) -- (v\x);
                                \draw[black!30!green,thick] (x\x) -- (u\x);
						}
						\fi}

	\coordinate (u) at (0,0);
        	\coordinate (x) at (1,0.75);
        	\coordinate (w) at (-0.25,0.75);
        	\coordinate (y) at (-1,1.25);
        	\coordinate (z) at (0.25,1.25);
        	\coordinate (v) at (0,2);

        	\draw[red, very thick] (w) -- (x);
        	\draw[red, very thick] (y) -- (z);
            \draw [dashed,thick] (u) -- (x) -- (v);
            \draw [dashed,thick] (u) -- (y) -- (v);
            \draw [dashed,thick] (u) -- (w) -- (z) -- (v);

      \end{tikzpicture}
\end{center}
\caption{The clause gadget. Level 2 paths are red, level 3 paths are green, and level 4 paths are dashed black. The three black paths have the same length, and at least one of them is possible unless both of green paths use dashed lines, which corresponds to the two literals being false in the clause.
\label{Gclo}}
\end{figure}

Now the three black paths have the same length, and one of them is possible unless both of the literals that appear in the clause are false. Note that we can suppose that we altered the value of $d$ for $G2$ gadgets until the $G2$ gadgets right above and below our clause gadget are in the right direction, depending on whether the literals appear positively or negatively in the clause.

The level 4 terminals are far enough from everything else not to interfere with
our level reasoning. By this reasoning, the level 4 paths cannot get enough
length to counterbalance any change to other paths. Moreover, the descriptions
above enable us to get an optimal solution for terminal up to level 3 from any
truth assignment of the literals. Then each pair of terminal of level 4, that
corresponds to a clause, will have the same length unless the clause is not
verified by the assignment, in which case it will have to go around an
additional level 3 path, and lose a constant value per non-verified clause.
Therefore our instance of \MSP{} is equivalent to the instance of {\sc Max 2-SAT}.

% TODO Uncomment in ARXIV version
% \paragraph*{Acknowledgments.}
% We are grateful to S{\'{a}}ndor Fekete for useful discussions.

\bibliographystyle{abbrv}
\bibliography{bib}

\begin{thebibliography}{10}

\bibitem{abrahamsen2020geometric}
M.~Abrahamsen, P.~Giannopoulos, M.~L{\"o}ffler, and G.~Rote.
\newblock Geometric multicut: Shortest fences for separating groups of objects
  in the plane.
\newblock {\em Discrete \& Computational Geometry}, 64(3):575--607, 2020.

\bibitem{ARRC11}
B.~Alper, N.~H. Riche, G.~Ramos, and M.~Czerwinski.
\newblock Design study of {LineSets}, a novel set visualization technique.
\newblock {\em IEEE Trans. Vis. Comput. Graphics}, 17(12):2259--2267, 2011.

\bibitem{Arora98}
S.~Arora.
\newblock {P}olynomial time approximation schemes for {E}uclidean traveling
  salesman and other geometric problems.
\newblock {\em J. {ACM}}, 45(5):753--782, 1998.

\bibitem{Arora2007}
S.~Arora.
\newblock {\em Approximation Algorithms for Geometric TSP}, pages 207--221.
\newblock Springer US, Boston, MA, 2007.

\bibitem{red-blue}
S.~Arora and K.~L. Chang.
\newblock Approximation schemes for degree-restricted {MST} and red-blue
  separation problems.
\newblock {\em Algorithmica}, 40(3):189--210, 2004.

\bibitem{bastert1998geometric}
O.~Bastert and S.~P. Fekete.
\newblock Geometric wire routing.
\newblock Technical report, Technical Report 332, Zentrum f{\"u}r Angewandte
  Informatik, 1998.

\bibitem{isaac15}
S.~Bereg, K.~Fleszar, P.~Kindermann, S.~Pupyrev, J.~Spoerhase, and A.~Wolff.
\newblock Colored non-crossing euclidean steiner forest.
\newblock In K.~M. Elbassioni and K.~Makino, editors, {\em Algorithms and
  Computation - 26th International Symposium, {ISAAC} 2015, Nagoya, Japan,
  December 9-11, 2015, Proceedings}, volume 9472 of {\em Lecture Notes in
  Computer Science}, pages 429--441. Springer, 2015.

\bibitem{JGAA-499}
T.~{Castermans}, M.~{van Garderen}, W.~{Meulemans}, M.~{Nöllenburg}, and
  X.~{Yuan}.
\newblock Short plane supports for spatial hypergraphs.
\newblock {\em Journal of Graph Algorithms and Applications}, 23(3):463--498,
  2019.

\bibitem{CHKL13}
T.~M. Chan, H.-F. Hoffmann, S.~Kiazyk, and A.~Lubiw.
\newblock Minimum length embedding of planar graphs at fixed vertex locations.
\newblock In S.~Wismath and A.~Wolff, editors, {\em Proc. 21st Int. Symp. Graph
  Drawing (GD'13)}, volume 8242, pages 376--387, 2013.

\bibitem{covid}
M.~Chertkov, R.~Abrams, A.~M.~E. Sikaroudi, M.~Krechetov, C.~N. Slagle,
  A.~Efrat, R.~Fulek, and E.~Oren.
\newblock Graphical models of pandemic.
\newblock {\em medRxiv}, 2021.

\bibitem{segmentation}
M.~C. Cooper.
\newblock The tractability of segmentation and scene analysis.
\newblock {\em International Journal of Computer Vision}, 30(1):27--42, 1998.

\bibitem{np-red-blue}
P.~Eades and D.~Rappaport.
\newblock The complexity of computing minimum separating polygons.
\newblock {\em Pattern Recognit. Lett.}, 14(9):715--718, 1993.

\bibitem{EHKP15}
A.~{Efrat}, Y.~{Hu}, S.~{Kobourov}, and S.~{Pupyrev}.
\newblock Mapsets: Visualizing embedded and clustered graphs.
\newblock {\em J. Graph Algorithms and Applications}, 19(2):571--593, 2015.

\bibitem{EN11}
J.~Erickson and A.~Nayyeri.
\newblock Shortest non-crossing walks in the plane.
\newblock In {\em Proc. 22nd ACM-SIAM Symp. Discrete Algorithms (SODA'11)},
  pages 297--308, 2011.

\bibitem{PrivFekete}
S.~Fekete.
\newblock Personal communication, Feb. 2022.

\bibitem{steiner-tree-np-hard}
M.~R. Garey and D.~S. Johnson.
\newblock The rectilinear steiner tree problem in {NP} complete.
\newblock {\em {SIAM} Journal of Applied Mathematics}, 32:826--834, 1977.

\bibitem{goethem2017painter}
A.~v. Goethem, I.~Kostitsyna, M.~v. Kreveld, W.~Meulemans, M.~Sondag, and
  J.~Wulms.
\newblock The painter’s problem: covering a grid with colored connected
  polygons.
\newblock In {\em International Symposium on Graph Drawing and Network
  Visualization}, pages 492--505. Springer, 2017.

\bibitem{papa95}
M.~Grigni, E.~Koutsoupias, and C.~H. Papadimitriou.
\newblock {A}n {A}pproximation {S}cheme for {P}lanar {G}raph {TSP}.
\newblock In {\em 36th Annual Symposium on Foundations of Computer Science,
  Milwaukee, Wisconsin, USA, 23-25 October 1995}, pages 640--645. {IEEE}
  Computer Society, 1995.

\bibitem{gudmundsson99RedBlue}
J.~Gudmundsson and C.~Levcopoulos.
\newblock {A} {F}ast {A}pproximation {A}lgorithm for {TSP} with {N}eighborhoods
  and {R}ed-{B}lue {S}eparation.
\newblock In T.~Asano, H.~Imai, D.~T. Lee, S.-i. Nakano, and T.~Tokuyama,
  editors, {\em Computing and Combinatorics}, pages 473--482, Berlin,
  Heidelberg, 1999. Springer Berlin Heidelberg.

\bibitem{HKKLS}
F.~Hurtado, M.~Korman, M.~J. van Kreveld, M.~L{\"o}ffler, V.~Sacrist{\'a}n,
  A.~Shioura, R.~I. Silveira, B.~Speckmann, and T.~Tokuyama.
\newblock Colored spanning graphs for set visualization.
\newblock {\em Comput. Geom.}, 68:262--276, 2018.
\newblock Special issue in memory of {F}erran {H}urtado.

\bibitem{kindermann2018partition}
P.~Kindermann, B.~Klemz, I.~Rutter, P.~Schnider, and A.~Schulz.
\newblock The partition spanning forest problem, 2018.

\bibitem{focs21}
S.~Kisfaludi{-}Bak, J.~Nederlof, and K.~Wegrzycki.
\newblock {A} {G}ap-{ETH}-{T}ight {A}pproximation {S}cheme for {E}uclidean
  {TSP}.
\newblock In {\em 2021 IEEE 62nd Annual Symposium on Foundations of Computer
  Science (FOCS)}, pages 351--362. IEEE, 2022.

\bibitem{kostitsyna2017SteinerArborescences}
I.~Kostitsyna, B.~Speckmann, and K.~Verbeek.
\newblock Non-crossing geometric steiner arborescences.
\newblock In Y.~Okamoto and T.~Tokuyama, editors, {\em Proc. 28th Int. Symp.
  Algorithms \& Computation (ISAAC'17)}, volume~92 of {\em LIPIcs}, pages
  54:1--54:13. Schloss Dagstuhl~-- Leibniz-Zentrum f{\"u}r Informatik, 2017.

\bibitem{KMN01}
Y.~Kusakari, D.~Masubuchi, and T.~Nishizeki.
\newblock Finding a noncrossing {S}teiner forest in plane graphs under a 2-face
  condition.
\newblock {\em J. Comb. Optim.}, 5(2):249--266, 2001.

\bibitem{LMMPS95}
T.~M. Liebling, F.~Margot, D.~M{\"u}ller, A.~Prodon, and L.~Stauffer.
\newblock Disjoint paths in the plane.
\newblock {\em ORSA J. Comput.}, 7(1):84--88, 1995.

\bibitem{planar-sep}
R.~J. Lipton and R.~E. Tarjan.
\newblock Applications of a planar separator theorem.
\newblock {\em {SIAM} J. Comput.}, 9(3):615--627, 1980.

\bibitem{mata-michell}
C.~S. Mata and J.~S.~B. Mitchell.
\newblock Approximation algorithms for geometric tour and network design
  problems (extended abstract).
\newblock In J.~Snoeyink, editor, {\em Proceedings of the Eleventh Annual
  Symposium on Computational Geometry, Vancouver, B.C., Canada, June 5-12,
  1995}, pages 360--369. {ACM}, 1995.

\bibitem{geomspannet}
G.~Narasimhan and M.~H.~M. Smid.
\newblock {\em Geometric Spanner Networks}.
\newblock Cambridge University Press, 2007.

\bibitem{Papadimitriou77}
C.~H. Papadimitriou.
\newblock The {E}uclidean {T}raveling {S}alesman {P}roblem is {NP}-{C}omplete.
\newblock {\em Theoretical Computer Science}, 4(3):237--244, 1977.

\bibitem{complexity-book}
C.~H. Papadimitriou.
\newblock {\em Computational complexity}.
\newblock Academic Internet Publ., 2007.

\bibitem{Pap99}
E.~Papadopoulou.
\newblock {$k$}-{P}airs non-crossing shortest paths in a simple polygon.
\newblock {\em Int. J. Comput. Geom. Appl.}, 9(6):533--552, 1999.

\bibitem{noncrossing-paths-mitchell}
V.~Polishchuk and J.~S.~B. Mitchell.
\newblock Thick non-crossing paths and minimum-cost flows in polygonal domains.
\newblock In J.~Erickson, editor, {\em Proceedings of the 23rd {ACM} Symposium
  on Computational Geometry, Gyeongju, South Korea, June 6-8, 2007}, pages
  56--65. {ACM}, 2007.

\bibitem{PM07}
V.~Polishchuk and J.~S.~B. Mitchell.
\newblock Thick non-crossing paths and minimum-cost flows in polygonal domains.
\newblock In {\em Proc. 23rd ACM Symp. Comput. Geom. (SoCG'07)}, pages 56--65,
  2007.

\bibitem{RaoS98}
S.~Rao and W.~D. Smith.
\newblock {A}pproximating {G}eometrical {G}raphs via "{S}panners" and
  "{B}anyans".
\newblock In {\em Proceedings of the Thirtieth Annual {ACM} Symposium on the
  Theory of Computing ({STOC} 1998)}, pages 540--550. {ACM}, 1998.

\bibitem{reinbacher2008delineating}
I.~Reinbacher, M.~Benkert, M.~van Kreveld, J.~S. Mitchell, J.~Snoeyink, and
  A.~Wolff.
\newblock Delineating boundaries for imprecise regions.
\newblock {\em Algorithmica}, 50(3):386--414, 2008.

\bibitem{retsinas2016imageSegmentation}
G.~Retsinas, G.~Louloudis, N.~Stamatopoulos, and B.~Gatos.
\newblock Efficient document image segmentation representation by approximating
  minimum-link polygons.
\newblock In {\em 2016 12th IAPR Workshop on Document Analysis Systems (DAS)},
  pages 293--298. IEEE, 2016.

\bibitem{takahashi1992planeGraphsPaths}
J.~Takahashi, H.~Suzuki, and T.~Nishizeki.
\newblock Algorithms for finding non-crossing paths with minimum total length
  in plane graphs.
\newblock In {\em International Symposium on Algorithms and Computation}, pages
  400--409. Springer, 1992.

\bibitem{takahashi1993rectPaths}
J.~Takahashi, H.~Suzuki, and T.~Nishizeki.
\newblock Finding shortest non-crossing rectilinear paths in plane regions.
\newblock In {\em International Symposium on Algorithms and Computation}, pages
  98--107. Springer, 1993.

\bibitem{toussaint1986optimal}
G.~T. Toussaint.
\newblock {\em An optimal algorithm for computing the relative convex hull of a
  set of points in a polygon}.
\newblock MacGill University. School of Computer Science, 1986.

\end{thebibliography}

\end{document}